\documentclass{article}
\usepackage{url}
\usepackage{color}
\usepackage{bm}
\usepackage{amsmath}
\usepackage{amssymb}\usepackage{amsthm}
\usepackage{amsfonts}
\usepackage{amscd}
\begin{document}
\theoremstyle{plain}
\newtheorem*{ithm}{Theorem}
\newtheorem*{idefn}{Definition}
\newtheorem{thm}{Theorem}[section]
\newtheorem{lem}[thm]{Lemma}
\newtheorem{dlem}[thm]{Lemma/Definition}
\newtheorem{prop}[thm]{Proposition}
\newtheorem{set}[thm]{Setting}
\newtheorem{cor}[thm]{Corollary}
\newtheorem*{icor}{Corollary}
\theoremstyle{definition}
\newtheorem{assum}[thm]{Assumption}
\newtheorem{notation}[thm]{Notation}
\newtheorem{defn}[thm]{Definition}
\newtheorem{clm}[thm]{Claim}
\newtheorem{ex}[thm]{Example}
\theoremstyle{remark}
\newtheorem{rem}[thm]{Remark}
\numberwithin{equation}{section}

\newcommand{\unit}{\mathbb I}
\newcommand{\ali}[1]{{\mathfrak A}_{[ #1 ,\infty)}}
\newcommand{\alm}[1]{{\mathfrak A}_{(-\infty, #1 ]}}
\newcommand{\nn}[1]{\lV #1 \rV}
\newcommand{\br}{{\mathbb R}}
\newcommand{\dm}{{\rm dom}\mu}
\newcommand{\inn}{({\rm {inner}})}
\newcommand{\Ad}{\mathop{\mathrm{Ad}}\nolimits}
\newcommand{\Proj}{\mathop{\mathrm{Proj}}\nolimits}
\newcommand{\RRe}{\mathop{\mathrm{Re}}\nolimits}
\newcommand{\RIm}{\mathop{\mathrm{Im}}\nolimits}
\newcommand{\Wo}{\mathop{\mathrm{Wo}}\nolimits}
\newcommand{\Prim}{\mathop{\mathrm{Prim}_1}\nolimits}
\newcommand{\Primz}{\mathop{\mathrm{Prim}}\nolimits}
\newcommand{\ClassA}{\mathop{\mathrm{ClassA}}\nolimits}
\newcommand{\Class}{\mathop{\mathrm{Class}}\nolimits}
\newcommand{\diam}{\mathop{\mathrm{diam}}\nolimits}
\def\qed{{\unskip\nobreak\hfil\penalty50
\hskip2em\hbox{}\nobreak\hfil$\square$
\parfillskip=0pt \finalhyphendemerits=0\par}\medskip}
\def\proof{\trivlist \item[\hskip \labelsep{\bf Proof.\ }]}
\def\endproof{\null\hfill\qed\endtrivlist\noindent}
\def\proofof[#1]{\trivlist \item[\hskip \labelsep{\bf Proof of #1.\ }]}
\def\endproofof{\null\hfill\qed\endtrivlist\noindent}

\newcommand{\sym}{\mathop{\mathrm{Sym}}\nolimits}
\newcommand{\co}{\mathop{\mathrm{co}}\nolimits}
\newcommand{\Dia}{\mathop{\mathrm{D}}\nolimits}
\newcommand{\Hom}{\mathop{\mathrm{Hom}}\nolimits}
\newcommand{\sgn}{\mathop{\mathrm{sgn}}\nolimits}

\def\qed{{\unskip\nobreak\hfil\penalty50
\hskip2em\hbox{}\nobreak\hfil$\square$
\parfillskip=0pt \finalhyphendemerits=0\par}\medskip}
\def\proof{\trivlist \item[\hskip \labelsep{\bf Proof.\ }]}
\def\endproof{\null\hfill\qed\endtrivlist\noindent}
\def\proofof[#1]{\trivlist \item[\hskip \labelsep{\bf Proof of #1.\ }]}
\def\endproofof{\null\hfill\qed\endtrivlist\noindent}
\newcommand{\ind}{\mathop{\mathrm{ind}}\nolimits}
\newcommand{\varphii}{\varphi}
\newcommand{\pgs}{\caP_{\sigma}}
\newcommand{\oo}{{\boldsymbol\varphii}}
\newcommand{\caA}{{\mathcal A}}
\newcommand{\caB}{{\mathcal B}}
\newcommand{\caC}{{\mathcal C}}
\newcommand{\caD}{{\mathcal D}}
\newcommand{\caE}{{\mathcal E}}
\newcommand{\caF}{{\mathcal F}}
\newcommand{\caG}{{\mathcal G}}
\newcommand{\caH}{{\mathcal H}}
\newcommand{\caI}{{\mathcal I}}
\newcommand{\caJ}{{\mathcal J}}
\newcommand{\caK}{{\mathcal K}}
\newcommand{\caL}{{\mathcal L}}
\newcommand{\caM}{{\mathcal M}}
\newcommand{\caN}{{\mathcal N}}
\newcommand{\caO}{{\mathcal O}}
\newcommand{\caP}{{\mathcal P}}
\newcommand{\caQ}{{\mathcal Q}}
\newcommand{\caR}{{\mathcal R}}
\newcommand{\caS}{{\mathcal S}}
\newcommand{\caT}{{\mathcal T}}
\newcommand{\caU}{{\mathcal U}}
\newcommand{\caV}{{\mathcal V}}
\newcommand{\caW}{{\mathcal W}}
\newcommand{\caX}{{\mathcal X}}
\newcommand{\caY}{{\mathcal Y}}
\newcommand{\caZ}{{\mathcal Z}}
\newcommand{\bbA}{{\mathbb A}}
\newcommand{\bbB}{{\mathbb B}}
\newcommand{\bbC}{{\mathbb C}}
\newcommand{\bbD}{{\mathbb D}}
\newcommand{\bbE}{{\mathbb E}}
\newcommand{\bbF}{{\mathbb F}}
\newcommand{\bbG}{{\mathbb G}}
\newcommand{\bbH}{{\mathbb H}}
\newcommand{\bbI}{{\mathbb I}}
\newcommand{\bbJ}{{\mathbb J}}
\newcommand{\bbK}{{\mathbb K}}
\newcommand{\bbL}{{\mathbb L}}
\newcommand{\bbM}{{\mathbb M}}
\newcommand{\bbN}{{\mathbb N}}
\newcommand{\bbO}{{\mathbb O}}
\newcommand{\bbP}{{\mathbb P}}
\newcommand{\bbQ}{{\mathbb Q}}
\newcommand{\bbR}{{\mathbb R}}
\newcommand{\bbS}{{\mathbb S}}
\newcommand{\bbT}{{\mathbb T}}
\newcommand{\bbU}{{\mathbb U}}
\newcommand{\bbV}{{\mathbb V}}
\newcommand{\bbW}{{\mathbb W}}
\newcommand{\bbX}{{\mathbb X}}
\newcommand{\bbY}{{\mathbb Y}}
\newcommand{\bbZ}{{\mathbb Z}}
\newcommand{\str}{^*}
\newcommand{\lv}{\left \vert}
\newcommand{\rv}{\right \vert}
\newcommand{\lV}{\left \Vert}
\newcommand{\rV}{\right \Vert}
\newcommand{\la}{\left \langle}
\newcommand{\ra}{\right \rangle}
\newcommand{\ltm}{\left \{}
\newcommand{\rtm}{\right \}}
\newcommand{\lcm}{\left [}
\newcommand{\rcm}{\right ]}
\newcommand{\ket}[1]{\lv #1 \ra}
\newcommand{\bra}[1]{\la #1 \rv}
\newcommand{\lmk}{\left (}
\newcommand{\rmk}{\right )}
\newcommand{\al}{{\mathcal A}}
\newcommand{\md}{M_d({\mathbb C})}
\newcommand{\ainn}{\mathop{\mathrm{AInn}}\nolimits}
\newcommand{\id}{\mathop{\mathrm{id}}\nolimits}
\newcommand{\Tr}{\mathop{\mathrm{Tr}}\nolimits}
\newcommand{\Ran}{\mathop{\mathrm{Ran}}\nolimits}
\newcommand{\Ker}{\mathop{\mathrm{Ker}}\nolimits}
\newcommand{\Aut}{\mathop{\mathrm{Aut}}\nolimits}
\newcommand{\spn}{\mathop{\mathrm{span}}\nolimits}
\newcommand{\Mat}{\mathop{\mathrm{M}}\nolimits}
\newcommand{\UT}{\mathop{\mathrm{UT}}\nolimits}
\newcommand{\DT}{\mathop{\mathrm{DT}}\nolimits}
\newcommand{\GL}{\mathop{\mathrm{GL}}\nolimits}
\newcommand{\spa}{\mathop{\mathrm{span}}\nolimits}
\newcommand{\supp}{\mathop{\mathrm{supp}}\nolimits}
\newcommand{\rank}{\mathop{\mathrm{rank}}\nolimits}
\newcommand{\idd}{\mathop{\mathrm{id}}\nolimits}
\newcommand{\ran}{\mathop{\mathrm{Ran}}\nolimits}
\newcommand{\dr}{ \mathop{\mathrm{d}_{{\mathbb R}^k}}\nolimits} 
\newcommand{\dc}{ \mathop{\mathrm{d}_{\cc}}\nolimits} \newcommand{\drr}{ \mathop{\mathrm{d}_{\rr}}\nolimits} 
\newcommand{\zin}{\mathbb{Z}}
\newcommand{\rr}{\mathbb{R}}
\newcommand{\cc}{\mathbb{C}}
\newcommand{\ww}{\mathbb{W}}
\newcommand{\nan}{\mathbb{N}}\newcommand{\bb}{\mathbb{B}}
\newcommand{\aaa}{\mathbb{A}}\newcommand{\ee}{\mathbb{E}}
\newcommand{\pp}{\mathbb{P}}
\newcommand{\wks}{\mathop{\mathrm{wk^*-}}\nolimits}
\newcommand{\mk}{{\Mat_k}}
\newcommand{\mnz}{\Mat_{n_0}}
\newcommand{\mn}{\Mat_{n}}
\newcommand{\dist}{\dc}
\newcommand{\braket}[2]{\left\langle#1,#2\right\rangle}
\newcommand{\ketbra}[2]{\left\vert #1\right \rangle \left\langle #2\right\vert}
\newcommand{\abs}[1]{\left\vert#1\right\vert}
\newtheorem{nota}{Notation}[section]
\def\qed{{\unskip\nobreak\hfil\penalty50
\hskip2em\hbox{}\nobreak\hfil$\square$
\parfillskip=0pt \finalhyphendemerits=0\par}\medskip}
\def\proof{\trivlist \item[\hskip \labelsep{\bf Proof.\ }]}
\def\endproof{\null\hfill\qed\endtrivlist\noindent}
\def\proofof[#1]{\trivlist \item[\hskip \labelsep{\bf Proof of #1.\ }]}
\def\endproofof{\null\hfill\qed\endtrivlist\noindent}
\newcommand{\lal}{{\boldsymbol\lambda}}

\newcommand{\ZZ}{\bbZ_2\times\bbZ_2}
\newcommand{\SSS}{\mathcal{S}}
\newcommand{\cs}{S}
\newcommand{\ct}{t}
\newcommand{\hS}{S}
\newcommand{\vv}{{\boldsymbol v}}
\newcommand{\ala}{a}
\newcommand{\bet}{b}
\newcommand{\gam}{c}
\newcommand{\alphas}{\alpha}
\newcommand{\alphai}{\alpha^{(\sigma_{1})}}
\newcommand{\alphan}{\alpha^{(\sigma_{2})}}
\newcommand{\betas}{\beta}
\newcommand{\betai}{\beta^{(\sigma_{1})}}
\newcommand{\betan}{\beta^{(\sigma_{2})}}
\newcommand{\alphass}{\alpha^{{(\sigma)}}}
\newcommand{\uu}{V}
\newcommand{\vp}{\varsigma}
\newcommand{\vpr}{R}
\newcommand{\tg}{\tau_{\Gamma}}
\newcommand{\sgg}{\Sigma_{\Gamma}^{(\sigma)}}
\newcommand{\nh}{1}
\newcommand{\rk}{2,a}
\newcommand{\nii}{1,a}
\newcommand{\nhh}{3,a}
\newcommand{\sjt}{2}
\newcommand{\sjtg}{2}
\newcommand{\bcg}{\caB(\caH_{\alpha})\otimes  C^{*}(\Sigma_{\Gamma}^{(\sigma)})}
\newcommand{\Uo}{{\rm U}(1)}
\newcommand{\ua}{\mathfrak{p}}
\newcommand{\mkA}{{\mathfrak A}}
\newcommand{\mkB}{{\mathfrak B}}
\newcommand{\rar}{{\boldsymbol r}}

\title{Classification of  symmetry protected topological phases in quantum spin
 chains}

\author{Yoshiko Ogata \thanks{ Graduate School of Mathematical Sciences
The University of Tokyo, Komaba, Tokyo, 153-8914, Japan
Supported in part by
the Grants-in-Aid for
Scientific Research, JSPS.}}
\maketitle

\begin{abstract}
We consider the classification problem of 
symmetry protected topological (SPT) phases on quantum spin systems.
SPT phases are gapped short-range-entangled quantum phases with a symmetry $G$.
We explain that in one and two-dimensional quantum spin systems, 
there are $H^{2}(G, U(1))$/ $H^{3}(G, U(1))$-valued invariant, confirming
a physicists conjecture.
\end{abstract}

\maketitle

\tableofcontents

\section{Symmetry protected topological phases}
Phase transitions in macroscopic systems have been among central topics of study in theoretical physics.
Traditional examples of phase transitions were those between a disordered phase and an ordered phase.
It is well known that such a phase transition can be characterized by identifying a suitable order parameter, which is nothing but an extensive sum of some local operators.  A notable example is the ferromagnetic phase transition in the Ising model, where the order parameter is the bulk magnetization.
This is the topic of the lecture note by Michael Aizenman in the present volume.

In recent years it has become clear that certain quantum many-body systems exhibit phase transitions between two phases which have no orders.
Such phase transitions cannot be characterized by standard order parameters, and are called ``topological phase transitions'' \cite{Nobel2016}. 
We note  that here the term ``topological'' should not be interpreted in its strict mathematical sense.

In the present lecture we focus on mathematical aspects of the problem of symmetry protected topological phases in quantum spin systems.
The problem has its root in Haldane's discovery (which brought him the 2016 Nobel prize in physics \cite{Nobel2016}) that the antiferromagnetic Heisenberg chain with spin 1 has ``exotic'' low energy properties, essentially different from 
those of its spin $1/2$ counterpart.
The problem has been extensively studied from experimental, numerical, theoretical, and mathematical points of view.
After the introduction by Gu and Wen \cite{GuWen2009} of the notion of symmetry protected topological (SPT) phases, it has become clear that the model studied by Haldane belongs to a nontrivial SPT phase characterized by an index that takes value in the second group cohomology.
The problem of SPT phases has been further developed by physicists, and it has been conjectured that SPT phases in $\nu$ dimensional quantum spin systems are classified by an index that takes value in the $(\nu+1)$-th group cohomology of the relevant symmetry group.

\medskip
Let us be slightly more precise.
A quantum spin system is characterized by a self-adjoint operator called Hamiltonian, which governs the time evolution of the system.
In short SPT phases are equivalence classes of Hamiltonians satisfying certain conditions on the symmetry as well as low energy properties.
Let us consider the set of all Hamiltonians with some fixed symmetry, which have a unique gapped ground state in the bulk.
In two or higher dimensions, in order to exclude models with so called ``topological order'', we further assume that all Hamiltonians in the set can be smoothly
deformed into
a common trivial gapped Hamiltonian without closing the gap.
We say 
two such Hamiltonians are equivalent, if they can be smoothly deformed into
each other, without breaking the symmetry.
We call an equivalence class of this classification, a
symmetry protected topological (SPT) phase.
In this lecture, we consider this problem in operator algebraic framework.
In this framework, what we deal with is not Hamiltonians, but a time evolution
(strongly continuous one parameter group of automorphisms on a $C^{*}$-algebra).
The reason for considering the problem in this framework is because
it enables us to treat infinite system directly.
It is essential  to treat the infinite system directly for analysis of SPT phases for us,
because the invariant of our classification can be most naturally defined on the infinite system.
The situation is quite analogous to that of Fredholm index.
Recall that 
the Fredholm index of an operator $T$ (with
$\dim \ker T, \dim \ker T^*<\infty$ and $T\caH$ closed) on a Hilbert space $\caH$ is defined by
$
\ind T:=\dim \ker T-\dim \ker T^*
$.
By linear algebra, this is always $0$ if $\caH$ is of finite dimensional.
However, it is not the case if the dimension of $\dim\caH$ is infinite.
For example, on $\caH=l^2(\bbN)$, the unilateral shift
\begin{align*}
S\lmk \xi_1,\xi_2,\xi_3,\ldots\rmk
:=\lmk 0, \xi_1,\xi_2,\xi_3,\ldots\rmk,\quad \lmk \xi_1,\xi_2,\xi_3,\ldots\rmk\in l^2(\bbN)
\end{align*}
has $\ker S=\{0\}$, $\ker S^*=\bbC(1,0,0,0,\ldots)$ and 
$
\ind S =0-1=-1
$.
From this, we see that in order to get non-zero Fredholm index, one need to consider infinite dimensional
Hilbert spaces.
Quite analogous phenomena occurs in SPT-theory.

In this section, we introduce the SPT-phases in operator algebraic framework.
\subsection{Finite dimensional quantum mechanics}
In order to motivate us for the operator algebraic framework of quantum statistical mechanics,
we first recall finite dimensional quantum mechanics in this subsection.

In finite-dimensional quantum mechanics,
physical observables are represented by
elements of $\Mat_n$, the algebra of $n\times n$-matrices.
Each positive matrix $\rho$ with $\Tr\rho=1$ (called a density matrix) defines 
a physical state by
\[
\omega_\rho: \Mat_n\ni A\mapsto \Tr\lmk \rho A\rmk\in\bbC.
\]
We call this map $\omega_\rho$ a {sate}.
This corresponds to the
{expectation values} of each {physical observables} $A\in\Mat_n$, in the 
{physical state {$\omega_\rho$}}.
When a state ${\omega_\rho}$ can not be written as a convex combination of two
different states, it is said to be $ pure$.
A state $\omega_\rho$ is pure if and only if $\rho$ is a rank one projection.
{Time evolution} is given by {a self-adjoint matrix $H$},
called {Hamiltonian},
via a formula 
\[
\Mat_n\ni A\mapsto \tau_t(A):=e^{it{H}} A e^{-it{H}},
\quad t\in\bbR.
\]
Let $p$ be the spectral projection of $H$ corresponding to the 
{lowest eigenvalue}.
A state ${\omega_{\rho}}(A):=\Tr{\rho} A$ on $\Mat_n$ is said to be {a ground state} of {$H$}
if the support of $\rho$ is under $p$.
Ground state is {unique} if and only if
$ p$ is a rank one projection, i.e., the {lowest eigenvalue} of {$H$}
is {non-degenerated}.
In this case, the {unique ground state} is $\omega_{ p}(A):=\Tr{p} A$,
and it is {pure} because $p$ is rank one.
Let {$G$} be a finite group and suppose that there is a group action
 {$\beta : G\to \Aut(\Mat_n)$} given by unitaries $V_g$, $g\in G$ 
 \[
 {\beta_g}(A):=\Ad\lmk{ V_g}\rmk\lmk A\rmk,\quad A\in \Mat_d,\quad g\in G.
 \]
 If a Hamiltonian $H$ satisfies
$ {\beta_g}({H})={H}$ for all $g\in G$,
we say
$H$ is $\beta$-invariant.
If a {$\beta$-invariant} Hamiltonian $H$ has a 
{unique ground state} $\omega_{p}(A):=\Tr{p} A$,
then this {unique ground state} $\omega_{p}$  is $\beta$-invariant $\omega_{ p}\lmk \beta_g(A)\rmk=\omega_{p}(A)$, $A\in\Mat_n$, because the spectral projection $p$ is $\beta$-invariant,
i.e., $\beta_g(p)=p$.

\subsection{Basic facts about $C^{*}$-dynamical systems}\label{basiccssec}
In order to extend the finite dimensional quantum mechanics 
to infinite ones, we move to general $C^*$-dynamical systems.
In general $C^*$-dynamical systems, physical observables are
represented by elements of a $C^*$-algebra.
A $C^{*}$-algebra $\mkA$ is a Banach $*$-algebra with norm satisfying
\[
\lV A^{*}A\rV=\lV A\rV^{2},\quad A\in\mkA.
\]
In this lecture, we always consider $C^{*}$-algebras with a unit $\unit$, i.e., a unital $C^{*}$-algebra.
The algebra of $n\times n$ matrices $\Mat_n$ is a $C^*$-algebra with uniform norm
\[
\lV A\rV:=\sup\left\{\lV A\xi\rV \mid \xi\in \bbC^n,\; \lV \xi\rV=1\right\}.
\]
A linear functional $\omega$ over a $C^{*}$-algebra $\mkA$ is defined to be positive if 
\[
\omega(A^{*}A)\ge 0,
\]
for all $A\in\mkA$.
In physical settings, we regard a $C^{*}$-algebra $\mkA$ an algebra of physical local observables and
$\omega(A)$ denotes the expectation value of an observable $A\in\mkA$ in a state $\omega$.
A positive linear functional $\omega$ on $\mkA$ with $\lV \omega\rV=1$ is called a state.
A state $\omega$ over a $C^{*}$-algebra $\mkA$ is defined to be pure 
if it can not be written as a convex combination of two different states.
Recall that physical state in finite quantum mechanics is a map 
$\omega_\rho: \Mat_{n}\to \bbC$ given by a density matrix.
In fact, this $\omega_\rho$ is a positive normalized linear functional on $\Mat_n$, namely it is a state in the above sense.
It is pure if and only if $\rho$ is a one-rank projection.

A {representation} of a $C^*$-algebra $\mkA$ is defined to be a pair $(\caH,\pi)$, where $\caH$
is a Hilbert space and $\pi$ is a $*$-homomorphism of $\mkA$ into $\caB(\caH)$
(the set of all bounded operators on $\caH$.)
Two representations of a $C^*$-algebra $\mkA$, $(\caH_1,\pi_1)$, $(\caH_2,\pi_2)$
are {unitary equivalent} if
there is a unitary $u:\caH_1\to\caH_2$ such that
\begin{align*}
\Ad(u) \circ\pi_1(A)=\pi_2(A),\quad A\in\mkA.
\end{align*}

Given a state, we may associate a corresponding representation.
\begin{thm}
For each state $\omega$ on $\mkA$, there exist a
representation $\pi_{\omega}$ of $\mkA$ on a Hilbert space $\caH_{\omega}$
and a unit vector $\Omega_{\omega}\in\caH_{\omega}$
such that
\begin{align}\label{gns}
\omega(A)=\braket{\Omega_{\omega}}{\pi_{\omega}(A)\Omega_{\omega}},\quad A\in\mkA,\quad\text{and}
\quad
\caH_{\omega}=\overline{\pi_{\omega}(\mkA)\Omega_{\omega}}.
\end{align}
Here, $\overline{\cdot}$ denotes the norm closure.
It is {unique up to unitary equivalence}.
\end{thm}
The latter condition in (\ref{gns}) means that the vector $\Omega_\omega$ is cyclic
for $\pi_\omega(\mkA)$ in $\caH_\omega$.
We use the convention of mathematical physics for inner product : namely, it is linear for the right variable.
The representation $(\caH_{\omega},\pi_{\omega}, \Omega_{\omega})$ is called a 
Gelfand-Naimark-Segal (GNS)-representation/GNS-triple.

An injective $*$-homomorphism of a $C^*$-algebra $\mkA$
onto itself is called a $*$-automorphism of $\mkA$.
We denote by $\Aut(\mkA)$, the set of all $*$-automorphisms on $\mkA$.
For example, for a unitary $V\in\mkA$, 
$\Ad(V)(A):=VAV^*$, $A\in\mkA$
defines a $*$-automorphism on $\mkA$.
Such a $*$-automorphism is called an inner automorphism.
In this lecture, we occasionally write $\alpha=\inn$ when $\alpha=\Ad(V)$
with some unitary $V\in\mkA$.
A $*$-automorphism which is not inner is said to be outer.

For a given representation $(\caH,\pi)$ of a $C^*$-algebra
$\mkA$, we say an automorphism $\alpha$ of $\mkA$
is implementable by a unitary in $(\caH,\pi)$
if there is a unitary $u$ on $\caH$
satisfying
\[
\Ad(u)\lmk \pi(A)\rmk=\pi\circ\alpha(A),\quad
A\in\mkA.
\]
If a state $\omega$ is invariant under some $\alpha\in\Aut(\mkA)$,
then it can be implemented in the GNS-representation by a unitary.
Namely the following theorem holds.
\begin{thm}\label{gnsthm}
If $\omega\circ\alpha=\omega$, then there is a unitary $V_{\alpha}$ on $\caH_{\omega}$
such that
\begin{align*}
V_{\alpha}\pi_{\omega}(A)V_{\alpha}^{*}=\pi_{\omega}\lmk \alpha(A)\rmk,\quad
V_{\alpha}\Omega_{\omega}
=\Omega_{\omega},
\quad A\in\mkA.
\end{align*}
\end{thm}
Note from the condition that
\begin{align*}
V_{\alpha}\pi_{\omega}(A)\Omega_{\omega}
=\pi_{\omega}\lmk\alpha(A)\rmk\Omega_{\omega},
\quad A\in\mkA.
\end{align*}
In fact, that is how $V_\alpha$ is defined.

In particular, let $\gamma: G\to \Aut(\mkA)$
be a group action of $G$ on a $C^*$-algebra $\mkA$, and suppose that
a state $\omega$ is invariant under $\gamma$.
Then from the Theorem, we get
unitaries $V_g:=V_{\gamma_g}$.
It is a genuine unitary representation because
\begin{align*}
V_gV_h\pi_\omega(A)\Omega_\omega
=\pi_\omega\circ\gamma_g\gamma_h(A)\Omega_\omega
=\pi_\omega\circ\gamma_{gh}(A)\Omega_\omega
=V_{gh} \pi_\omega(A)\Omega_\omega.
\end{align*}

Using the GNS representation and the Hahn-Banach theorem, 
we can show that any $C^{*}$-algebra is isomorphic to a norm-closed self-adjoint algebra of bounded operators on a Hilbert space.

For a subset $\caM$ of operators on a Hilbert space $\caH$,
let $\caM'$ denote its commutant, i.e., the set of all bounded operators on $\caH$
commuting with every operator in $\caM$.
A von Neumann algebra on a Hilbert space $\caH$ is a $*$-subalgebra $\caM$ of $\caB(\caH)$
(the set of all bounded operators on $\caH$)
such that $\caM=\caM''$.
A center of a von Neumann algebra is defined as $\caZ(\caM):=\caM\cap\caM'$.
A von Neumann algebra $\caM$ is called a factor if it has a trivial center.
Given a state $\omega$ on a $C^*$-algebra,
 we obtain a von Neumann algebra $\pi_\omega(\mkA)''$
  out of the GNS representation of $\omega$.
 We call this von Neumann algebra 
 a von Neumann algebra associated to $\omega$.

A von Neumann algebra is said to be a type $I$-factor if it is isomorphic to
$\caB(\caH)$ for some Hilbert space $\caH$.
Type $I$ factors will play an important role in the analysis of SPT phases.
From this point of view, the following theorems are useful.
\begin{thm}\label{wigner}
Let $\caH_1$, $\caH_2$ be Hilbert spaces 
and $\tau: \caB(\caH_1)\to\caB(\caH_2)$
a $*$-isomorphism.
Then there is a unitary $u: \caH_1\to\caH_2$
such that
$\tau(x)=uxu^*$, $x\in\caB(\caH_1)$.
\end{thm}

\begin{thm}\label{purethm}
The following conditions are equivalent for a state $\omega$:
\begin{description}
\item[(1)] The representation $(\caH_{\omega},\pi_{\omega})$
is irreducible, i.e., the only closed subspaces of  $\caH_{\omega}$
which are invariant under the action of $\pi_{\omega}(\mkA)$
are the trivial subspaces $\{0\}$ and $\caH_{\omega}$.
\item[(2)] $\omega$ is pure,
\item[(3)] $\pi_{\omega}(\mkA)''=\caB(\caH_{\omega})$,
\item[(4)] $\pi_\omega(\mkA)'=\bbC\unit$.
\end{description}
\end{thm}

Another notion which plays an important role in the analysis of SPT phases is
{\it quasi-equivalence}. 
In probability theory, a probability measure $\lambda$ is {absolutely continuous} with respect to
another probability measure $\mu$ if there is a positive $h\in L^1(\mu)$
satisfying
\[
\int f\; d\lambda=\int f\; h\; d\mu,
\]
 for all {bounded measurable functions} $f$.
Furthermore, two measures are {equivalent}  if they are absolutely continuous to each other.
Quasi-equivalence is a non-commutative version of it.

If $(\caH,\pi)$ is a representation of a $C^{*}$-algebra $\mkA$,
then a state $\varphi$ of $\mkA$ is said to be $\pi$-normal if there exists 
a positive trace-class operator $\rho$ (called a density operator)
on $\caH$ with $\Tr(\rho)=1$
such that 
\begin{align}
\varphi(A)=\Tr\lmk \rho\pi(A)\rmk,\quad \text{for all}\quad A\in\mkA.
\end{align}
Two representations $\pi_{1}$ and $\pi_{2}$ of a $C^{*}$-algebra $\mkA$ are said to be
quasi-equivalent, written $\pi_{1}\sim_{q.e.}\pi_{2}$ 
if each $\pi_{1}$-normal state is $\pi_{2}$-normal and conversely.
\begin{thm}\label{qethm}
The following conditions are equivalent for representations $\pi_{1}$ and $\pi_{2}$ of $\mkA$:
\begin{description}
\item[(1)] $\pi_{1}\sim_{q.e.}\pi_{2}$
\item[(2)] there exists an isomorphism $\tau: \pi_{1}(\mkA)''\to \pi_{2}(\mkA)''$
such that $\tau\lmk \pi_{1}(A)\rmk=\pi_{2}(A)$
for all $A\in\mkA$.
\end{description}
\end{thm}
Two states $\omega_{1}$ and $\omega_{2}$ of $\mkA$ are said to be quasi-equivalent if 
$\pi_{\omega_{1}}$ and $\pi_{\omega_{2}}$ are quasi-equivalent.
Although any states on a matrix algebra are quasi-equivalent,
that is not the case for infinite systems.
Indeed, the following physical interpretation can be made.
By the cyclicity of the GNS vector $\Omega_{\omega}$ in the GNS-representation (\ref{gns}),
if $\omega_{1}$ and $\omega_{2}$ are quasi-equivalent,
then
$\omega_{2}$ can be approximated arbitrarily well in norm topology of $\mkA^{*}$
by states of the form
\begin{align*}
\sum_{n=1}^{\infty} \omega_{1}\lmk A_{n}^{*} \cdot A_{n}\rmk,
\end{align*}
with some $A_{n}\in\mkA$.
Regarding elements in $\mkA$ as local observables, it may be interpreted 
that $\omega_{2}$ can be obtained from $\omega_{1}$ via some 
local perturbation.
Because the converse also holds, heuristically we may say that 
two states $\omega_{1}$ and $\omega_{2}$ being quasi-equivalent 
means that they are macroscopically the same.

Two representations $\pi_{1}$ and $\pi_{2}$ of a $C^{*}$-algebra $\mkA$ 
are mutually singular if 
none of $\pi_2$-normal state is $\pi_1$-normal and converse.
One can regard this as an analog of mutually singular
measures (like $\delta$-function and Lebesgue measure)
in probability theory.

Recall that a von Neumann algebra is a factor if it has a trivial center.
\begin{thm}\label{facthm}
Let $(\caH,\pi)$ be a representation of a $C^*$-algebra $\mkA$.
If $\pi(\mkA)''$ is a factor, 
then any subrepresentation of $\pi$
is quasi-equivalent to $\pi$.
\end{thm}
We say that a state $\omega$ is a factor state
if associated von Neumann algebra is a factor.
Physically, it can be regarded as a {\it homogeneous} state, in the sense that
it cannot be written as a mixture of macroscopically distinct states (like a convex combination of
all spin up state and all spin down states).
A pure state is of course homogeneous, which even forbids mixture of macroscopically same states.

If an automorphism $\alpha$ on a $C^*$-algebra $\mkA$ is implementable (by a unitary operator $u$) 
in a GNS-representation
$(\caH_\omega,\pi_\omega)$ of a state $\omega$ on $\mkA$,
then $\alpha$ preserves the set of $\pi_\omega$-normal states.
Indeed, let $\varphi$ be a $\pi_\omega$-normal state.
Then there is a density operator $\rho$ on $\caH_\omega$
satisfying
\[
\varphi(A):=\Tr\lmk \rho \pi_\omega(A)\rmk,\quad
A\in{\mkA}.
\]
We then have
\begin{align*}
\varphi\circ\alpha(A)
=\Tr\lmk \rho \pi_\omega\circ\alpha(A)\rmk
=\Tr\lmk \rho  u\pi_\omega(A)u^*\rmk
=\Tr\lmk \lmk u^*\rho  u\rmk\pi_\omega(A)\rmk,\quad A\in{\mkA}.
\end{align*}
Hence the state $\varphi\circ\alpha$ is given by
a density operator $u^*\rho  u$.
This means that $\varphi\circ\alpha$ is $\omega$-normal.
From the above interpretation of macroscopic equivalence,
$\alpha$ being implementable on the GNS representation of $\omega$
means that it does not move $\omega$ {\it macroscopically}.

We also prepare tensor product of $C^*$-algebras.
Let $\mkA$, $\mkB$ be $C^{*}$-algebras with faithful representations
$(\caH_{\mkA},\pi_{\mkA})$, $(\caH_{\mkB},\pi_{\mkB})$.
The the spatial $C^{*}$-norm $\lV\cdot\rV$ on algebraic tensor product 
$\mkA\odot \mkB$ is defined by
\begin{align*}
\lV \sum_{i} a_{i}\otimes b_{i}\rV
:=\lV
\sum_{i} \pi_{\mkA}(a_{i})\otimes \pi_{\mkB}(b_{i})
\rV_{\caB(\caH_{\mkA}\otimes\caH_{\mkB})}.
\end{align*}
The right-hand side means the norm as an element of $\caB(\caH_{\mkA}\otimes\caH_{\mkB})$.
The completion of $\mkA\odot \mkB$ with respect to this norm
is denoted by $\mkA\otimes\mkB$.
This is the spacial tensor product of $\mkA$ and $\mkB$.
For states $\omega_{\mkA}$, $\omega_{\mkB}$ 
on $\mkA$, $\mkB$, there exists a unique state
$\omega_{\mkA}\otimes \omega_{\mkB}$ on  $\mkA\otimes\mkB$
such that
\begin{align}
\lmk
\omega_{\mkA}\otimes \omega_{\mkB}
\rmk
(A\otimes B)
=\omega_{\mkA}(A)\omega_{\mkB}(B),\quad A\in\mkA,\; B\in\mkB.
\end{align}
Furthermore, for representations
$(\caH_\mkA,\pi_\mkA)$, $(\caH_\mkB,\pi_\mkB)$
of $\mkA,\mkB$,
there is a unique representation $\pi_{\mkA}\otimes\pi_{\mkB}$ of $\mkA\otimes \mkB$
on $\caH_{\mkA}\otimes\caH_{\mkB}$
such that
\[
\lmk
\pi_{\mkA}\otimes\pi_{\mkB}
\rmk(A\otimes B)
=\pi_{\mkA}(A)\otimes\pi_{\mkB}(B),\quad
A\in\mkA, \quad B\in\mkB.
\]

In the analysis of SPT-phases we encounter the following situation.
\begin{lem}\label{elementlem}
Let $\omega_{\mkA}$, $\omega_{\mkB}$ 
be pure states on $C^*$-algebras $\mkA$, $\mkB$,
with GNS triples $(\caH_{\mkA},\pi_{\mkA},\Omega_{\mkA})$ and 
$(\caH_{\mkB},\pi_{\mkB},\Omega_{\mkB})$.
Let $\alpha_{\mkA}$, $\alpha_{\mkB}$ be automorphisms on $\mkA$, 
$\mkB$
such that 
\begin{align*}
\omega_{\mkA}\alpha_{\mkA}\otimes \omega_{\mkB}\alpha_{\mkB}
=\lmk \omega_{\mkA}\otimes \omega_{\mkB}\rmk
\circ\Ad V,
\end{align*}
for some unitary $V\in \mkA\otimes\mkB$.
Then there is a unitary $u_{\mkB}$ on $\caH_\mkB$ such that
\begin{align*}
\Ad(u_\mkB)\circ\pi_{\mkB}(A)=\pi_{\mkB}\circ\alpha_{\mkB}(A),\quad A\in \mkB.
\end{align*}
\end{lem}
\begin{proof}
First, note that $(\caH_{\mkA}\otimes\caH_{\mkB},
\pi_{\mkA}\otimes\pi_{\mkB}, \Omega_{\mkA}\otimes\Omega_{\mkB})$
is a GNS representation of $\omega_{\mkA}\otimes \omega_{\mkB}$.
By the assumption,
$\omega_{\mkA}\otimes \omega_{\mkB}$ is invariant under
\[
\lmk \alpha_{\mkA}\otimes\alpha_{\mkB}\rmk\circ \Ad(V^*).
\]
Then, from Theorem \ref{gnsthm}, there is a unitary $w$
on $\caH_{\mkA}\otimes\caH_{\mkB}$ such that
\begin{align*}
\Ad(w) \lmk \pi_{\mkA}\otimes\pi_{\mkB}\rmk
=\lmk \pi_{\mkA}\otimes\pi_{\mkB}\rmk \lmk \alpha_{\mkA}\otimes\alpha_{\mkB}\rmk.
\end{align*}
(We absorbed $\Ad V^*$ part to $\Ad(w)$.)
From this, we get
\[
\Ad(w)\lmk \unit_{\caH_{\mkA}} \otimes \pi_{\mkB}(\mkB)''\rmk
\subset \unit_{\caH_{\mkA}} \otimes \pi_{\mkB}(\mkB)'',\quad
\Ad(w^*)\lmk \unit_{\caH_{\mkA}} \otimes \pi_{\mkB}(\mkB)''\rmk
\subset \unit_{\caH_{\mkA}} \otimes \pi_{\mkB}(\mkB)''.
\]
Hence we obtain
\[
\Ad(w)\lmk \unit_{\caH_{\mkA}} \otimes \pi_{\mkB}(\mkB)''\rmk
=\unit_{\caH_{\mkA}} \otimes \pi_{\mkB}(\mkB)''.
\]
Now, because $\omega_{\mkB}$ is a pure state, we have
$\pi_{\mkB}(\mkB)''=\caB(\caH_{\mkB})$ by
Theorem \ref{purethm}.
This means there is a $*$-automorphism $\tau$ on $\caB(\caH_{\mkB})$
such that
\[
\Ad(w)\lmk\unit_{\caH_{\mkA}}\otimes x\rmk
=\unit_{\caH_{\mkA}}\otimes\tau (x),\quad x\in \caB(\caH_{\mkB}).
\]
By Theorem \ref{wigner}, a $*$-isomorphism $\tau:  \caB(\caH_{\mkB})\to  \caB(\caH_{\mkB})$
is implemented by a unitary $u_\mkB\in \caH_{\mkB}$, i.e.
\[
\tau(x)=\Ad(u_\mkB)(x),\quad x\in \caB(\caH_{\mkB}).
\]
Hence from all the equations above, we obtain
\begin{align*}
\begin{split}
&\lmk \unit_{\caH_{\mkA}}\otimes\pi_{\mkB}\alpha_{\mkB}(A)\rmk
=\Ad(w) \lmk \unit_{\caH_{\mkA}}\otimes\pi_{\mkB}(A)\rmk\\
&=\unit_{\caH_{\mkA}}\otimes\tau\pi_{\mkB}(A)
=\unit_{\caH_{\mkA}}\otimes \Ad(u_\mkB)\pi_{\mkB}(A),\quad A\in\mkB.
\end{split}
\end{align*}
This proves the Lemma.
\end{proof}
The following Lemma is sometimes useful.
\begin{lem}\label{ih}
Let $\mkA$, $\mkB$ be $C^*$-algebras.
Let $\omega$ be a pure state on $\mkA\otimes\mkB$.
Then the GNS representation $\pi_{\omega\vert_{\mkA}}$ of $\omega\vert_{\mkA}$ (the restriction of $\omega$ to
$\mkA\simeq\mkA\otimes \bbC\unit_{\mkB}$ ), 
and the restriction of the GNS representation $\pi_\omega$ of $\omega$ to $\mkA$,
 $\pi_{\omega}\vert_{\mkA}$
 are quasi-equivalent.
\end{lem}
\begin{proof}
Let $(\caH_\omega, \pi_\omega,\Omega_\omega)$ be a GNS representation of $\omega$.
First we show that $\pi_{\omega}\lmk\mkA \rmk ''$
is a factor.
To see this, note that
\begin{align}
\pi_{\omega}\lmk\mkA \rmk''\cap \pi_{\omega}\lmk\mkA \rmk'
\subset \pi_{\omega}\lmk\mkB \rmk'\cap  \pi_{\omega}\lmk\mkA \rmk'
=\pi_\omega\lmk \mkA\otimes\mkB\rmk'
=\bbC\unit,
\end{align}
because $\omega$ is pure. (Recall Theorem \ref{purethm}.)
Hence $\pi_{\omega}\lmk\mkA \rmk ''$
is a factor.
Let 
\begin{align*}
&\caK:=\overline{\pi_\omega(\mkA)\Omega_\omega},\quad 
p: \;\; \text{projection onto } \caK,\quad \pi:=\pi_\omega (\cdot )p\vert_{\mkA}
\end{align*}
The projection $p$ belongs to $\pi_{\omega}\lmk\mkA \rmk '$, and
 one can check that 
$(\caK, \pi,\Omega_\omega)$ is a GNS representation of $\omega\vert_{\mkA}$.
Note that $\pi$ is a sub-representation of $\pi_\omega\vert_{\mkA}$.
Because $\pi_{\omega}\lmk\mkA \rmk ''$
is a factor,
 Theorem \ref{facthm} guarantees that $\pi$ and  $\pi_{\omega}\vert_{\mkA}$ are quasi-equvalent.
\end{proof}

\subsection{Quantum spin systems}\label{quantumspinsec}

Let $d\in \bbN$ and $\nu\in\bbN$ be fixed.
Physically, $\frac{d-1}2$ denotes 
the size of on-site spin (spin quantum number)  and 
$\nu$ denotes the spacial dimension.
We denote by ${\mathfrak S}_{\bbZ^\nu}$, the set of all finite subsets of $\bbZ^{\nu}$.
For each finite subset $\Lambda\in {\mathfrak S}_{\bbZ^\nu}$, 
we associate a finite dimensional $C^{*}$-algebra
\begin{align*}
\caA_{\Lambda}:=\bigotimes_{\Lambda}\Mat_{d}.
\end{align*}
Here, $\Mat_{d}$ is the algebra of $d\times d$ -matrices.
If $\Lambda_{1},\Lambda_{2}\in {\mathfrak S}_{\bbZ^\nu}$
satisfy an inclusion relation $\Lambda_{1}\subset\Lambda_{2}$,
then we may embed $\caA_{\Lambda_{1}}$ into $\caA_{\Lambda_{2}}$ as
\begin{align*}
\iota_{\Lambda_{1},\Lambda_{2}}: \caA_{\Lambda_{1}}\to \caA_{\Lambda_{2}},\quad
\iota_{\Lambda_{1},\Lambda_{2}}(A):=A\otimes\unit_{\Lambda_{2}\setminus \Lambda_{1}},\quad
A\in \caA_{\Lambda_{1}}.
\end{align*}
Clearly, we have $\iota_{\Lambda_{2},\Lambda_{3}}\iota_{\Lambda_{1},\Lambda_{2}}
=\iota_{\Lambda_{1},\Lambda_{3}}$ if $\Lambda_{1}\subset\Lambda_{2}\subset\Lambda_{3}$.
Then the $\nu$-dimensional quantum spin system $\caA_{\bbZ^\nu}$ is the 
$C^{*}$-inductive limit of this inductive net $(\iota_{\Lambda_{1}, \Lambda_{2}}, \caA_{\Lambda})$. Namely,
$\caA_{\bbZ^\nu}$ is the unique $C^{*}$-algebra satisfying the following conditions.
\begin{description}
\item[(i)]
There is a net of $*$-homomorphisms 
$\mu_{\Lambda} :\caA_{\Lambda}\to \caA_{\bbZ^\nu}$, with $\Lambda\in {\mathfrak S}_{\bbZ^\nu}$,
such that $\mu_{\Lambda_{1}}=\mu_{\Lambda_{2}}\circ \iota_{\Lambda_{1},\Lambda_{2}}$
if
$\Lambda_{1}\subset\Lambda_{2}$, and 
\item[(ii)]
If $(\caB, \{ \lambda_{\Lambda}\}_{\Lambda\in {\mathfrak S}_{\bbZ^\nu}})$ is a system,
where $\caB$ is a $C^{*}$-algebra, $\lambda_{\Lambda}: \caA_{\Lambda}\to \caB$
is a $*$-homomorphism for each $\Lambda\in {\mathfrak S}_{\bbZ^\nu}$, and where $\lambda_{\Lambda_{1}}=\lambda_{\Lambda_{2}}\circ \iota_{\Lambda_{1},\Lambda_{2}}$ for any
$\Lambda_{1},\Lambda_{2}\in {\mathfrak S}_{\bbZ^\nu}$
with $\Lambda_{1}\subset\Lambda_{2}$,
then there is one and only one $*$-homomorphism $\lambda: \caA_{\bbZ^\nu}\to \caB$
such that $\lambda_{\Lambda}=\lambda\circ \mu_{\Lambda}$.
\end{description}
The subset $\cup_{\Lambda\in {\mathfrak S}_{\bbZ^\nu}}\mu_{\Lambda}\lmk \caA_{\Lambda}\rmk$,
which is dense in
$\caA_{\bbZ^\nu}$, is called a local algebra, and we denote it by $\caA_{\rm loc}$.
Because $\mu_\Lambda$ is injective, we identify
$\caA_\Lambda$ and $\mu_{\Lambda}\lmk \caA_{\Lambda}\rmk$
and omit the symbol $\mu_{\Lambda}$
For each infinite subset $\Gamma$, we may define $\caA_\Gamma$
in exactly the same manner.
The $C^*$-algebra $\caA_\Gamma$ can be regarded as a $C^*$-subalgebra
of $\caA_{\bbZ^\nu}$ that we will regard it so.
We say an element $A$ has a support in $\Gamma$ if
it belongs to $\caA_\Gamma$.
If an automorphism $\alpha$ acts trivially on $\caA_{\Gamma^c}$
for some $\Gamma\subset \bbZ^{\nu}$,
we say that $\alpha$ has a support in $\Gamma$.
In the analysis of SPT-phases, this support becomes important.

Corresponding to the space translation of $\bbZ^\nu$,
we may define a $\bbZ^\nu$-action $\tau: \bbZ^{\nu}\to \Aut\lmk \caA_{\bbZ^\nu}\rmk$
 on $\caA_{\bbZ^\nu}$
by
\[
\tau_{\bf y}\lmk A^{(\bf x)} \rmk
= \lmk A^{(\bf x+y)} \rmk
\]
for $A\in\Mat_d$ and ${\bf x,y}\in\bbZ^\nu$.
Here $ A^{(\bf x)}$ indicates an element in
$\caA_{\bbZ^\nu}$ with $A$ in the $\bf x$-site of the tensor product of 
$\caA_{\bbZ^\nu}$ and the unit in any other site.
This $\tau$ is called the space translation on $\caA_{\bbZ^\nu}$.

In operator algebraic framework, physical models are specified with a map
called {\it interactions}.
An interaction $\Phi$ is a map 
\[
\Phi: {\mathfrak S}_{\bbZ^\nu}\to \caA_{\rm loc}
\]
satisfying
\begin{align*}
\Phi(X) = \Phi(X)^*\in \caA_{X}
\end{align*}
for all $X \in {\mathfrak S}_{\bbZ^\nu}$. 
Physically, this $\Phi(X)$ indicates an interaction term between 
spins inside of $X$.
An interaction $\Phi$ is said to be translationally invariant,
if it satisfies the following condition
\[
\tau_{\bf x}\lmk \Phi(X)\rmk
=\Phi(X+{\bf x}),\quad A\in{\mathfrak S}_{\bbZ^\nu},\quad {\bf x}\in\bbZ^\nu.
\]

The easiest type of interaction is an {\it on-site interaction}, satisfying
\begin{align}\label{onsitedef}
\Phi(X)=0,\quad\text{if}\quad |X|\neq 1.
\end{align}
It means, the only possibly non-zero interaction terms are of the form
$\Phi(\{\bf x\})$, with ${\bf x}\in\bbZ^\nu$.
In this case, all interaction terms commute with each other.

Physically, we are more interested in interactions which do have non-zero
interaction terms between different sites of $\bbZ^{\nu}$.
For example, let 
$\{S_j\}_{j=1,2,3}$ be generators of the irreducible representation of ${\mathfrak{su}(2)}$ on $\bbC^d$.
Then an interaction of $\caA_{\bbZ^\nu}$ given by 
\begin{align}\label{heisenberg}
\Phi(\{x,x+1\})=\sum_{j=1}^3\,S^{(x)}_jS^{(x+1)}_j,
\end{align}
is called antiferromagnetic Heisenberg chain and has been extensively studied.

Now, given an interaction,
we would like to define a dynamics on $\caA_{\bbZ^\nu}$ out of it.
In order for that, we need to assume that $\Phi$ is ``suitably local''.
The simplest one among such good conditions is {\it uniformly bounded and finite range}.
An interaction is 
of finite range if there exists  an $m\in {\mathbb N}$ such that
$\Phi(X)=0$ for $X$ with a diameter larger than $m$.
It is  uniformly bounded if it satisfies
\begin{align*}
\sup_{X\in  {\mathfrak S}_{\bbZ^\nu}}\lV
\Phi(X)
\rV<\infty.
\end{align*}
We can extensively relax these conditions. See \cite{nsy}.

Given a good interaction, we may define a $C^*$-dynamics, i.e., strongly continuous one parameter group of
automorphisms on $\caA_{\bbZ^\nu}$.
For an interaction $\Phi$ and a finite set $\Lambda\subset {\bbZ^\nu}$, we define the local Hamiltonian
by
\begin{equation}\label{GenHamiltonian}
\lmk H_{\Phi}\rmk_{\Lambda}:=\sum_{X\subset{\Lambda}}\Phi(X).
\end{equation}
For a uniformly bounded finite range interaction $\Phi$, the limit
\begin{equation}\label{dyn}
\alpha_t^{\Phi}(A)=\lim_{\Lambda\to\bbZ^\nu} e^{it(H_{\Phi})_{\Lambda}} Ae^{-it(H_{\Phi})_{\Lambda}},\quad
t\in \bbR,\quad A\in\caA_{\bbZ^\nu}
\end{equation}
exists and defines a dynamics $\alpha^\Phi$ on $\caA_{\bbZ^\nu}$.
The reason for us to consider the dynamics $\alpha^{\Phi}$
instead of Hamiltonians, is because there is no
mathematically meaningful limit of local Hamiltonians $\lmk H_{\Phi}\rmk_{\Lambda}$
as $\Lambda\to\bbZ^\nu$ in (\ref{GenHamiltonian}).
But
the limit (\ref{dyn}) makes sense.
For the same reason, a ground state is defined in terms of the dynamics $\alpha_{\Phi}$.
\begin{defn}\label{groundstatedef}
Let $\delta_\Phi$ be the generator of $\alpha^\Phi$.
A state $\omega$ on $\caA_{\bbZ^\nu}$ is called an $\alpha^\Phi$-ground state
if the inequality
\[
-i\omega\lmk A^*{\delta_\Phi}\lmk A\rmk\rmk\ge 0
\]
holds
for any element $A$ in the domain $\caD({\delta_\Phi})$ of ${\delta_\Phi}$.
\end{defn}
We occasionally say ground state of $\Phi$ instead of $\alpha^{\Phi}$-ground state.

Let $(\caH_\varphi,\pi_\varphi,\Omega_\varphi)$ be the GNS triple of a $\alpha^\Phi$-ground state $\varphi$.
Then there exists a unique positive operator $H_{\varphi,\Phi}$ on $\caH_\varphi$ such that
$e^{itH_{\varphi,\Phi}}\pi_\varphi(A)\Omega_\varphi=\pi_\varphi(\tau_t^\Phi(A))\Omega_\varphi$,
for all $A\in\caA$ and $t\in\mathbb R$.
We call this $H_{\varphi,\Phi}$ the bulk Hamiltonian associated with $\varphi$.
Note that $\Omega_\varphi$ is an eigenvector of $H_{\varphi,\Phi}$ with eigenvalue $0$. See \cite{BR2} for the general theory.

Let us consider the corresponding condition for a finite quantum system $\Mat_n$ with dynamics
given by a Hamiltonian $H$ 
\begin{align}\label{finitealpha}
\alpha_t(A)=e^{it H} A e^{-itH},\quad t\in{\mathbb R}, \quad A\in{\Mat_n}.
\end{align}
Let $P$ be the spectral projection of $H$ corresponding to the lowest eigenvalue $E_{0}$.
Recall that 
a state $\omega$ on $\Mat_{n}$ is given by a density matrix $\rho$ 
with the formula $\omega(A)=\Tr\rho A$.
Let $s(\omega)$ be the range projection of this $\rho$.
Then one can check that $\omega$
 is an $\alpha$-ground state if and only if
the support $s(\omega)$ of $\omega$ satisfies $s(\omega)\le P$.
Recall that this is the very definition of ground state, in finite quantum mechanics.
In fact note that the generator $\delta$ of $\alpha$ in (\ref{finitealpha}) is
$\delta(A)=i[H,A]$.
If $s(\omega)\le P$, then we have
\begin{align*}
-i\omega\lmk A^*{\delta}\lmk A\rmk\rmk=
\omega\lmk A^{*}(H-E_{0})A\rmk\ge 0,\quad A\in\Mat_{n}.
\end{align*}
Conversely, suppose that $-i\omega\lmk A^*{\delta}\lmk A\rmk\rmk\ge 0$
for all $A\in\Mat_{n}$. For unit eigenvectors $\xi,\eta$ of $H$ with $H\xi=E_{0}\xi$, $H\eta=E\eta$, for $E>E_{0}$,
set $A\in\Mat_{n}$ to be a matrix satisfying $A\zeta=\braket{\eta}{\zeta}\xi$ for any $\zeta\in\bbC^{n}$.
Substituting this $A$, we get 
\begin{align*}
0\le -i\omega\lmk A^*{\delta}\lmk A\rmk\rmk
=\lmk E_{0}-E\rmk\braket{\eta}{\rho\eta}
\end{align*}
Because $E_{0}-E<0$, this means that $\braket{\eta}{\rho\eta}=0$ for any such $\eta$.
Hence we conclude that $\rho P=\rho$, namely, $s(\omega)\le P$.

Therefore, our definition in operator algebraic framework can be regarded
as a generalization of this physical definition of a ground state.

Note, in general, there can be many states satisfying the condition in Definition \ref{groundstatedef}.
Namely, the ground state does not need to be unique.
However, what we are interested in this lecture is the situation that it is unique.
Furthermore, we require the existence of gap.
\begin{defn}
Suppose that there is a unique $\alpha^\Phi$-ground state $\omega_\Phi$.
Then we say $\Phi$ has a unique gapped ground state in the bulk if
there exists a $\gamma>0$ such that
\begin{align}\label{gapeq1}
-i\omega_\Phi\lmk A^*{\delta_\Phi}\lmk A\rmk\rmk\ge \gamma
\omega_\Phi(A^*A),\quad \text{for all }A\in \caD({\delta_\Phi}) \;\text{with} \; \omega_\Phi(A)=0.
\end{align}
\end{defn}
Let us consider the corresponding condition for a finite system $\Mat_n$ with dynamics
(\ref{finitealpha}).
Then the above condition corresponds to the situation
that ``the lowest eigenvalue of $H$ is non-degenerated and 
the difference between the lowest eigenvalue and the second lowest eigenvalue is at least $\gamma$''.

Let $\Lambda_N:=[-N,N]^\nu\cap \bbZ^\nu$.
Let $\omega_{N,\Phi}$ be a ground state of the local Hamiltonian
$ \lmk H_{\Phi}\rmk_{\Lambda_N}$ on $\caA_{\Lambda_N}$.
Then, any wk$*$-accumulation point of $\omega_{N,\Phi}$
for $N\to\infty$ is an $\alpha^\Phi$-ground state.
For this reason, there is always an $\alpha^\Phi$-ground state.
Let $\sigma\lmk \lmk H_{\Phi}\rmk_{\Lambda_N}\rmk$
be the spectrum of $ \lmk H_{\Phi}\rmk_{\Lambda_N}$,
and $E_N$ its infimum.
Suppose that there is a $\gamma>0$
such that 
\[
\sigma\lmk \lmk H_{\Phi}\rmk_{\Lambda_N}\rmk
\setminus \{E_N\}\subset [E_N+\gamma,\infty)
\]
holds
for any $N\in\bbN$.
If furthermore the $\alpha^{\Phi}$-ground state $\omega_{\Phi}$ is unique,
then we have
\[
\omega_{\Phi}={\mathrm {wk}}*-\lim_{N\to\infty} \omega_{N,\Phi}
\]
and it is a unique gapped ground state with gap at least $\gamma$.

One remarkable property of unique gapped ground state is the exponential decay of correlation functions.
\begin{thm}[\cite{hk} \cite{ns06} \cite{ns09}]\label{expdec}
Let $\Phi$ be a uniformly bounded finite
range interaction with a {unique gapped ground state} $\omega_{\Phi}$.
Then the {correlation functions of $\omega_\Phi$
decay exponentially fast} : there exists $\mu>0$ and a constant
$C>0$ such that for all $A\in\caA_X$,
$B\in \caA_Y$, with finite $X,Y\subset\bbZ^\nu$
\[
\lv
\omega_{\Phi}\lmk A B\rmk
-\omega_{\Phi}(A)\omega_{\Phi}(B)
\rv\le C \lV A\rV\lV B\rV |X|{e^{-\mu d(X,Y)}}.
\]
\end{thm}
This means $\omega_\Phi$ is { ``almost like a product state"}.

In the theory of SPT-phases, we are interested in interactions satisfying some given symmetry. 
In this lecture, we consider on-site finite group symmetry in
one and two-dimensional quantum spin systems
and reflection symmetry in one dimensional quantum systems.

On-site finite group symmetry is given as follows.
We fix a  finite group $G$ and a (projective) unitary representation $U$ of $G$
on $\bbC^d$.
Then there is a unique automorphism $\beta_{g}$ satisfying
\begin{align*}
\beta_g(A)=\lmk \bigotimes_{x\in\Lambda} U(g)\rmk A\lmk \bigotimes_{x\in\Lambda} U(g)^*\rmk,\;
g\in G,\; \Lambda\in{\mathfrak S}_{\bbZ^\nu},\; A\in\caA_{\Lambda}.
\end{align*}
Clearly, this gives an action of $G$ on $\caA_{\bbZ^\nu}$, i.e., $\beta_{g}\beta_{h}=\beta_{gh}$
for $g,h\in G$.
We call this action of $G$, an on-site symmetry given by $G$ and $U$.
Likewise, for each subset $\Gamma\subset\bbZ^\nu$,
we may define an action $\beta_g^\Gamma$.
We say an interaction $\Phi$ is $\beta$-invariant
if $\beta_g(\Phi(X))=\Phi(X)$
for all $X\in {\mathfrak S}_{\bbZ^\nu}$ and $g\in G$.
For a ground state $\varphi$ of a $\beta$-invariant interaction $\Phi$, one can check that
$\varphi\circ\beta_{g}$ is also a ground state of $\Phi$.
Therefore, if a $\beta$-invariant interaction $\Phi$
has a unique ground state $\omega_{\Phi}$, the ground state is
$\beta$-invariant, $\omega_{\Phi}\circ\beta_{g}=\omega_{\Phi}$.

For example, let $U$ be a projective representation of
$G=\bbZ_2\times\bbZ_2$ on $\bbC^d$, given by
\begin{align*}
&U\lmk [1],[0] \rmk=e^{i\pi \hS_1},\quad U\lmk [0],[1] \rmk=e^{i\pi \hS_2},\quad U\lmk [1],[1] \rmk=e^{i\pi \hS_3}
\end{align*}
Here, $\{S_j\}_{j=1,2,3}$ are generators of the irreducible representation of ${\mathfrak{su}(2)}$ on $\bbC^d$.
If $d$ is odd, this gives a genuine representation while 
if $d$ is even, it is a projective representation.
One can check that antiferromagnetic Heisenberg chain (\ref{heisenberg})
is invariant under this symmetry.

Reflection symmetry in one dimensional quantum systems is given as follows.
Let $R:\bbZ\to\bbZ$ be the reflection : $R(i):=-i-1$, $i\in\bbZ$, and $\Theta_{R}\in\Aut(\caA_{\bbZ})$
an automorphism corresponding to $R$.
More precisely, $\Theta_R$ is defined as follows.
By $Q^{(j)}$, $j\in\bbZ$, we denote the element of $\caA_{\bbZ}$
with $Q\in\Mat_d$ in the $j$-th component of the tensor product of $\caA_{\bbZ}$ and the unit in any other component.
The reflection $\Theta_{R}$ is the unique  $*$-automorphism on $\caA_{\bbZ}$ 
which satisfies
\begin{align*}
\Theta_{R}\lmk Q^{(j)}\rmk=Q^{(-j-1)},\quad \text{for all}\; Q\in\Mat_{d}\;\text{and}\; j \in\bbZ.
\end{align*}
An interaction $\Phi$ is reflection invariant
if $\Theta_{R}(\Phi(X))=\Phi(R(X))$
for all $X\in {\mathfrak S}_\bbZ$.
One can check that antiferromagnetic Heisenberg chain (\ref{heisenberg})
is invariant under this symmetry.

\subsection{SPT-phases}
We denote the set of all uniformly bounded finite range interactions with
unique gapped ground state by $\caP_{U.G.}$.
First, we introduce an equivalence relation $\sim$ in this set $\caP_{U.G.}$
as follows.:
\begin{defn}[Definition of $\sim$]\label{simdef}
Two interactions $\Phi_0,\Phi_1\in \caP_{U.G.}$ are equivalent (denoted as $ \Phi_0\sim\Phi_1$)
if there is a path 
$\caP_{U.G.}$
in $\Phi: [0,1]\to \caP_{U.G.}$
satisfying the following conditions.
\begin{description}
\item[(i)] The path connects given interactions, i.e., $\Phi(0)=\Phi_{0}$ and $\Phi(1)=\Phi_{1}$.
\item[(ii)] For each $X\in{\mathfrak S}_{\bbZ^{\nu}}$,
the matrix-valued function $[0,1]\ni s\mapsto \Phi_s(X)\in\caA_{X}$ is smooth,
\item[(iii)] The gap is uniformly bounded from below by some $\gamma>0$ along the path.
Namely, for each $s\in[0,1]$, there is a $\gamma_{s}>0$ satisfying (\ref{gapeq1})
for $\Phi=\Phi(s)$, and $\inf_{s\in [0,1]}\gamma_{s}>0$.
\item[(iv)] The path of expectation values $[0,1]\ni s\mapsto \omega_s(A)\in\bbC$ of sub-exponentially localized
elements
$A\in\caA_{\bbZ^\nu}$ with respect to the ground state $\omega_s$ is regular
(in the sense that it is differentiable and the derivative is not too large compared to
some norm of $A$)
 with respect to $s\in[0,1]$.
\end{description}
\end{defn}
For more precise description of the conditions on the path,  see subsection \ref{autosec} Assumption \ref{assump} Assumption \ref{infautoassump}.
A part of these assumptions  are for technical reason, namely we would like to apply
Theorem \ref{autoinf}.
Physically, this condition means that {\it there is no phase transition along the path}.

Among many interactions in $\caP_{U.G.}$, what corresponds to SPT is
the ones {\it with short-range entanglement}.
Recall the definition of on-site interaction (\ref{onsitedef}).
Fix some on-site interaction $\Phi_{0}$ which belongs to
$\caP_{U.G.}$.
Such an interaction can be constructed as follows.
For each $x\in\bbZ^{\nu}$, fix any unit vector $\xi_{x}\in\bbC^{d}$.
Let $p_{x}$ be an orthogonal projection on $\bbC^{d}$ onto the subspace
$\bbC \xi_{x}$.
Set $\Phi_{0}$ as
\begin{align*}
\Phi_{0}\lmk \{x\}\rmk
=\unit-p_{x},\quad x\in\bbZ^{\nu},
\end{align*}
and $\Phi_{0}(X)=0$ for any other $X\in{\mathfrak S}_{\bbZ^{\nu}}$.
Then $\Phi_{0}\in \caP_{U.G.}$
and its unique ground state is of the form
\begin{align*}
\omega_{\Phi_{0}}=\bigotimes_{x\in\bbZ^{\nu}}\rho_{\xi_{x}},
\end{align*}
where $\rho_{\xi_{x}}$ is a vector state on $\Mat_{d}$
given by $\rho_{\xi_{x}}(A)=\braket{\xi_{x}}{A\xi_{x}}$.
Note that any on-site interaction can be smoothly deformed into
an on-site interaction of this form inside of $\caP_{U.G.}$,
just by smoothly flattering the spectrum of each $\Phi(\{x\})$.
Furthermore, any on-site interactions of this form can be deformed
into each other by rotating $\xi_{x}$ smoothly.
Therefore, the following set does not depend on the choice of $\Phi_{0}$:
\begin{align*}
\caP_{U.G.}^{0}:=
\left\{
\Phi\in \caP_{U.G.}\mid 
\Phi\sim \Phi_0.
\right\}.
\end{align*}
As we see in the next subsection, this set physically represents 
interactions whose unique gapped ground state is {\it short-range entangled}.
This excludes
 exotic states with “topological order” (or long-range entanglement) such as the ground states of the Kitaev toric code model \cite{Kitaev}. In general, a pure state with a non-trivial super selection sector
 has long-range entanglement \cite{NO}.

Now we introduce the on-site  symmetry $\beta$ to the game.
We are interested in 
the set 
\begin{align}\label{pugb}
\caP_{U.G.{{\beta}}}^{0}:=
\left\{
\Phi\in \caP_{U.G.}^{0}\mid 
\beta\text{-invariant}
\right\}.
\end{align}
We introduce an equivalence relation $\sim_{\beta}$ on
this set $\caP_{U.G.{{\beta}}}^{0}$.
\begin{defn}[Definition of $\sim_{\beta}$]\label{simbdef}
Two interactions $\Phi_0,\Phi_1\in  \caP_{U.G.{{\beta}}}^{0}$ are equivalent (denoted as $ \Phi_0\sim_{\beta}\Phi_1$)
if there is a path 
in $\Phi: [0,1]\to \caP_{U.G.{{\beta}}}^{0}$
satisfying the conditions (i)-(iv) of Definition \ref{simdef}.
\end{defn}
The equivalence class of $ \caP_{U.G.{{\beta}}}^{0}$
with this classification are the SPT-phases.
\begin{defn}[Symmetry Protected Topological (SPT) phases]
Each equivalence class of ${\caP}_{U.G.{\beta}}^{0}$
with respect to $\sim_{\beta}$ is called 
Symmetry Protected Topological (SPT) phases.
\end{defn}
This notion of SPT-phases was introduced by Gu and Wen in \cite{GuWen2009}.
This is the object we would like to classify.
Analogously, we can define SPT phases $ \caP_{U.G.{{\Theta_{R}}}}^{0}$
for the reflection symmetry $\Theta_{R}$.

\subsection{Split property}
The split property, studied extensively in operator algebra has turned out to be quite  important 
in analysis of SPT phases.
In this section we introduce the  notion.

Let us consider an inclusion of two von Neumann algebras $\caN\subset \caM$.
We say the inclusion is {\it split}, if there is a type $I$ factor $\caR$
such that $\caN\subset \caR\subset \caM$.
The inclusion which shows up in this lecture is the simplest
case, that is, $\caN$ itself is a type $I$ factor.
In particular, we will use the following terminology in 
this lecture.
\begin{defn}
Let $\Gamma_1, \Gamma_2 \subset\bbZ^{\nu}$ be non-empty disjoint
 subsets of $\bbZ^\nu$.
 A pure state $\omega$ on $\caA_{\Gamma_1\cup\Gamma_2}$
satisfies the split property with respect to 
$\Gamma_1$-$\Gamma_2$, if 
the von Neumann algebra $\pi_\omega(\caA_{\Gamma_1})''$
is a type $I$ factor.
\end{defn}
It is beneficial to heed in mind the following Theorems about type $I$ factors.
\begin{thm}[Theorem 1.31 V\cite{takesaki}]\label{ubh}
Let $\caM$ be a type $I$ factor acting on a Hilbert space $\caH$.
Then there are Hilbert spaces $\caH_1$, $\caH_2$
and a unitary $U:\caH\to \caH_1\otimes\caH_2$
such that
\begin{align*}
\Ad(U)(\caM)=\caB(\caH_1)\otimes\bbC\unit_{\caH_2}.
\end{align*}
\end{thm}
\begin{thm}[Theorem 2.30 V\cite{takesaki}]
Let $\caM_1$, $\caM_2$ be von Neumann algebras
acting on Hilbert spaces $\caH_1$, $\caH_2$.
Let $\caM$ be a von Neumann 
algebra acting on $\caH_1\otimes\caH_2$
generated by the algebraic tensor product $\caM_1{\odot}\caM_2$.
Then $\caM$ is of type $I$ if and only if $\caM_1$ and $\caM_2$
are both of type $I$.
\end{thm}
Using such properties, we can show the following.
\begin{prop}\label{splitcharac}
For a pure state $\omega$ on $\caA_{\Gamma_1\cup\Gamma_2}$, following conditions are equivalent.
\begin{description}
\item[(i)] The state $\omega$ satisfies the split property with respect to 
$\Gamma_1$-$\Gamma_2$
\item[(ii)] There are states $\omega_1$, $\omega_2$
on $\caA_{\Gamma_1}$, $\caA_{\Gamma_2}$
such that $\omega$ and $\omega_1\otimes\omega_2$ are quasi-equivalent.
\item[(iii)] There are irreducible representations
$(\caH_1,\pi_1)$, $(\caH_2,\pi_2)$ of $\caA_{\Gamma_1}$,
$\caA_{\Gamma_2}$
and a unit vector $\Omega\in\caH_1\otimes\caH_2$
such that $(\caH_1\otimes\caH_2,\pi_1\otimes\pi_2, \Omega)$
is a GNS-triple of $\omega$.
\end{description}
\end{prop}
Recalling the ``physical" interpretation that ``two states are quasi-equivalent means they are macroscopically the same", the above proposition tells us that
$\omega$ satisfies the split property with respect to 
$\Gamma_1$-$\Gamma_2$
if it is macroscopically product state with respect to the $\Gamma_1$-$\Gamma_2$ cut.

For a pure state with split property, the following holds.
\begin{lem}\label{splitlem5}{\cite{2dSPT}}
Let $\Gamma_1, \Gamma_2 \subset\bbZ^{\nu}$ be non-empty disjoint
 subsets of $\bbZ^\nu$.
Let $\omega$ be a pure state on ${\caA_{\Gamma_1}}\otimes{\caA_{\Gamma_2}}$ and
$\varphi_{{\caA_{\Gamma_1}}}$,  $\varphi_{{\caA_{\Gamma_2}}}$ states on ${\caA_{\Gamma_1}}$,
${\caA_{\Gamma_2}}$ respectively.
Assume that $\omega$ is quasi-equivalent to $\varphi_{{\caA_{\Gamma_1}}}\otimes\varphi_{{\caA_{\Gamma_2}}}$.
Then for any pure states $\psi_{{\caA_{\Gamma_1}}}$, $\psi_{{\caA_{\Gamma_2}}}$
on ${\caA_{\Gamma_1}}$, ${\caA_{\Gamma_2}}$, there are automorphisms 
$\gamma_{{\caA_{\Gamma_1}}}\in\Aut\lmk {\caA_{\Gamma_1}}\rmk$,
$\gamma_{{\caA_{\Gamma_2}}}\in\Aut\lmk {\caA_{\Gamma_2}}\rmk$ and
a unitary $u\in\caU\lmk{\caA_{\Gamma_1}}\otimes{\caA_{\Gamma_2}}\rmk$
such that 
\begin{align}
\omega=\lmk
\lmk
\psi_{{\caA_{\Gamma_1}}}\circ\gamma_{{\caA_{\Gamma_1}}}\rmk\otimes\lmk
\psi_{{\caA_{\Gamma_2}}}\circ\gamma_{{\caA_{\Gamma_2}}}
\rmk
\rmk\circ\Ad(u).
\end{align}
If $\psi_{{\caA_{\Gamma_1}}}$ and $\varphi_{{\caA_{\Gamma_1}}}$ are quasi-equivalent,
then we may set $\gamma_{{\caA_{\Gamma_1}}}=\id_{{\caA_{\Gamma_1}}}$.
\end{lem}

\subsection{Automorphic equivalence}\label{autosec}
For the investigation of the equivalence relation $\sim$, automorphic equivalence
on smooth path of interactions plays an important role.
Automorphic equivalence started as Hasting's adiabatic Lemma \cite{h1}\cite{hw}.
After seminal mathematical development in \cite{bmns} \cite{nsy},
it is now a strong mathematical machinery to analyze 
gapped state phases.
It basically says the following :  given a smooth path of gapped local Hamiltonians,
there is a {\it concrete} smooth path of automorphisms
interpolating ground state spaces.
Original Hastings's adiabatic Lemma was about finite systems.
In the framework given by \cite{bmns} \cite{nsy},
 the existence of gap for {\it  finite} systems is assumed as well, and 
the thermodynamic limit taken afterwords.
An infinite version, assuming the gap in the {\it infinite} system
was derived in \cite{mo}.

In this subsection, we review  Hasting's original adiabatic Lemma,
and introduce its infinite dimensional version.
Let 
\[
H: [0,1]\ni s\to H(s)\in\Mat_n
\]
be a smooth path of self-adjoint matrices of size $n$.
In physical settings this $H(s)$ is a Hamiltonian.
We consider the situation that there is a gap
between the lowest eigenvalue of $H(s)$
and the rest of the spectrum, which is uniform in $s$.
What we are interested in is the lowest eigenvalue eigenspace
of the Hamiltonian i.e, the ground state space.
For each parameter $s\in [0,1]$,
let $E(s)$ be the lowest eigenvalue of $H(s)$
and $P(s)$ be the spectral projection of $H(s)$
corresponding to the eigenvalue $E(s)$.
As we saw in subsection \ref{quantumspinsec},
a ground state of $H(s)$
is a positive linear functional on $\Mat_n$
with support under $P(s)$.
We assume that there is a constant $\gamma>0$
such that 
$\sigma(H(s))\setminus \{E(s)\}\subset [E(s)+\gamma,\infty)$
for any $s\in[0,1]$.
Here $\sigma(H(s))$ indicates the spectrum of $H(s)$.
Namely, we require that there is no cross-over of the ground state energy along the path.
The Hastings adiabatic Lemma is that there is a path of unitaries
following an {\it explicit differential equation}, which
maps $P(0)$, to $P(s)$.
\begin{thm}[Hastings\cite{h1}]\label{Hastings}
In the setting given above,
there is a smooth path of unitaries $U:[0,1]\to \Mat_n$
such that
\begin{align*}
\begin{split}
&P(s)=U(s) P(0) U(s)^*,\quad s\in[0,1],\\
&\frac{d}{ds} U(s)
=i D(s) U(s),\quad U(0)=\unit,\\
&D(s):=
\int_{-\infty}^\infty dt W_\gamma(t) e^{it H(s)} H'(s) e^{-it H(s)}.
\end{split}
\end{align*}
Here, $W_\gamma$ is a $L^1$-function on $\bbR$
with Fourier transformation 
satisfying
\begin{align}\label{fou}
\hat W_\gamma(k)=-\frac{i}{\sqrt{2\pi} k},\quad
k\notin (-\gamma,\gamma).
\end{align}
\end{thm}
This theorem says that there is a smooth path of unitaries
which maps ground state space $P(0)$ to $P(s)$.
An important point for application
is that   the path of unitaries is given by quite a concrete differential equation.
For the choice of $W_\gamma$,
there is some freedom.
The proof goes as follows.
Let $s_0\in [0,1]$ and $\Gamma$ be a contour on $\bbC$
having $E(s)$ inside and 
and $\sigma(H(s))\setminus \{E(s)\}$ outside
for $s$ near $s_0$.
Then our $P(s)$ can be represented as
\begin{align*}
P(s):=
\frac{1}{2\pi i}
\int_\Gamma dz (z-H(s))^{-1}
\end{align*}
near $s_0$.
Now we derive a differential equation satisfied by $P(s)$.
In order to do so, we first simply differentiate the above representation of
$P(s)$.
We then obtain
\begin{align}\label{dp}
\frac{d}{ds} P(s)
=\frac{1}{2\pi i}
\int_\Gamma dz (z-H(s))^{-1}H'(s)(z-H(s))^{-1}.
\end{align}
Now we make a very simple but key observation about a
path of orthogonal projetions. That is
\begin{align}\label{ppp}
P(s) P'(s) P(s)=(1-P(s)) P'(s) (1-P(s))=0.
\end{align}
This can be obtained by differentiating $P(s)^2=P(s)$,
\begin{align*}
P(S)P'(s)+P'(s)P(s)=P'(s).
\end{align*}
Multiplying $P(s)$ (resp. $1-P(s)$) from left and right,
we obtain (\ref{ppp}).
Therefore, we can rewrite (\ref{dp}), picking up off-diagonal
terms:
\begin{align*}
\begin{split}
&\frac{d}{ds} P(s)\\
&=\frac{1}{2\pi i}
\int_\Gamma dz
\lmk
\begin{gathered}
 P(s)(z-H(s))^{-1}H'(s)(z-H(s))^{-1}\lmk 1-P(s)\rmk\\
+(1-P(s))(z-H(s))^{-1}H'(s)(z-H(s))^{-1} P(s)
\end{gathered}
\rmk\\
&=
P(s)
H'(s)\frac{1}{2\pi i} \lmk \int_\Gamma dz(z-E(s))^{-1}(z-H(s))^{-1}\lmk 1-P(s)\rmk\rmk\\
&\quad +\frac{1}{2\pi i}
\lmk \int_\Gamma dz
(1-P(s))(z-E(s))^{-1}(z-H(s))^{-1}\rmk H'(s) P(s)\\
&=
P(s)
H'(s) (E(s)-H(s))^{-1}\lmk 1-P(s)\rmk +
(1-P(s))(E(s)-H(s))^{-1}H'(s) P(s).
\end{split}
\end{align*}
In the last line we used the fact that 
$\Gamma$
has $E(s)$ inside and 
and $\sigma(H(s))\setminus \{E(s)\}$ outside.
 On the other hand, using the property of the Fourier transform (\ref{fou}),
 the last term can be written in terms of $D(s)$ as
 \begin{align*}
& P(s)
H'(s) (E(s)-H(s))^{-1}\lmk 1-P(s)\rmk +
(1-P(s))(E(s)-H(s))^{-1}H'(s) P(s)\\
&=i\left[D(s), P(s)\right].
 \end{align*}
 Hence $P(s)$ satisfies the differential equation
 \begin{align*}
 \frac{d}{ds} P(s)=i\left[D(s), P(s)\right].
 \end{align*}
 This proves the Theorem.
 
Now we would like to apply this Theorem to the Heisenberg dynamics given by time-dependent interactions. 
First, let us introduce paths of interactions.
\begin{assum}\label{assump}
Let  $\Phi (\cdot~ ; s) : \mathfrak{S}_{\mathbb{Z}^\nu } \to \mathcal{A}_{\rm loc}$ be a family of uniformly bounded, finite range interactions parameterized by $s\in [0,1]$. We assume the following:
\begin{description}
\item[(i)]
For each $X\in{\mathfrak S}_{\bbZ^\nu}$, the map
$[0,1]\ni s\to \Phi(X;s)\in\caA_{X}$ is continuous and piecewise $C^1$.
We denote by $\dot{\Phi}(X;s)$ 
the corresponding derivatives.
The interaction obtained by differentiation is denoted by $\dot\Phi(s)$, for each $s\in[0,1]$.
\item[(ii)]
There is a number $R\in\nan$
such that $X \in {\mathfrak S}_{\bbZ^\nu}$ and $\diam{X}\ge R$ imply $\Phi(X;s)=0$, for all $s\in[0,1]$.
\item[(iii)] Interactions are bounded as follows
\begin{align*}
\sup_{s\in[0,1]}\sup_{X\in {\mathfrak S}_{\bbZ^\nu}}
\lmk
\lV
\Phi\lmk X;s\rmk
\rV+|X|\lV
\dot{\Phi} \lmk X;s\rmk
\rV
\rmk<\infty.
\end{align*}
\item[(iv)]
Setting
\begin{align*}
b(\varepsilon):=\sup_{Z\in{\mathfrak S}_{\bbZ^\nu}}
\sup_{s,s_0 \in[0,1],0<| s-s_0|<\varepsilon}
\lV
\frac{\Phi(Z;s)-\Phi(Z;s_0)}{s-s_0}-\dot{\Phi}(Z;s_0)
\rV
\end{align*}
for each $\varepsilon>0$, we have
$\lim_{\varepsilon\to 0} b(\varepsilon)=0$.

\end{description}
\end{assum}

We would like to apply Theorem \ref{Hastings} to path of local Hamiltonians 
 satisfying Assumption \ref{assump}.
For each $s\in[0,1]$ we have local Hamiltonians
\begin{align*}
H_{\Lambda,\Phi}(s):=\sum_{Z\subset\Lambda}\Phi(Z,s),\quad s\in[0,1],\quad \Lambda\in{\mathfrak S}_{\bbZ^2}.
\end{align*}
In addition to Assumption \ref{assump}, we 
need to require a {\it uniform} gap in the spectrum of these local Hamiltonians.
\begin{assum}\label{thermo}
\begin{description}
\item[(i)]There is $\gamma>0$
such that 
\begin{align*}
\sigma\lmk H_{\Lambda,\Phi}(s)\rmk\setminus
\left\{
\inf \sigma\lmk H_{\Lambda,\Phi}(s)\rmk
\right\}
\subset
\left[
\inf \sigma\lmk H_{\Lambda,\Phi}(s)\rmk+\gamma,\infty
\right),\quad \Lambda\in{\mathfrak S},\quad s\in[0,1].
\end{align*}
\item[(ii)]
For each $s\in[0,1]$, there exists a unique $\tau_{\Phi(s)}$-ground state
$\omega_{\Phi(s)}$. 

\end{description}
\end{assum}
Applying Theorem \ref{Hastings} to the path of local Hamiltonians $H_{\Lambda,\Phi}(t)$,
we obtain a path of unitaries $V_\Lambda(s)$ in $\caA_{\Lambda}$  given by a differential equation
\begin{align}\label{vh}
\frac{d}{ds} V_\Lambda(s)
=i \lmk \int_{-\infty}^\infty dt W_\gamma(t) e^{it H_{\Lambda,\Phi}(s)}H_{\Lambda,\dot{\Phi}}(s) e^{-it H_{\Lambda,\Phi}(s)} \rmk V_\Lambda(s),
\end{align}
with $V(0)=\unit$,
interpolating ground state projections.
Because
\[
 e^{it H_{\Lambda,\Phi}(s)}H_{\Lambda,\dot{\Phi}}(s) e^{-it H_{\Lambda,\Phi}(s)}
 =\sum_{Z\subset\Lambda}\lmk  e^{it H_{\Lambda,\Phi}(s)} \dot\Phi(Z,t) e^{-it H_{\Lambda,\Phi}(s)}\rmk,\quad t\in[0,1],
\]
is the summation of ``almost local" terms thanks to the Lieb-Robinson bound \cite{nsy},
it can be written in terms of interaction,
not finite range, but satisfying some locality.
As a result, the following theorem is obtained.
\begin{thm}[\cite{bmns}\cite{nsy}]\label{autonsy}
Let $\Phi$ be a path of interactions satisfying Assumption \ref{assump} and Assumption \ref{thermo}.
Then the thermodynamic limit
\[
\alpha_s(A):=\lim_{\Lambda\to \bbZ^\nu} \Ad\lmk V_{\Lambda}(s)\rmk(A),\quad
A\in\caA_{\bbZ^\nu}
\]
exists and defines a strongly continuous family of automorphisms, satisfying
\[
\omega_{\Phi(s)}=\omega_{\Phi(0)}\circ \alpha_s,\quad s\in[0,1].
\]
\end{thm}
Note that in Assumption \ref{thermo}, 
we assume the existence of the uniform gap in finite systems.
This is against our policy that we should start from infinite systems directly.
Motivated by this, we introduced an infinite-version of Theorem \ref{autonsy}.
There, we assume the gap in the bulk. We now explain about this infinite version.

For $\Gamma\subset\bbZ^{\nu}$, we denote by  $\Pi_{\Gamma}:\caA\to \caA_{\Gamma}$ the conditional expectation with respect to the trace state $\tau$, i.e., a linear map satisfying
\begin{align*}
\Pi_\Gamma(a\otimes b)
=a\tau(b),\quad a\in \caA_{\Gamma},\quad b\in\caA_{\Gamma^c}.
\end{align*}
For each $N\in \bbN$, recall that  $\Lambda_N:=[-N,N]^\nu\cap \bbZ^\nu$.
Let us consider the following subset of $\caA$.
\begin{defn}
Let $f:(0,\infty)\to (0,\infty)$ be a continuous decreasing function 
with $\lim_{t\to\infty}f(t)=0$.
For each $A\in\caA$, let
\begin{align}
\lV A\rV_f:=\lV A\rV
+ \sup_{N\in \nan}\lmk\frac{\lV
A-\bbE_{\Lambda_N}(A)
\rV}
{f(N)}
\rmk.
\end{align}
We denote by $\caD_f$ the set of all $A\in\caA$ such that
$\lV A\rV_f<\infty$.
\end{defn}
The set $\caD_f$ is a $*$-algebra which is a Banach space with respect to 
the norm $\lV\cdot\rV_f$.
Instead of  Assumption \ref{thermo}, we assume the following.
\begin{assum}\label{infautoassump}
\begin{description}
\item[(i)] For each $s\in[0,1]$, there exists a unique $\tau_{\Phi(s)}$-ground state
$\omega_{\Phi(s)}$. 
\item[(ii)] 
There exists a $\gamma>0$ such that
\begin{align}\label{gapeq}
-i\omega_{\Phi(s)}\lmk A^*{\delta_{\Phi(s)}}\lmk A\rmk\rmk\ge \gamma
\omega_{\Phi(s)}(A^*A),\quad \text{for all }A\in \caD({\delta_{\Phi(s)}}) \;\text{with} \; \omega_{\Phi(s)}(A)=0,
\end{align}
 for
all $s\in[0,1]$.
\item[(iii)]
There exists $0<\beta<1$ satisfying the following:
Set $\zeta(t):=e^{-t^{ \beta}}$.
Then for each $A\in D_\zeta$, 
\[
[0,1]\ni s\mapsto \omega_{\Phi(s)}(A)
\]
 is differentiable with respect to $s$, and there is a constant
$C_\zeta$ such that:
\begin{align}\label{dcon}
\lv
\frac{d{\omega_{\Phi (s)}}(A)}{ds}
\rv
\le C_\zeta\lV A\rV_\zeta,
\end{align}
for any $A\in D_\zeta$.
\end{description}
\end{assum}
One can show that  Assumption \ref{thermo} implies Assumption \ref{infautoassump}.
Namely, Assumption \ref{infautoassump} is a weaker condition.
Nevertheless, we obtained the following.
\begin{thm}[\cite{mo}]\label{autoinf}
Let $\Phi$ be a path of interactions satisfying Assumption \ref{assump} and Assumption \ref{infautoassump}.
Then the conclusion of Theorem \ref{autonsy} holds.
\end{thm}
Although the result of Theorem \ref{autoinf} is analogous to that of 
Theorem \ref{autonsy}, there is a crucial difference about the proof.
The proof of Hasting's theorem fully uses the fact that they are working on operators
{\it on the same Hilbert space}.
In particular, the property $P(s)^2=P(s)$
of the orthogonal projection was a key input for
the analysis.
On the other hand, for infinite version, the GNS representation of
$\omega_{\Phi(s)}$ for different value of $s$ s are mutually singular in general.
This requires a new idea.

In order to explain the point of the proof of Theorem \ref{autoinf},
we now prove Hastings's Lemma following the strategy of Theorem \ref{autoinf}.
Assume that $P(s)$ in the setting of Theorem \ref{Hastings}
to be of one rank (namely, each $H(s)$ has a unique ground state), and set $\omega_s:=\Tr P(s)(\cdot)$, and $\alpha_s:=\Ad(U(s)^*)$.
It suffices to show that 
\[
\frac{d}{ds} \omega_{s}\alpha_s^{-1}(A)
=0,
\]
for any $A$.
From the definition of $\alpha_s$,
this means
\begin{align}\label{7m}
0=\dot\omega_s\circ\alpha_s^{-1}(A)
+\omega_s\lmk
i\left[
D(s), \alpha_s^{-1}(A)
\right]
\rmk.
\end{align}
The question is what is $\dot\omega_s$.
To consider it,
we recall the fact that $\omega_s$ is a ground state of $H(s)$.
This means it is invariant under $i[H(s),\cdot]$.
Therefore, we have
\[
\omega_s\circ \lmk i[H(s),\cdot] \rmk=0.
\]
Differentiating this, we obtain
\begin{align}\label{6m}
\dot\omega_s\lmk i[H(s),A]\rmk
+\omega_s\lmk i[\dot H(s),A]\rmk =0.
\end{align}
If $i[H(s),\cdot]$ was invertible, this would give us a representation
of $\dot\omega_s$ in terms of $\omega_s$.
Unfortunately, $i[H(s),\cdot]$ is not invertible.
In order to solve the issue, an operation called Hasting's operator
$I_s$ plays an important role.
Let 
$w_\gamma$ be a non-negative $L^1$-function 
such that
\[
\supp\hat w_\gamma\subset (-\gamma,\gamma),\quad\text{and}\quad
\int dt w_\gamma(t)=1.
\]
In fact we may take this $\omega_\gamma$ and $W_\gamma$ (in Theorem \ref{Hastings}), as even and odd functions respectively
so that vanishing at infinity and 
\[
\frac{d}{dt} W_\gamma(t)=-\omega_{\gamma}(t),\quad t\in[0,\infty).
\]
(See \cite{nsy} section 6.)
Set
\[
I_s\lmk A\rmk
:=\int dt w_\gamma(t) \Ad\lmk e^{iH(s)t} \rmk(A).
\]
Because of the support condition of $\hat w_\gamma$,
we have
\begin{align}\label{3m}
\omega_s\lmk B^* I_s(A)\rmk
=\omega_s(B^*)\omega_s(A),
\end{align}
for any $A,B$.
Indeed, we have
\begin{align*}
&\omega_s\lmk B^* I_s(A)\rmk
=\int dt w_\gamma(t) 
\Tr\lmk P(s)
B^* e^{i\lmk H(s)-E(s\rmk t} A
\rmk\\
&=\Tr\lmk P(s)B^*
\lmk
\hat w_\gamma( -H(s)+E(s))
\rmk
 A\rmk.
\end{align*}
Now, note that because of the spectral gap
$\sigma(H(s))\setminus \{E(s)\}\subset [E(s)+\gamma,\infty)$,
and the support property of $\hat w_{\gamma}$,
$\supp\hat w_\gamma\subset (-\gamma,\gamma)$,
$\hat w_\gamma( H(s)-E(s))$ is proportional to $P(s)$,$\hat w_\gamma( H(s)-E(s))=cP(s)$.
Setting $A=B=\unit$, we see that $c=1$.
Hence we obtain (\ref{3m}).
The same argument allows us to prove (\ref{3m}) (for $\omega_s=\omega_{\Phi(s)}$) in infinite system, on GNS Hilbert space
of $\omega_{\Phi(s)}$,
using the bulk Hamiltonian $H_{\omega_{\Phi(s)},\Phi(s)}$ instead of $H(s)$ 
and an orthogonal projection onto $\bbC\Omega_{\Omega_{\Phi(s)}}$
instead of $P(s)$. The point is everything in the argument is well-defined even in infinite setting.

In particular, from (\ref{3m}), with $A$ replaced by $A^*$ and $B$ by $\lmk I_s(A)-  \omega_s(A)\rmk^*$,
 we have
\begin{align*}
\omega_s\lmk  \lmk I_s(A)-  \omega_s(A)\rmk \lmk I_s(A)-  \omega_s(A)\rmk^* \rmk=0.
\end{align*}
This means 
\begin{align}\label{5m}
\lmk\unit-P(s)\rmk\lmk I_s(A)-  \omega_s(A)\rmk
=\lmk I_s(A)-  \omega_s(A)\rmk.
\end{align}
Now differentiating (\ref{3m}), we have
\begin{align*}
\dot \omega_s\lmk B^* I_s(A)\rmk +\omega_s\lmk B^*\dot  I_s(A)\rmk
=\dot \omega_s(B^*)\omega_s(A) +\omega_s(B^*)\dot \omega_s(A).
\end{align*}
Substituting $B=\unit-P(s)$
to this equation, and recalling (\ref{5m}), we obtain 
\begin{align*}
\dot \omega_s(I_s(A))= \dot \omega_s \lmk I_s(A)-  \omega_s(A)\rmk 
=
\dot \omega_s\lmk \unit-P(s) \lmk I_s(A)-  \omega_s(A)\rmk \rmk 
=0.
\end{align*}
Here we used the clear fact that $\dot\omega_s(\unit)=0$.
The same argument allows us to prove $\dot \omega_{\Phi(s)}( I_s(A))=0$,
for infinite system if we replace the role of $\unit-P(s)$
by approximate identity of the left kernel 
associated to $\omega_{\Phi(s)}$.

Hence we have obtained $\dot \omega_s(I_s(A))=0$.
A good thing about it is that the difference between $I_s(A)$
and $A$
is given by
\begin{align*}
A-I_s(A)
=-\int dt w_\gamma(t) 
\int_{0}^t du\; i\left[H(s), \Ad\lmk e^{iH(s)u} \rmk(A)\right],
\end{align*}
and $i[H(s),\cdot]$ shows up.
Substituting this to $\dot\omega_s$,  and recalling $\dot \omega_s(I_s(A))=0$,
we obtain
\begin{align*}
\begin{split}
&\dot \omega_s\lmk A\rmk=
\dot \omega_s\lmk A-I_s(A)\rmk \\
&=-\int dt w_\gamma(t) 
\int_{0}^t du\; 
\dot \omega_s\lmk i\left[H(s), \Ad\lmk e^{iH(s)u} \rmk(A)\right]\rmk.
\end{split}
\end{align*}
Now we apply (\ref{6m}) to the last term.
Then we obtain
\begin{align*}
&\dot \omega_s\lmk A\rmk=
\int dt w_\gamma(t) 
\int_{0}^t du\; 
\omega_s\lmk i\left[\dot H(s), \Ad\lmk e^{iH(s)u} \rmk(A)\right]\rmk\\
&=
\int dt W_\gamma(t) 
\omega_s\lmk i\left[\dot H(s), \Ad\lmk e^{iH(s)t} \rmk(A)\right]\rmk\\
&=
-\int dt W_\gamma(t) 
\omega_s\lmk i\left[\Ad\lmk e^{iH(s)t}\lmk \dot H(s)\rmk \rmk(A)\right]\rmk\\
&=
-\omega_s\lmk
i\left[
D(s), A
\right]
\rmk.
\end{align*}
Now we use this formula to prove the desired equation (\ref{7m}).
Indeed we have
\begin{align*}
&\dot\omega_s\circ\alpha_s^{-1}(A)
+\omega_s\lmk
i\left[
D(s), \alpha_s^{-1}(A)
\right]
\rmk\\
&=
-\omega_s\lmk
i\left[
D(s), \alpha_s^{-1}(A)
\right]
\rmk
+\omega_s\lmk
i\left[
D(s), \alpha_s^{-1}(A)
\right]
\rmk=0.
\end{align*}
Hence we have proven the Theorem.
The important part here is that we use only objects, which has its infinite version.
For example the commutator $[H(s),\cdot]$, still makes sense
although $H(s)$ itself does not have any meaning in infinite system.

As a dynamics given by (time-dependent) local interactions,
the automorphism in Theorem \ref{autoinf}
satisfies some {\it quasi local} property.
For example, it satisfies the Lieb-Robinson bound.
Among such properties, what is important in analysis of SPT-phases is what we call
{\it factorization property}.

Factorization property in $1$-dimensional system is 
the following.
\begin{prop}\cite{TRI}
Let $\nu=1$.
The automorphism $\alpha_s$ in Theorem \ref{autoinf} satisfies a factorization property such that
\[
\alpha_s=\inn\circ\lmk\alpha_L\otimes\alpha_R\rmk,
\]
where $\alpha_L$ is an automorphism on $\caA_{(-\infty,-1]}$,
$\alpha_R$ an automorphism on $\caA_{[0,\infty)}$.
\end{prop}
A two-dimensional version we use is as follows.
For $0<\theta<\frac\pi 2$, a cone $C_\theta$ is defined by
\begin{align}\label{ctdef}
C_\theta:=
\left\{
(x,y)\in\bbZ^2\mid
|y|\le \tan \theta\cdot |x|
\right\}.
\end{align}
Furthermore, $H_L$, $H_R$,$H_U$, $H_D$  denotes half left/right and upper/lower planes,
and $C_{\theta,L}:=C_\theta\cap H_L$, $C_{\theta,R}:=C_\theta\cap H_R$,
$C_{\theta,U}:=C_\theta\cap H_U$, $C_{\theta,D}:=C_\theta\cap H_D$.
\begin{prop}
Let $\nu=2$.
For any $0<\theta<\frac\pi 2$,
there is $\alpha_L\in\Aut\caA_{H_L}$, $\alpha_R\in\Aut\caA_{H_R}$,
and $\Theta\in\Aut \caA_{\lmk C_\theta\rmk^c}$ such that
\begin{align}\label{tfac}
\alpha_s=\inn\circ\lmk\alpha_L\otimes\alpha_R\rmk\circ\Theta.
\end{align}
\end{prop}
Actually, $\alpha_s$ can be cut in  many directions simultaneously.
Using such properties, we can also show the following.
\begin{lem}\label{alex}
For any 
$0<\theta'<\theta<\frac\pi 2$,
and $\eta_L'\in \Aut \caA_{C_{\theta', L}}$, $\eta_R'\in \Aut \caA_{C_{\theta', R}}$,
there are automorphisms
$\eta_L\in \Aut \caA_{C_{\theta, L}}$, $\eta_R\in \Aut \caA_{C_{\theta, R}}$
such that 
\[
\alpha_s\circ \lmk \eta_L'\otimes\eta_R'\rmk\circ \alpha_s^{-1}=\inn\lmk
\eta_L\otimes\eta_R
\rmk.
\]
\end{lem}
Factorization property is quite simple but strong analytical property.
It also enables us to show that the existence of a non-trivial super selection sector results in
long-range entanglement \cite{NO}.

\section{SPT-phases in $1$-dimensional systems}\label{1dimsec}
In this section, we study SPT-phases  in $1$-dimensional systems.
We will see that there is a $H^2(G,\Uo)$-valued index for SPT with
on-site symmetry, and $\bbZ_2$-valued index for SPT with reflection symmetry.
They generalize the indices considered for Matrix product states.
\subsection{$H^{2}(G,\Uo)$ and projective representations}\label{onsiteone}

For SPT-phases in one dimension with on-site symmetry,
physicists conjectured that there is a $H^{2}(G,\Uo)$-valued invariant
of the classification.
In this subsection, we recall the definition of $H^{2}(G,\Uo)$ and its relation to projective representations.

A $2$-cochain $C^2(G,\Uo)$ is a map from 
$G^{\times 2}$ to $\Uo$.
It is said to be a $2$-cocycle if the relation 
\begin{align*}
\sigma(gh,k) \sigma(g,h) 
=
\sigma(g,hk) \sigma(h,k)
\end{align*}
holds.
The set of all $2$-cocycles are denoted by $Z^2(G,\Uo)$.
We introduce an equivalence relation to $Z^2(G,\Uo)$ by
$\sigma'\simeq \sigma$ if there exists $c:G\mapsto \bbT$ such that
\begin{align*}
\sigma'(g,h)=
\sigma(g,h) \frac{c(gh) }{ c(g)  c(h) }.
\end{align*}
The equivalence classes forms a group, which is the second group cohomology 
$H^2(G,{U(1)})$.

The second group cohomology $H^2(G,{U(1)})$
shows up naturally, when we consider projective representations.
A projective representation of $G$ on a Hilbert space $\caH$
is a map from $G$ to a unitary operators $\caU(\caH)$
on $\caH$
satisfying
\begin{align*}
u(g)u(h)=\sigma(g,h) u(gh),\quad g,h\in G.
\end{align*}
Here, $\sigma(g,h)$ is a $\Uo$-phase.
Using the associativity of $u$, we can show that $\sigma$
is a $2$-cocycle.
Indeed,
we have
\begin{align*}
u(g)u(h)u(k)=\sigma(g,h) u(gh)u(k)
=\sigma(g,h)\sigma(gh,k) u(ghk),
\end{align*}
and
\begin{align*}
u(g)u(h)u(k)
=\sigma(h,k) u(g)u(hk)
=\sigma(h,k) \sigma(g,hk) u(ghk).
\end{align*}
Comparing these equations, we obtain
$2$-cocycle relation for $\sigma$.

Note that given a group action $G$ on $\caB(\caH)$, for a Hilbert space $\caH$,
we get some projective representation.
Let
$\gamma :G\to \Aut\lmk\caB(\caH)\rmk$
be a group homomorphism.
Then 
by Theorem \ref{wigner}, there is a unitary $u_g$ on $\caH$
such that
\begin{align*}
\Ad(u_g)(x)=\gamma_g(x),\quad x\in \caB(\caH).
\end{align*}
Note that there is some $\Uo$-ambiguity about the choice of $u_g$:
for any map $c: G\to \Uo$,
we also have
\begin{align}\label{cu}
\Ad(c_gu_g)(x)=\gamma_g(x),\quad x\in \caB(\caH).
\end{align}
This $u : G\to \caU\lmk\caH\rmk$ is a projective group representation. 
Indeed, we have
\begin{align*}
\Ad(u_g u_h)(x)=\gamma_g\gamma_h(x)
=\gamma_{gh}(x)=\Ad(u_{gh})(x),\quad x\in \caB(\caH).
\end{align*}
This means
there is some $\sigma(g,h)\in\Uo(1)$
such that 
\begin{align*}
u_g u_h
=\sigma(g,h) u_{gh}.
\end{align*}
Namely, $u$ is a projective representation, and hence
$\sigma$ is a $2$-cocycle.
The $\Uo$-ambiguity of the choice $c$
(\ref{cu}) results in ambiguity of the $2$-cocycle 
$\sigma$. 
If we choose $c_g u_g$ instead of $u_g$,
then the $2$-cocycle obtained is
\begin{align*}
\tilde \sigma(g,h)=\sigma(g,h) c_gc_h\bar{c_{gh}}.
\end{align*}
However, from this equation, we see that the resulting second cohomology class
does not depend on such a choice, namely we have
\[
\lcm \tilde \sigma\rcm_{H^2(G,\Uo)}=\lcm \sigma\rcm_{H^2(G,\Uo)}.
\]
Hence, given a group action $\gamma$ on $\caB(\caH)$, 
we obtain an associated second cohomology class without ambiguity.

\subsection{Split property of unique gapped ground state in $1$-d}

An important property of the unique gapped ground state in one dimension
is the split property.
Using this property, one can easily define indices, which turns out to be invariants of the classification.
More precisely, the following Theorem is crucial in order to define the SPT-index
in one-dimensional systems. 
\begin{thm}[Matsui '13]\label{matsuisplit}
Let $\omega_{\Phi}$ be a {\it  unique gapped $\alpha^\Phi$-ground state}
for $\Phi\in \caP_{U.G.}$.
Then $\omega_{\Phi}$ satisfies the split property with respect to $\caA_L$ and $\caA_R$, i.e.,
$\pi_{\omega_{\Phi, R}}(\caA_R)''$ is a type I factor.
\end{thm}
Thanks to this theorem, known for {\it any} unique gapped ground state, 
we can classify
\begin{align*}
\caP_{U.G.{{\beta}}}:=
\left\{
\Phi\in \caP_{U.G.}\mid 
\beta\text{-invariant}
\right\}.
\end{align*}
not only the SPT-phases $\caP_{U.G.{{\beta}}}^0$ (\ref{pugb}).
Although it is expected that $\caP_{U.G.}^0= \caP_{U.G.}$ for one dimensional systems,
there is no result proving it.
Therefore, at least at this stage of research,
classification of $\caP_{U.G.}$ is stronger result than
classification of $\caP_{U.G.}^0$.

\subsection{Matrix product state}\label{mps}
Analysis of SPT-phases was initially carried out for matrix product states.
In this section, we introduce matrix product states.

In general, even constructing an example of gapped Hamiltonians is a highly non-trivial problem.
The non-commutativity of quantum systems makes proving the existence of spectral gap
 a very hard problem
. 
This is the case already for local Hamiltonians.
If the interaction terms $\Phi(X)$
in local Hamiltonians
\begin{align*}
\lmk H_{\Phi}\rmk_{\Lambda}:=\sum_{X\subset{\Lambda}}\Phi(X).
\end{align*}
 mutually commute, we can consider the joint distribution. But typically, they don't commute.

Fortunately, for one-dimensional systems, there is a recipe and technics
by Fannes-Nachtergaele-Werner \cite{Fannes:1992vq} and \cite{Nachtergaele:1996vc} to construct a class of Hamiltonians with unique gapped ground state.
This is sometimes called
Matrix product state (MPS) Hamiltonians.
In this section, we explain this recipe.

Recently there was great progress in higher dimensions. 
The technics from  \cite{Fannes:1992vq} and \cite{Nachtergaele:1996vc}  has turned out to be
useful even in higher dimensional systems.
See \cite{allny} \cite{ln} and \cite{lsw}. 

Throughout this section,
we fix some orthonormal basis of $\cc^d$, $\{\psi_{\mu}\}_{\mu=1}^d$,
and denote by $e_{\mu\nu}$, $\mu,\nu=1,\ldots d$
the system of matrix unit corresponding to it.
The interactions given by \cite{Fannes:1992vq} is a type of interactions so called parent Hamiltonians.
A parent Hamiltonian is made out of a sequence of subspaces of $\bigotimes_{i=0}^{l-1} \bbC^d$, $l\in\nan$, which satisfies
the condition called the intersection property.
\begin{defn}
We say that a sequence of subspaces $\{\caD_N\}_{N\in \nan}$, $\caD_N\subset \bigotimes_{i=0}^{N-1}\cc^d$, $N\in\nan$, satisfies the intersection property,
if there exists an $m\in\nan$, such that the relation 
\begin{equation*}
\caD_N = \bigcap_{x=0}^{N-m} (\bbC^{d})^{\otimes x}\otimes \caD_{m}\otimes (\bbC^{d})^{\otimes N-m-x}, 
\end{equation*} holds for all $N\ge m$.
\end{defn}
Let $\caD_N\subset \bigotimes_{i=0}^{N-1}\cc^d$, $N\in\nan$
be subspaces satisfying the intersection property as above.
Let $h_m$ be the orthogonal projection onto the orthogonal complement of $\caD_m$ in $\otimes_{i=0}^{m-1}\cc^d$.
We set
\begin{align*}
\Phi(X):=
\left\{
\begin{gathered}
\tau_x\lmk h_m\rmk ,
\quad \text{if}\; X=[x,x+m-1]\\
0,\quad \text{otherwise}
\end{gathered}
\right..
\end{align*}
Here $\tau_x$ denotes the space translation on $\bbZ$.
Then, by the intersection property, we see that
\[
\ker \lmk \lmk H_\Phi\rmk_{[0,N-1]}\rmk=\caD_N
\]
for all $N\ge m$.
Hamiltonians given in this way are called parent Hamiltonians.
Parent Hamiltonians are nice in the sense that
1.we know that the lowest eigenvalues of local Hamiltonians are zero, and
2.we know how the ground state
spaces of local Hamiltonians, i.e., 
$
\ker \lmk \lmk H_\Phi\rmk_{[0,N-1]}\rmk=\caD_N
$
 look like.
Now we would like to construct a sequence of spaces 1.
satisfying the intersection property, and 2.
 the corresponding parent local Hamiltonian
$\ker \lmk \lmk H_\Phi\rmk_{[0,N-1]}\rmk$ is gapped with respect to the open boundary conditions.
This is carried out by a linear map determined by a $d$-tuple of matrices.
First, we prepare an $d$-tuple of $k\times k$ matrices
$\vv:=(v_1,\ldots,v_d)$ and $m\in\nan$.
Define a subspace $\caG_{m,\vv}$
of 
$\bigotimes_{i=0}^{m-1} \bbC^d$ by the range of 
the following map $\Gamma_{m,\vv}:\Mat_k
\to \bigotimes_{i=0}^{m-1} \bbC^d$,
\begin{align*}
\Gamma_{m,\vv}\lmk X\rmk
=\sum_{\mu_0,\ldots,\mu_{m-1}\in \{1,\cdots,d\}}\lmk\Tr X \lmk
v_{\mu_0}v_{\mu_1}\cdots v_{\mu_{m-1}}
\rmk^*\rmk\bigotimes_{i=0}^{m-1}\psi_{\mu_i},\quad X\in\Mat_k.
\end{align*}
Let $h_{m,\vv}$ be the orthogonal projection onto $\caG_{m,\vv}^{\perp}$ in $\otimes_{i=0}^{m-1}\cc^d$.
Then we set
\begin{align*}
\Phi_{m,\vv}(X):=
\left\{
\begin{gathered}
\tau_x\lmk h_{m,\vv}\rmk ,
\quad \text{if}\; X=[x,x+m-1]\\
0,\quad \text{otherwise}
\end{gathered}
\right..
\end{align*}
Of course, for a random choice of $\vv$, the spaces
$\{\caG_{N,\vv}\}_N$ does not satisfy the intersection property, and
the local Hamiltonian given by $h_{m,\vv}$ does not have a gap.
We are interested in some sufficient condition for $\vv$ to satisfy these properties.
The sufficient condition which was introduced in \cite{Fannes:1992vq} is 
the following equivalent properties.
\begin{thm}\label{thmprim}
For $\vv := (v_1,\ldots,v_d)\in\Mat_k^{\times d}$, let $T_\vv:\Mat_k\to \Mat_k$ be the completely positive map given by
\begin{align*}
T_\vv(A)=\sum_{i=1}^{d}v_iAv_i^* ,\quad A\in\Mat_k.
\end{align*} 
Assume that the spectral radius $r_\vv$ of $T_\vv$,  is $1$.
The following properties are equivalent.
\begin{enumerate}
\item
There exist a unique {faithful} state $\varphi_\vv$ and a {strictly positive} element
$e_\vv\in\Mat_k$ satisfying 
 \[
\lim_{N\to\infty} T_\vv^N(A)=\varphi_\vv(A) e_\vv,\quad
A\in \Mat_k.
\]
\item The spectrum $\sigma\lmk T_\vv\rmk$ of $T_\vv$ satisfies
{$\sigma\lmk T_\vv\rmk \cap\bbT=\{1\}$.} $1$ is a non degenerate eigenvalue of $T_\vv$. There exist a faithful $T_\vv$-invariant state $\varphi_\vv$
and a strictly positive $T_\vv$-invariant element $e_\vv\in\Mat_k$.
\item
There exists an $m\in\nan$
such that 
{${\mathcal K}_m(\vv)=\Mat_k$}, where
\begin{equation*}
{\mathcal K}_m(\vv) :=\spn\left\{v_{\mu_1}v_{\mu_{2}}\cdots v_{\mu_{m}}\mid
(\mu_1,\mu_2,\ldots,\mu_{l})\subset\{1,\ldots, d\}^{\times m}\right\}.
\end{equation*}

\end{enumerate}
\end{thm}
When $d$-tuple of matrices $\vv$ satisfies these (equivalent) conditions, we say
that it is {\it primitive}.

Given a primitive $d$-tuple $\vv$, we can construct a matrix product state (MPS), which was initially
called a finitely correlated state.
Let $\varphi_\vv$ be the state, and $e_\vv$ the positive element
given in (1) of the above theorem.
Then the formula
\begin{align}
\omega_{\vv}\lmk
\bigotimes_{i=1}^{l}
e_{\mu_i,\nu_i}
\rmk=
\varphi_\vv\lmk v_{\mu_1}\cdots v_{\mu_l} e_{\vv} v_{\nu_l}^*\cdots v_{\nu_1}^*\rmk,\quad
\mu_i,\nu_i\in\SSS,\quad i=1,\ldots,l,\quad l\in\nan,
\end{align}
defines a translationally invariant state on $\caA_{\bbZ}$.
This is the matrix product state generated by $\vv$.

Here is the Theorem of Fannes-Nachtergaele-Werner :
\begin{thm}[\cite{Fannes:1992vq}]\label{fnw}
Assume that $\Gamma_{m-1,\vv}$ is injective.
Then ${\Phi_{m,\vv}}$ has a unique gapped ground state $\omega_{\vv}$.
\end{thm}
\begin{rem}
In fact, the uniqueness of the ground state in our sense (Definition \ref{groundstatedef}
) requires some additional analysis. See
\cite{Matsui1}, 
\cite{Ogata3}.
\end{rem}
Hence, we got a class of gapped Hamiltonians by the FNW-recipe.
It is nice to have a systematic way of constructing gapped Hamiltonians.
However, the power of MPS is not only that.
If an interaction satisfies some (a bit strong) conditions,
we can show that the ground state is MPS.
\begin{thm}[\cite{Matsui1}(See also \cite{Ogata3} for primitivity)]\label{mat}
Let $\Phi$ be a translation invariant finite range interaction
with a
 unique ground state $\omega_{\Phi}$ on $\caA_{\bbZ}$.
 Suppose that
there exists a constant $c\in\nan$ such that
\begin{align}\label{ubr}
1\le \dim \ker\lmk \lmk H_\Phi\rmk_{[1,N]}\rmk\le c,
\end{align}
for all $N\in\nan$.
Then $\omega_{\Phi}$ is a MPS, namely there is a primitive $\vv$
such that
$\omega_{\Phi}=\omega_{\vv}$.
\end{thm}
This theorem follows from the following Theorem from operator algebra theory.
\begin{thm}[\cite{arv}]\label{arv}
Let $\caH$ be a separable infinite dimensional Hilbert space, and $d\in \nan$.
Let $\Xi : B(\caH)\to B(\caH)$ be a
unital endomorphism of $B(\caH)$
such that $
\lmk \Xi\lmk B(\caH)\rmk\rmk'
$
is isomorphic to $\Mat_d$.
Then there exist $S_i\in B(\caH)$, $i=1,\ldots,d$ such that
\begin{align}\label{eq:cuntz}
S_i^*S_j=\delta_{ij},\quad
\sum_{j=1}^d S_j x S_j^*=\Xi(x),\quad x\in B(\caH).
\end{align}
\end{thm}
Note that $S_{j}$s in the theorem form a representation of Cuntz algebra \cite{BR2}.

Let $\omega_{\Phi}$ be a unique gapped ground state for an interaction $\Phi$ which satisfies the conditions in Theorem \ref{mat}.
By Theorem \ref{matsuisplit},
there is an irreducible representation $(\caK,\pi
)$
of $\caA_R$, quasi-equivalent to the GNS representation of $\omega_\Phi\vert_{\caA_R}$.
Because $\omega_{\Phi}$ is translationally invariant,
the space translation extends uniquely to
an endomorphism on $\pi
(\caA_R)''=\caB(\caK)$.
Applying Theorem \ref{arv} to it, 
we obtain
operators $\cs_\mu\in B(\caH_\omega)$ with $\mu=1,\ldots d$
satisfying the following:
\begin{align}\label{stachi}
\begin{split}
&\cs_{\mu}^*\cs_{\nu}=\delta_{\mu\nu}\unit,\\
&\sum_{\mu=1,\ldots,d} \cs_{\mu} \pi
(A) \cs_{\mu}^*=\pi
\circ\tau_{1}(A),\quad A\in \caA_R .\\
&\pi
\lmk e_{\mu\nu}\otimes\unit_{[1,\infty)}\rmk
=\cs_\mu \cs_\nu^*\quad
\text{for all} \quad \mu,\nu=1,\ldots,d.
\end{split}
\end{align}
Here $ e_{\mu\nu}\otimes\unit_{[1,\infty)}$ indicates an element $e_{\mu\nu}$ in $\caA_{\{0\}}=\Mat_{d}$
embedded into $\caA_R$.
From this, we have
\begin{align}\label{ssss}
\pi
\lmk\bigotimes_{k=0}^{l-1}e_{\mu_k,\nu_k}\rmk
=\cs_{\mu_0}\cdots \cs_{\mu_{l-1}}\cs_{\nu_{l-1}}^*\cdots \cs_{\nu_0}^*,
\end{align}for all $l\in\nan$, $\mu_k,\nu_k=1,\ldots,d$.
Because $\pi$ and the GNS representation of $\omega_\Phi\vert_{\caA_R}$
are quasi-equivalent, there is a density operator on $\caK$
such that
\[
\omega_{\Phi}(A)=\Tr\rho\pi(A),\quad A\in\caA_R.
\]
Because of the first inequality of (\ref{ubr})
and the uniqueness of the ground state, $\omega_{\Phi}$ has to be frustration-free,
namely, it should satisfy $\omega_{\Phi}(\Phi(X))=0$ for any $X\in \mathfrak S_{\bbZ}$.
From this and the second inequality of (\ref{ubr}), one can show that
the support  of $\rho$ is of finite rank. (See \cite{Matsui1}.)
The $n$-tuple $\vv$ generating the finitely correlated state $\omega$, in Theorem \ref{mat} is given as a
restriction of $S_i$s to this range of $\rho$.

Recall that it is believed that any interactions $\Phi_1,\Phi_2$ with unique gapped ground state in
one-dimension are equivalent $\Phi_1\sim\Phi_2$.
Although in general, it is still an open problem,
for the class of interactions satisfying the conditions in  Theorem \ref{mat}, we can say it is true.
\begin{thm}[\cite{Ogata3}]
Let $\Phi_1,\Phi_2$ be uniformly bounded finite range interactions satisfying
the assumptions in Theorem \ref{mat}.
Then we have $\Phi_1\sim\Phi_2$.
Furthermore, they can be connected without breaking the space translation invariance.
\end{thm}


In the physicist's approach on SPT-phases, the uniqueness theorem of
MPS gets important.
In quantum information theory, it is called the {\it fundamental theorem of MPS}.
It is due to Fannes-Nachtergaele-Werner again.
\begin{thm}[\cite{fnwpure}]\label{fnwpure}
Suppose that primitive $\vv\in\Mat_k^{\times d}$ and
$\tilde \vv\in\Mat_{\tilde k}^{\times d}$
give the same state $\omega_{\vv}=\omega_{\tilde\vv}$.
Then $k=\tilde k$ and there is a $c\in\bbT$,
a unitary $k\times k $ matrix $U$,
such that
\[
\tilde v_\mu=c U v_{\mu} U^*,\quad \mu=1,\ldots, d.
\]
\end{thm}

\subsection{$H^{2}(G,\Uo)$-index for MPS with on-site symmetry}\label{mpsonsite}

Now let us consider MPS which are invariant under the on-site symmetry $\beta$ given by a unitary representation $U$ of a finite group $G$.
In \cite{po}, \cite{po2}, \cite{Perez-Garcia2008}
a $H^2(G,\Uo)$-valued index was obtained by the following argument from $\beta$-invariant MPS.
In this section we explain about this.

Let $\omega_{\vv}$ be a $\beta$-invariant MPS
given by a primitive $\vv\in\Mat_k^{\times d}$.
Because $\omega_{\vv}$ is $\beta_g$-invariant,
we can check that $d$-tuple $\tilde \vv(g)$ given by
\[
\tilde v_{\mu}(g):=\sum_{\nu=1}^d U(g)_{\nu\mu} v_\nu,\quad g\in G,\quad
\mu=1,\ldots,d
\]
also gives the same state.
Namely, we have
\[
\omega_{\tilde \vv(g)} =\omega_{\vv}.
\]
Then, by the fundamental theorem of MPS,
there is $c(g)\in\bbT$, a $k\times k$
unitary matrix $u(g)$, such that
\begin{align*}
c(g)u(g) v_\mu u(g)^*
=\tilde v_\mu (g)
=\sum_{\nu=1}^d U(g)_{\nu\mu} v_\nu.
\end{align*}
Now, let us consider $\tilde v_\mu(gh)$.
From above, we have
 \begin{align*}
 \tilde v_\mu(gh)=c(gh)u(gh) v_\mu u(gh)^*.
 \end{align*}
 On the other hand, by the definition of $ \tilde v_\mu(gh)$,
 we have
 \begin{align*}
  &\tilde v_\mu(gh)
  =\sum_{\nu=1}^d U(gh)_{\nu\mu} v_\nu
  =\sum_{\nu,\lambda=1}^d U(g)_{\nu\lambda}
  U(h)_{\lambda\mu} v_\nu\\
&=\sum_{\lambda=1}^d  U(h)_{\lambda\mu}
\tilde v_\lambda(g)
=c(g)u(g) \sum_{\lambda=1}^d  U(h)_{\lambda\mu} v_\lambda u(g)^*\\
&=c(g)c(h) u(g)u(h) v_\mu u(h)^* u(g)^*.
 \end{align*} 
 Hence we obtain
 \begin{align*}
 c(gh)u(gh) v_\mu u(gh)^*
 =c(g)c(h) u(g)u(h) v_\mu u(h)^* u(g)^*,
 \end{align*}
 for all $\mu=1,\ldots, d$.
 Recalling (3) of Theorem \ref{thmprim},
 we then see that $c(gh)=c(g)c(h)$ and
 \begin{align*}
 u(gh) x u(gh)^*
 =u(g)u(h) x u(h)^* u(g)^*,
 \end{align*}
 for all $x\in\Mat_k$.
 This means there is some $\sigma(g,h)\in\Uo$
 such that
 \begin{align*}
 u(g)u(h)=\sigma(g,h) u(gh),\quad g,h\in G.
 \end{align*}
 Hence we obtain a projective representation, and a second cohomology $[\sigma]_{H^2(G,\Uo)}\in H^2(G,\Uo)$.
Here we considered linear representation of $G$ 
but anti-linear ones can be treated analogously.
This is the index obtained in  \cite{po}, \cite{po2}, \cite{Perez-Garcia2008}.

\subsection{$H^{2}(G,\Uo)$-index for SPT with on-site symmetry}

The derivation of $H^2(G,\Uo)$-valued index given in the previous subsection is elegant and nice,
but there is a crucial problem for our use.
That is, it is defined only on a subset of 
$\caP_{U.G.\beta}^0$, i.e., MPS. 
In order to think about classification problem, we definitely need some index which are defined
{\it all} the $\caP_{U.G.\beta}^0$.

Such an index can be derived from split property.
In fact, the existence of a projective representation associated to state with split property
is automatic as we saw in subsection \ref{onsiteone}, and it
 has been noticed for operator algebraist from long time ago \cite{Matsui2}.
Thanks to Matsui's theorem, we know that our unique ground state $\omega_\Phi$
for $\beta$-invariant $\Phi\in \caP_{U.G.}$
satisfies the split property.
From Theorem \ref{matsuisplit},
and Proposition \ref{splitcharac},
the GNS triple of $\omega_{\Phi}$
is of the form 
$(\caK_L\otimes\caK_R,\rho_L\otimes\rho_R, \Omega)$
with irreducible representations
$(\caK_L,\rho_L)$, $(\caK_R,\rho_R)$ of $\caA_{L}$,
$\caA_{R}$
and a unit vector $\Omega\in\caK_L\otimes\caK_R$.
In this representation, we have
\begin{align*}
\lmk \lmk \rho_L\otimes\rho_R\rmk\lmk\caA_R\rmk\rmk ''
=\bbC\unit_{\caK_L}\otimes \caB(\caK_R),\quad
\lmk \lmk \rho_L\otimes\rho_R\rmk\lmk\caA_L\rmk\rmk''
=\caB(\caK_L)\otimes\bbC\unit_{\caK_R}.
\end{align*}
From now on, we use this structure of GNS triple repeatedly.
Furthermore, because $\omega_{\Phi}$ is a unique gapped ground state
of $\beta$-invariant interaction, it is $\beta$-invariant.
Therefore, by Theorem \ref{gnsthm}, there is a unitary representation $V$
of $G$
on $\caK_L\otimes\caK_R$
such that
\begin{align*}
\Ad\lmk V_g\rmk\circ\lmk \rho_L\otimes\rho_R\rmk
= \rho_L\beta_g^L\otimes\rho_R\beta_g^R,\quad
V_g\Omega=\Omega.
\end{align*}
If we restrict this to $\caA_R$,
then we get
\begin{align*}
\Ad\lmk V_g\rmk\lmk\unit_{\caK_L}\otimes \rho_R(A)\rmk
=\unit_{\caK_L}\otimes \rho_R\beta_g^R(A),\quad A\in\caA_R.
\end{align*}
From this and irreducibility of $\rho_R$, we obtain a $*$-automorphism
$\gamma_g$ of  $\caB(\caK_R)$ such that
\begin{align*}
&\Ad \lmk V_g\rmk\lmk \unit_{\caK_L}\otimes x\rmk
=\unit_{\caK_L}\otimes \gamma_g(x),\quad x\in\caB(\caK_R),\\
&\gamma_g\rho_R(A)= \rho_R\beta_g^R(A),\quad A\in\caA_R.
\end{align*}
Note that $\gamma : G\to \Aut\lmk \caB(\caK_R)\rmk$ is a group action.
By Theorem \ref{wigner}, there is a unitary $u_R(g)$ on $\caK_R$
such that
\begin{align*}
\Ad(u_R(g))(x)=\gamma_g(x),\quad x\in \caB(\caK_R).
\end{align*}
Because $\gamma$ is a group action $u_R : G\to \caU\lmk \caK_R\rmk$
is a projective representation.
Hence we got a projective representation
$u_R$  of $G$ on $\caK_R$
such that
\begin{align*}
\Ad(u_R(g))\rho_R(A)=\rho_R\beta_g^R(A),\quad A\in\caA_R.
\end{align*}

Let $h_R$ be the second group cohomology class associated to the projective
representation.
Using Theorem  \ref{wigner}, one can show that this
value $h_R$ is independent of the choice of
$\caK_R,\rho_R,\caK_L,\rho_L$ (see \cite{TRI}).
This value turned out to be the desired invariant.
\begin{thm}[\cite{TRI}]\label{h2sta}
This second cohomology class $h_{\Phi}:=h_R\in H^2(G,\Uo)$
is an invariant of $\sim_\beta$.
\end{thm}

Now, let us investigate what this $H^2(G,\Uo)$-valued index represents.
By the same argument as the derivation of $u_R$, there is 
a projective representation $u_L$ of $G$
 on $\caK_L$
such that
\begin{align*}
\Ad(u_L(g))\rho_L(A)=\rho_L\beta_g^L(A),\quad A\in\caA_L.
\end{align*}
Let $\sigma_L$, $\sigma_R$ be $2$-cocycles associated to
projective representations $u_L$, $u_R$.
Now both of $V_g$ and $u_L(g)\otimes u_R(g)$
implements $\beta_g$
\begin{align*}
\Ad\lmk V_g\rmk\circ\lmk \rho_L\otimes\rho_R\rmk
= \rho_L\beta_g^L\otimes\rho_R\beta_g^R
=\Ad(u_L(g)\otimes u_R(g)) \circ\lmk \rho_L\otimes\rho_R\rmk.
\end{align*}
Because $ \rho_L\otimes\rho_R$ is irreducible,
this means 
that 
\[
V_g=c_g(u_L(g)\otimes u_R(g)),\quad g\in G,
\] 
with some $\Uo$-phase $c_g$.
Absorbing this $c_g$ to $u_L(g)$, we may assume that $c_g=1$.
Note that it does not change the 2nd cohomology class associated to
$u_L$.
Hence we obtain
\begin{align}\label{Vuu}
V_g=u_L(g)\otimes u_R(g),\quad g\in G.
\end{align}
From the fact that $V$ is a genuine representation,
while $u_L$, $u_R$ are projective representations
with $2$-cocycles $\sigma_L$, $\sigma_R$, we get
\begin{align*}
&u_L(gh)\otimes u_R(gh)
=V_{gh}
=V_gV_h\\
&=u_L(g) u_L(h)\otimes u_R(g)u_R(h)
=\sigma_L(g, h)\sigma_R(g,h)
u_L(gh)\otimes u_{R}(gh).
\end{align*}
Hence we obtain the relation
\begin{align*}
\sigma_L(g,h)\sigma_R(g,h)=\unit.
\end{align*}
Let us think about what this means.
First, $V$ is a {\it genuine} representation
corresponding to the $\beta$-action
of $G$ on the {\it whole chain}.
It splits into left and right, to projective representations
$u_L(g)$ and $u_R(g)$.
Let us consider the special case that  $\omega_{\Phi}$ is of product form with respect to left-right cut,
$\omega_{\Phi}=\omega_{\Phi,L}\otimes\omega_{\Phi,R}$.
Note that $\omega_{\Phi, L}$/ $\omega_{\Phi, R}$ is a pure $\beta^L$/$\beta^R$-invariant 
state on $\caA_L$/$\caA_R$.
Let $(\caH_{\omega_{\Phi,L}}, \pi_{\omega_{\Phi,L}},\Omega_{\omega_{\Phi,L}})$,
$(\caH_{\omega_{\Phi,R}}, \pi_{\omega_{\Phi,R}},\Omega_{\omega_{\Phi,R}})$
be  GNS representations of $\omega_{\Phi, L}$, $\omega_{\Phi, R}$.
Because $\omega_{\Phi, L}$, $\omega_{\Phi, R}$ are pure,
the representations $\pi_{\omega_{\Phi,L}}$, $\pi_{\omega_{\Phi,R}}$
are irreducible.
Because of the $\beta^L$/$\beta^R$-invariance,
there are genuine unitary representations $u_L$, $u_R$
of $G$ on $\caH_{\omega_{\Phi,L}}$, $\caH_{\omega_{\Phi,R}}$,
implementing $\beta^L$/$\beta^R$, and fixing $\Omega_{\omega_{\Phi,L}}$, $\Omega_{\omega_{\Phi,R}}$.
The triple 
\[
\lmk
\caH_{\omega_{\Phi,L}}\otimes \caH_{\omega_{\Phi,R}},
\pi_{\omega_{\Phi,L}}\otimes \pi_{\omega_{\Phi,R}},
\Omega_{\omega_{\Phi,L}}\otimes \Omega_{\omega_{\Phi,R}}
\rmk
\]
is a GNS triple of $\omega_{\Phi}=\omega_{\Phi,L}\otimes\omega_{\Phi,R}$,
and 
\[
V_g:=u_{L}(g)\otimes u_{R}(g)
\]
 is the genuine representation
implementing $\beta$ on $\caA_{\bbZ}$ and 
$V_g\lmk \Omega_{\omega_{\Phi,L}}\otimes \Omega_{\omega_{\Phi,R}}\rmk=\Omega_{\omega_{\Phi,L}}\otimes \Omega_{\omega_{\Phi,R}}$.
Hence in this product form case,
the genuine representation $V$ split into left and right genuine representations
$u_L,u_R$.
In terms of cohomology,
$1\in H^2(G,\Uo)$ split into $1\cdot 1$.
Namely, when there is no correlation at all between left and right half infinite chains,
the cohomology split is trivial.
From this, we can say that
$h_R$ indicates
how left-half and right-half are correlated,
in terms of the group action.

Lastly, we explain how to prove the stability of the index, Theorem \ref{h2sta}.
Namely, we would like to show that our index is invariant with respect to $\sim_\beta$.
Let $\Phi_1,\Phi_2\in\caP_{U.G.}$ be $\beta$-invariant
interactions such that $\Phi_1\sim_\beta\Phi_2$.
A key observation is the following.
\begin{prop}\cite{TRI}\label{facinv}
If $\Phi_1\sim_\beta \Phi_2$, then 
there are automorphisms $\alpha_L\in\Aut(\caA_L)$ and $\alpha_R\in\Aut(\caA_R)$ satisfying
\begin{description}
\item[(i)]
$\omega_{\Phi_2}$ and $\omega_{\Phi_1}\circ\lmk\alpha_{L}\otimes \alpha_{R}\rmk$
are quasi-equivalent,
\item[(ii)]
$\alpha_L\circ\beta_{g,L}=\beta_{g,L}\circ \alpha_L$, and
$\alpha_R\circ\beta_{g,R}=\beta_{g,R}\circ \alpha_R$
for all $g\in G$.
\end{description}
\end{prop}
{(i)}  says that we can move from $\omega_{\Phi_1}$ to $\omega_{\Phi_2}$ without
changing the correlation between left-right chains, by $\alpha_L\otimes\alpha_R$.
{(ii)} says that each $\alpha_L$ and $\alpha_R$ do not 
change the representations of each sides.
Hence
{ left-half and right-half correlation
in terms of the group action}, can not be changed by $\alpha_{L}\otimes \alpha_{R}$.
This proposition is nothing but the factorization property of the automorphic equivalence 
explained in subsection \ref{autosec}.

To proceed the proof, let
 $u_{R,1}$, $u_{R,2}$, $(\caK_{R,1},\rho_{R,1})$, $(\caK_{R,2},\rho_{R,2})$
  be the projective representations and irreducible representations given above,
corresponding to $\omega_{\Phi_1}$ and $\omega_{\Phi_2}$.
They satisfy
\begin{align*}
\Ad\lmk {u_{R,1}(g)}\rmk\circ \rho_{R,1}(A)=\rho_{R,1}\circ\beta_{g,R}(A),\quad
\Ad\lmk {u_{R,2}(g)}\rmk \circ \rho_{R,2}(A)=\rho_{R,2}\circ\beta_{g,R}(A),
\end{align*}
for $A\in\caA_R$.
Let $\pi_{1}$, $\pi_2$ be GNS representations of $\omega_{\Phi_1}$ and $\omega_{\Phi_2}$.
Then we have
\begin{align*}
\pi_2\sim_{\rm q.e.} \pi_1\circ\lmk\alpha_L\otimes \alpha_R
\rmk,
\end{align*}
by (i) of Proposition \ref{facinv}.
From this, we have
\begin{align*}
\rho_{R,2}\sim_{\rm q.e.} \left.\pi_2\right\vert_{\caA_R}\sim_{\rm q.e.}\left.\pi_1\right\vert_{\caA_R}\circ\alpha_R\sim_{\rm q.e.} \rho_{R,1}\circ\alpha_R.
\end{align*}
Hence there exists a $*$-isomorphism $\tau :\caB(\caK_1)\to\caB(\caK_2)$
such that
\begin{align*}
\rho_{R,2}=\tau\circ\rho_{R,1}\circ\alpha_R.
\end{align*}
By Theorem \ref{wigner} and the irreducibility of $\rho_{R,1}, \rho_{R,2}$,  there exists a
unitary $W : \caK_1\to\caK_2$ such that $\tau=\Ad(W)$.
Using this, we get
\begin{align*}
&\Ad(u_{R,2}(g))\circ\rho_{2,R}
=\rho_{2,R}\circ\beta_{g,R}
=
\tau\circ\rho_{R,1}\circ\alpha_R\circ\beta_{g,R}
\end{align*}
Now recall that $\alpha_R$ and $\beta_{g,R}
$ commutes.
Therefore, 
\begin{align*}
&\tau\circ\rho_{R,1}\circ\alpha_R\circ\beta_{g,R}
=
\tau\circ\rho_{R,1}\circ\beta_{g,R}\circ\alpha_R\\
&=
\tau\circ \Ad\lmk {u_{R,1}(g)}\rmk\circ\rho_{R,1}\circ\alpha_R
=
\tau\circ \Ad\lmk {u_{R,1}(g)}\rmk\circ\tau^{-1}\circ\rho_{R,2}\\
&=\Ad\lmk\Ad_W(u_{R,1}(g))\rmk\circ\rho_{R,2}
\end{align*}
\\
Hence there is $c_g\in U(1)$ such that
\begin{align*}
u_{R,2}(g)=c_g\Ad_W(u_{R,1}(g))
\quad g\in G.
\end{align*}
 From this, $u_{R,1}$ and $u_{R,2}$ have the {same second cohomology class}.
 The same proof goes through for anti-linear actions.

\subsection{Lieb-Schultz-Mattis (LSM) type theorems}
The cohomology valued index $h_{\Phi,\beta}$ is important as an invariant of our classification of SPT phases.
However, there is still another use of it, that is, Lieb-Schultz-Mattis (LSM) type theorem which we explain in this subsection.

It is a hard problem to decide if a given interaction has a unique gapped ground state or not.
For example, Haldane conjectured in 1983 that
Antiferromagnetic Heisenberg chain  
\begin{align*}
\Phi_{Heisenberg}(\{x,x+1\})=\sum_{j=1}^3\,S^{(x)}_jS^{(x+1)}_j,
\end{align*}
does not have a unique gapped ground state if $d$ is {even},
while it does if $d$ is {odd}.
The latter part of the conjecture is still an open problem.
Lieb-Schultz-Mattis \cite{LSM} and  Affleck-Lieb \cite{AL}
 proved affirmatively the first part of the conjecture.
They proved the following : if interaction has 
$\Uo$ on-site symmetry  and space translation symmetry 
and the spin-size $d$ is even, then it cannot have a unique gapped ground state.
They used the continuity of $\Uo$ in an ingenious way, to prove the existence of
low-energy excitation.
There has been a lot of progress after that \cite{HL}\cite{Os}\cite{NS}\cite{Matsui2}, but
traditionally, the {symmetry} has been some {continuous} ones.
Recently, it was claimed by physicists that Lieb-Schultz-Mattis type theorem holds for some 
{discrete} symmetries. It is proved for matrix product states parent Hamiltonians
\cite{ChenGuWEn2011},
\cite{WPVZ}
\cite{fuji}.

Using our $H^2(G,\Uo)$-valued index, we can prove their claim.
Let $G$ be a finite group and let $\beta_g:=\bigotimes_{\bbZ}\Ad(U_g)$
with $U$ a projective representation of $G$
on $\bbC^d$.
\begin{thm}[\cite{ot}, \cite{OTT}]
If the second cohomology class $h$ of $U$ is non-trivial, then there is no
translation invariant interaction in $\caP_{U.G.\beta}$.
\end{thm}
\begin{proof}
Suppose that there exists an interaction $\Phi\in \caP_{U.G.\beta}$ which is translation invariant.
Let 
$h_{\Phi,\beta}^x\in H^2(G, U(1))$ be an
 index associated to $\left.\omega_{\Phi}\right\vert_{\caA_{[x,\infty)\cap \bbZ}}$
 for each $x\in\bbZ$. (Note that we can cut the chain anywhere.)
Then we have
\[
h_{\Phi,\beta}^ {x}=h+h_{\Phi,\beta}^ {x+1},\quad\text{and}\quad
h_{\Phi,\beta}^ {x}=h_{\Phi,\beta}^{x+1}.
\]
Hence $h=0$.
\end{proof}
Recall that antiferromagnetic Heisenberg chain $\Phi_{Heisenberg}$
has the $\bbZ_2\times\bbZ_2$-symmetry.
There, cohomology class of the projective representation $U$ of $\bbZ_2\times\bbZ_2$ is trivial if $d$ is odd,
nontrivial if $d$ is even.
Hence we have $\Phi_{Heisenberg}\notin \caP_{U.G.\beta}$, if $d$ is even.
This gives an alternative proof for 
the first half of Haldane conjecture.
Analogously, we can prove the following, using the $H^2(G,\Uo)$-index.
\begin{thm}[\cite{OTT}]
If there is $\Phi\in \caP_{U.G.\beta}$, origin-reflection invariant, 
then $h=2 h_1$ with some $h_1\in H^2(G,\Uo)$.
\end{thm}

\subsection{A classification of pure states on quantum spin chains satisfying the 
split property with on-site finite group symmetries}

We obtained an invariant of the classification of SPT phases for
on-site finite group symmetry,
It is an open problem if they are {\it complete invariant}.

However, for on-site symmetry, we still can say that it is
a complete invariant of some (philosophically)
analogous classification.
That is, a classification of pure states on quantum spin chains satisfying the 
split property with on-site finite group symmetries,

The classification we consider here is the following.
\begin{defn}[\cite{GS}]
Let $\omega_{0}$, $\omega_{1}$ be pure $\beta$-invariant states on $\caA_{\bbZ}$
satisfying the split property with respect to $\caA_{L}$, $\caA_{R}$.\\
We write $\omega_{0}\sim_{{\rm split},\beta}\omega_{1}$
if there exist automorphisms $\alpha_{L}$ 
and $\alpha_{R}$ on $\caA_L$, $\caA_R$ such that
\begin{description}
\item[(i)]
$\omega_{1}$ and $\omega_{0}\circ\lmk\alpha_{L}\otimes \alpha_{R}\rmk$
are quasi-equivalent,
\item[(ii)]
$\alpha_L\circ\beta_{g,L}=\beta_{g,L}\circ \alpha_L$, and
$\alpha_R\circ\beta_{g,R}=\beta_{g,R}\circ \alpha_R$
for all $g\in G$.
\end{description}
\end{defn}
How can this be (philosophically)
analogous to our original classification?
The interpretation comes from rough identification
that 
\begin{center}
gapped \\
$\approx$  ground state is almost product with respect to
left-right cut\\
 $\approx$
 ground state has the split property.
\end{center}
These $\approx$ are not $=$.
However, philosophically, they share the concept.
From this analogy, {\it the existence of smooth path of 
interactions with unique gapped ground state}
is rephrased as {\it  the existence of
 a map which does not change the correlation
between left and right chains}.
In (i) of the classification, we are considering an automorphism
$\alpha_{L}\otimes \alpha_{R}$
of product form. This is a kind of map which does not change 
the correlation
between left and right chains at all.
Hence, we regard this (i) as an analog of $\sim$.
With this interpretation, adding the condition (ii)
corresponds to $\sim_\beta$, that is, the requirement
the symmetry should be preserved along the path.

Note that the $H^2(G,\Uo)$-valued index is actually defined
not only for $\omega_{\Phi}$ with $\beta$-invariant $\Phi\in\caP_{U.G.}$,
but it is also defined
for a $\beta$-invariant pure split state $\omega$.
Let us denote that $H^2(G,\Uo)$-valued index by $h(\omega)$,
for $\beta$-invariant pure split state $\omega$.

Then we obtain the following.
\begin{thm}\cite{GS}\label{geometry}
The $H^2(G,\Uo)$-valued index is a complete invariant of
$\sim_{{\rm split},\beta}$.
\end{thm}
Homogeneity of pure state space of separable simple $C^*$-algebras
has been studied extensively from old time in the operator algebra community
\cite{powers}, \cite{brat}, \cite{fkk}, \cite{kos}.
For the proof of Theorem \ref{geometry},
we used the machinery developed in these papers.
The point of our analysis, compared to theirs is how to encode the $H^2(G,\Uo)$
to the problem.
We encode
 the information of $[\sigma]_{H^2(G,U(1))}$ 
to the twisted dynamical system $\Sigma^{(\sigma)}=(G,\caA_R,\beta_{R},\sigma)$, and
consider the twisted crossed product 
$C^*(\Sigma^{(\sigma)})$ of $\Sigma^{(\sigma)}$.
Then
each $\omega$ with $h(\omega)=[\sigma]_{H^2(G, U(1))}$
gives a covariant representation of $\Sigma^{(\sigma)}$.
We then consider the homogeneity problem for $C^*\lmk \Sigma^{(\sigma)}\rmk$.

\subsection{An easy derivation of $H^{2}(G,\Uo)$}
In subsection \ref{onsiteone}, we derived
$H^{2}(G,\Uo)$-valued index for 
SPT phases with on-site symmetry, using the split propety,
{\it proven for all the unique gapped ground states}.
It relied on Matsui's Theorem, which is a corollary 
of Hasting's area law \cite{area}. Hasting's area law is quite a big theorem.
However, if you are interested in {\it only on SPT phases},
then there is a handy way to derive the $H^{2}(G,\Uo)$-valued index.
The point is, if $\Phi$ is in a SPT phase, by definition,
we have $\Phi\sim\Phi_0$, with the fixed reference on-site interaction $\Phi_0$.
Using the automorphic equivalence,
this means there is an automorphism $\alpha$ given by local interactions
such that $\omega_{\Phi}=\omega_{\Phi_0}\circ\alpha$.
Recall that $\omega_{\Phi_0}$ is a pure product state.
Therefore, there are pure states $\omega_L$, $\omega_R$ on $\caA_L$, $\caA_R$ such that
\[
\omega_{\Phi_0}=\omega_L\otimes\omega_R.
\]
On the other hand, as we saw in subsection \ref{autosec},
 $\alpha$ satisfies
a factorization property
\[
\alpha=\inn \circ\lmk
\alpha_L\otimes\alpha_R
\rmk.
\]
Now, because $\omega_{\Phi}$ is a unique ground state of $\beta$-invariant interaction
$\Phi$, it is $\beta$-invariant.
From this and above decompositions of $\omega_{\Phi_0}$, $\alpha$,
we get
\begin{align*}
&\lmk \omega_L\alpha_L\otimes\omega_R\alpha_R\rmk\circ\inn
=
\omega_\Phi=\omega_{\Phi}\circ\beta_g
=\lmk \omega_L\alpha_L\beta_g^{L}\otimes\omega_R\alpha_R\beta_g^R\rmk\circ\inn.
\end{align*}
Because $\omega_L$, $\omega_R$ are pure states, we may apply
Lemma \ref{elementlem}.
Let $(\caH_R,\pi_R,\Omega_R)$ be a GNS representation of $\omega_R$.
Because $\omega_R$ is pure, $\pi_R$ is irreducible.
 Applying this Lemma, we obtain
  a unitary $u_{g}$ on $\caH_R$ such that
\begin{align*}
\Ad(u_g)\circ\pi_{R}\circ\alpha_R(A)=\pi_{R}\circ\alpha_{R}\circ\beta_g^R(A),\quad A\in 
{\caA_R}.
\end{align*}
Note that 
\begin{align*}
&\Ad(u_gu_h)\circ\pi_{R}\circ\alpha_R(A)=\pi_{R}\circ\alpha_{R}\circ\beta_g^R\beta_h^R(A)\\
&=\pi_{R}\circ\alpha_{R}\circ\beta_{gh}^R(A)
=\Ad(u_{gh})\circ \pi_R\circ\alpha_R(A),
\end{align*}
for each $A\in {\caA_R}$ and $g,h\in G$.
Because $\pi_R$ is irreducible, this implies
\[
\Ad(u_gu_h)(x)=\Ad(u_{gh})(x)
\]
for all $x\in \caB(\caH)$.
Hence there is some $\sigma(g,h)\in \Uo$
such that 
\[
u_gu_h=\sigma(g,h)u_{gh},\quad g,h\in G.
\]
Namely, $u_g$ is a projective representation.
In fact this is the very projective representation $u_R(g)$ we found in
subsection \ref{onsiteone}.

\subsection{SPT phases with the reflection symmetry}
At the end of this section, we consider SPT-phases with the reflection symmetry.
Let us first explain the $\bbZ_2$-index introduced by Pollmann.et.al for the reflection symmetry \cite{po}\cite{po2}.
 Let $\omega_{\vv}$ be a $\Theta_R$-invariant MPS
given by a primitive $\vv\in\Mat_k^{\times d}$. 
Let $\rho_\vv$ be the density matrix corresponding to $\varphi_\vv$. (Recall Theorem \ref{thmprim}.)
Replacing $v_\mu$ with $e_{\vv}^{\frac 12}v_\mu e_{\vv}^{-\frac 12}$,
we may assume that $e_\vv=\unit$.
Because $\omega_{\vv}$ is $\Theta_R$-invariant,
we can check that
\[
\tilde v_{\mu}:=\rho_{\vv}^{-\frac 12} \bar{v_{\mu}}^* \rho_{\vv}^{\frac 12},\quad 
\mu=1,\ldots,d
\]
also give the same state.
Here, $\bar\cdot$ indicates the complex conjugation with respect to a CONS of $\Mat_k$
diagonalizing $\rho_{\vv}$.
Namely, we have
\[
\omega_{\tilde \vv} =\omega_{\vv}.
\]
Then, using the fundamental theorem of MPS (Theorem \ref{fnw}), one can show that
there is a $c\in\bbT$, a $k\times k$
anti-unitary matrix $\theta$, such that
\begin{align}\label{vfanda}
v_\mu=c\rho_{\vv}^{-\frac 12} \theta v_\mu^* \theta \rho_{\vv}^{\frac 12},\quad
\mu=1,\ldots,d.
\end{align}
(See \cite{RI} Lemma 5.1 for the detailed proof.)
As in section \ref{mpsonsite}, using the primitivity of $\vv$,
we can conclude that 
\[
\theta^2=\unit \quad\text{or}\quad \theta^2=-\unit.
\]
This signature is Pollmann et.al.'s $\bbZ_2$-index.

Now, as before, we would like to extend it to general $\omega_\Phi$, for
$\Theta_R$-invariant
$\Phi\in\caP_{U.G.}$.
Let $\Phi\in\caP_{U.G.}$ be a $\Theta_R$-invariant interaction.
By the {split property} and the {reflection symmetry} of $\omega_{\Phi}$,
the GNS of $\omega_\Phi$ is of the form
\begin{align*}
(\caH\otimes\caH, \hat\pi:=\lmk \pi\circ\Theta_R\vert_{\caA_L}\rmk \otimes \pi,\Omega).
\end{align*}
Let  $\Gamma$ be the unitary implementing
$\Theta_R$, i.e.,
\begin{align*}
\Gamma\hat\pi(A)\Omega=\hat\pi\circ\Theta_R(A)\Omega,\quad A\in\caA.
\end{align*}
Then
there is a $\sigma=\pm 1$ such that
\begin{align*}
\Gamma\lmk\xi\otimes \eta\rmk={\sigma}\lmk\eta\otimes \xi\rmk,\quad
\text{for all}\quad \xi,\eta\in\caH.
\end{align*}
This $\sigma$ is our $\bbZ_2$-valued index.
\begin{defn}[O '19]
 Define the index of $\Phi\in \caP_{U.G.,\Theta_R}$ by {$\sigma_{\Phi,\Theta_R}:=\sigma$}.
\end{defn}

From the left-right symmetric form, 
\begin{align*}
(\caH\otimes\caH, \hat\pi:=\lmk \pi\circ\Theta_R\vert_{\caA_L}\rmk \otimes \pi,\Omega),
\end{align*}
the index $\sigma$ allows a characterization with the {Tomita-Takesaki theory}.
We briefly explain it here.
For simplicity, we assume that the GNS-vector $\Omega$
is cyclic and separating for $\bbC\unit\otimes\caB(\caH)$ in $\caH\otimes\caH$.
(Otherwise, we restrict $\bbC\unit\otimes\caB(\caH)$ to the support of $\Omega$.)
By this simplification, there is an anti-liner operator
on $\caH\otimes\caH$
with domain $\lmk \bbC\unit\otimes\caB(\caH)\rmk\Omega$
such that
\begin{align}\label{takesdef}
S\lmk \unit\otimes x\rmk\Omega:=
\lmk \unit\otimes x^*\rmk\Omega,\quad x\in \caB(\caH).
\end{align}
This $S$ is closable.
Let
\[
\bar S=J\Delta^{\frac 12}
\]
be the polar decomposition of $\bar S$.
The anti-linear operator $J$ is called the modular conjugation
and the positive operator $\Delta$ is called the modular operator.
By a general theory of Tomita-Takesaki theory (see \cite{BR1}),
we have
\begin{align}\label{tomtake}
\Delta\Omega=J\Omega=\Omega.
\end{align}
Then the following holds.
\begin{prop}[\cite{RI}]\label{tanoshi}
There is an anti-unitary $\theta$ on $\caH$ such that
\begin{align*}
J\lmk\unit\otimes x\rmk J^*=\theta x\theta^*\otimes \unit,\quad
\theta^2=\sigma_{\Phi,\Theta_R}\unit.
\end{align*}
\end{prop}

This characterization can be used to show that our 
index is the generalization of Pollmann et.al.'s index on MPS.
Roughly speaking, it can be understood as follows.
Let us consider a $\Theta_R$-invariant MPS $\omega_\vv$ given by a primitive $\vv$.
For simplicity, we still pretend that the GNS-vector $\Omega$ of $\omega_{\vv}$
is cyclic and separating for $\bbC\unit\otimes\caB(\caH)$ in $\caH\otimes\caH$,
although it is not true for MPS.(See the original paper \cite{RI} for proper formulation.)
In this case, we have
\begin{align}\label{mod}
\Delta^{\frac 12}\lmk \unit \otimes x\rmk\Omega
=\lmk \unit \otimes \rho_\vv^\frac 12 x \rho_\vv^{-\frac 12}\rmk\Omega,
\quad
x\in B(\caH).
\end{align}

Recall that we may take $v_\mu$ as a restriction of $S_\mu$
in (\ref{stachi}) to some finite dimensional subspace.
For simplicity, let us omit the difference between $S_\mu$ and $v_\mu$.
The point is, in our setting, there are two ways to send operators
in $\bbC\unit\otimes B(\caH)$
to operators in $\caB(\caH)\otimes \bbC\unit$.
One is by modular theory,
that is,
\begin{align*}
\begin{split}
&\lmk
\unit\otimes \rho_\vv^{\frac 12} \theta S_\mu\theta^* \rho_\vv^{-\frac 12}
\rmk\Omega
=
J\Delta^{\frac 12}
\lmk
\unit\otimes \rho_\vv^{-\frac 12} \theta S_\mu^*\theta^* \rho_\vv^{\frac 12}
\rmk\Omega\\
&=J\lmk
\unit\otimes \theta S_\mu^*\theta^*
\rmk\Omega
=J\lmk
\unit\otimes \theta S_\mu^*\theta^*
\rmk J^*\Omega\\
&=\lmk S_\mu^*\otimes \unit \rmk\Omega.
\end{split}
\end{align*}
Here, we used (\ref{takesdef}),  (\ref{mod}), (\ref{tomtake}), and Proposition \ref{tanoshi}.
On the other hand, using the fact that $\omega_\vv$
is translationally invariant, we can show that there is some $a\in\bbT$
such that
\begin{align*}
a \lmk \unit \otimes S_\mu^*\rmk\Omega=
\lmk S_\mu^*\otimes \unit \rmk\Omega,\quad \mu=1,\ldots,d.
\end{align*}
To understand it heuristically, recall that the left side of $\caH\otimes \caH$ corresponds to
left infinite chain, while the right-hand side corresponds to
the right infinite chain.
Comparing these equation, and the assumption that $\Omega$ is separating
for $\bbC\unit\otimes\caB(\caH)$, we obtain
\begin{align*}
\rho_\vv^{\frac 12} \theta S_\mu\theta^* \rho_\vv^{-\frac 12}=a S_\mu^*,\quad \mu=1,\ldots,d.
\end{align*}
This corresponds to (\ref{vfanda}).
Therefore, the signature of $\theta^2$ is the Pollmann et.al.'s index.
However, from Proposition \ref{tanoshi},
we have $\theta^2=\sigma_{\Phi,\Theta_R}\unit$.
Hence our index coincides with Pollman's index in MPS.

\section{SPT-phases in $2$-dimensional systems}
In this section we explain the derivation of $H^3(G,\Uo)$-valued index
for $2$-dimensional SPT obtained in \cite{2dSPT}.
Before starting the analysis, we fix some notation.
We denote by $H_L$, $H_R$,
$H_U$, $H_D$, the
left, right, upper, lower half-planes i.e.,
\begin{align*}
H_L:=\left\{ (x,y)\in\bbZ^2\mid x\le -1\right\},\quad
H_R:=\left\{ (x,y)\in\bbZ^2\mid 0\le x\right\},\\
H_U:=\left\{ (x,y)\in\bbZ^2\mid 0\le y\right\},\quad
H_D:=\left\{ (x,y)\in\bbZ^2\mid y\le -1\right\}.
\end{align*}
For each $\Gamma\subset \bbZ^\nu$, let $\beta_g^\Gamma$
be the restriction of $\beta_g$ to $\caA_{\Gamma}$.
With a bit abuse of notation, we also use the same symbol $\beta_g^\Gamma$
to denote $\beta_g^\Gamma\otimes\id_{\caA_{\Gamma^c}}$.
We use notation $\beta_g:=\beta_g^{\bbZ^2}$, $\beta_g^U:=\beta_g^{H_U}$,
$\beta_g^{RU}:=\beta_g^{H_R\cap H_U}$, $\beta_g^{LU}:=\beta_g^{H_L\cap H_U}$.
Throughout this section, we fix a trivial $\Phi_0\in\caP_{U.G.}^0$.
Recall that its unique gapped ground state $\omega_{\Phi_0}$ is of product form
\begin{align*}
\omega_{\Phi_{0}}=\bigotimes_{x\in\bbZ^{\nu}}\rho_{\xi_{x}}.
\end{align*}
We denote by $\omega_R$, $\omega_L$, $\omega_U$, $\omega_D$,
the restriction of $\omega_{\Phi_{0}}$ to $\caA_{H_R}$, $\caA_{H_L}$, $\caA_{H_U}$, $\caA_{H_D}$
respectively.
Because $\omega_{\Phi_{0}}$ is a pure state of a product form,
both of $\omega_R$ and $\omega_L$ are pure.
Let  $(\caH_L,\pi_L,\Omega_L)$,
$(\caH_R,\pi_R,\Omega_R)$ be GNS-triples of $\omega_L$, $\omega_R$, respectively.
Then $(\caH_0,\pi_0,\Omega_0)$,
with $\caH_0=\caH_L\otimes\caH_R$,
$\pi_0:=\pi_L\otimes\pi_R$,
$\Omega_0:=\Omega_L\otimes\Omega_R$
is a GNS triple for $\omega_{\Phi_0}$.
Note that all of $\pi_L$, $\pi_R$, $\pi_0$
are irreducible because
$\omega_L$, $\omega_R$, $\omega_0$
are pure.

\subsection{Derivation of $H^3(G,\Uo)$ out of cocycle action}\label{gker}
In one dimensional system, we obtained a $H^2(G,\Uo)$-valued 
index out of projective representations.
According to physicists, we are supposed to find 
a  $H^3(G,\Uo)$-valued index.
How can we find such an index?

We can find a hint in literatures of operator algebra
on the classification of cocycle actions.
\cite{Connes}, \cite{jones}.

Let $\{\gamma_g\}_{g\in G}$ be automorphisms on  a factor $\caM$.
Assume that there are unitaries $\{ u(g,h)\}_{g,h\in G}\subset\caU(\caM)$ such that
\begin{align}\label{grel}
\gamma_g\gamma_h=
\Ad\lmk u(g,h)\rmk\gamma_{gh}.
\end{align}
This $\gamma$ is called a cocycle action.
Using this relation twice, we obtain 
\begin{align*}
\gamma_g\gamma_h\gamma_k
=\Ad\lmk u(g,h)\rmk\gamma_{gh}\gamma_k
=\Ad\lmk u(g,h) u(gh,k)\rmk\gamma_{ghk}.
\end{align*}
Here we used the relation (\ref{grel})
to $\gamma_g\gamma_h$ and then to $\gamma_{gh}\gamma_k$.
On the other hand, we have
\begin{align*}
&\gamma_g\gamma_h\gamma_k
=\gamma_g \Ad\lmk u(h,k)\rmk\gamma_{hk}
=\Ad\lmk
\gamma_g\lmk u(h,k)\rmk
\rmk\gamma_g \gamma_{hk}\\
&=\Ad\lmk
\gamma_g\lmk u(h,k)\rmk u (g,hk)
\rmk\gamma_{ghk}.
\end{align*}
Here we used the relation (\ref{grel})
to $\gamma_h\gamma_k$ and then $\gamma_{g}\gamma_{hk}$.
By the associativity, of course above two equations should give the same value.
Hence we get
\begin{align*}
\Ad\lmk u(g,h) u(gh,k)\rmk(a)
=\Ad\lmk
\gamma_g\lmk u(h,k)\rmk u (g,hk)
\rmk(a),\quad a\in\caM.
\end{align*}
This means
\begin{align*}
\lmk u(g,h) u(gh,k)\rmk^*
\lmk
\gamma_g\lmk u(h,k)\rmk u (g,hk)
\rmk
\in\caM\cap \caM'=\bbC\unit.
\end{align*}
Namely,
there exists some phase $c(g,h,k)$
such that 
\begin{align}\label{3coc}
u(g,h) u(gh,k)
=c(g,h,k)\gamma_g\lmk u(h,k)\rmk u (g,hk),\quad
g,h,k\in G.
\end{align}
Using this relation repeatedly, we can show that 
$c: G^{\times 3}\to \Uo$ satisfies the $3$-cocycle relation.
Indeed, for $g,h,k,f\in G$, by repeated use of (\ref{3coc}), we get
\begin{align}\label{temp}
& u(g,h) u(gh,k)u(ghk,f)
=\lcm  u(g,h) u(gh,k)\rcm u(ghk,f)\notag\\
& = c(g,h,k)
 \gamma_g\lmk u(h,k)\rmk 
u(g,hk) u(ghk,f)
\notag\\
& = c(g,h,k)
 \gamma_g\lmk u(h,k) \rmk
 \cdot
\lcm  u(g,hk)u(ghk,f)\rcm
\notag\\
& =c(g,h,k)
 \gamma_g\lmk u(h,k) 
 \rmk\cdot
c(g,hk,f) \lmk \gamma_g\lmk  u(hk,f)\rmk
 \lmk u(g,hkf)\rmk \rmk\notag\\
 &=
  c(g,h,k) c(g,hk,f)
 \lmk 
 \gamma_g\lcm u(h,k)  u(hk,f)\rcm \rmk\cdot
 \lmk u(g,hkf)\rmk \notag\\
 &=c(g,h,k)
 c(g,hk,f) c(h,k,f)
 \gamma_g\lmk
 \gamma_h \lmk u(k,f)\rmk \lmk u(h,kf)\rmk\rmk
 \cdot
 \lmk u(g,hkf)\rmk\notag\\
 &=c(g,h,k)
 c(g,hk,f) c(h,k,f)
 \cdot \gamma_g\gamma_h \lmk  u(k,f)\rmk \cdot 
 \lcm\lmk \gamma_g\lmk u(h,kf)\rmk \rmk
 u(g,hkf)\rcm\notag\\
& =c(g,h,k)
 c(g,hk,f) c(h,k,f) \overline{c(g,h,kf)}
 \cdot \left\{\gamma_g\gamma_h  \lmk u(k,f)\rmk \rmk\}\cdot 
\lmk
u(g,h)u(gh,kf)
\rmk.\notag\\
\end{align}
Here, (and below) we apply  (\ref{3coc}) for terms in $\lcm\cdot\rcm$ to get the succeeding equality.
Using (\ref{grel}) for $\gamma_g\gamma_h$
and again (\ref{3coc}), we obtain
\begin{align}
(\ref{temp})
&=c(g,h,k)
 c(g,hk,f) c(h,k,f)\overline{c(g,h,kf)}
 \lmk
\Ad\lmk   u(g,h) \rmk\gamma_{gh}
\rmk
\lmk u(k,f)\rmk 
\lmk
 u(g,h)u(gh,kf)
\rmk\notag\\
&=c(g,h,k)
 c(g,hk,f) c(h,k,f)\overline{c(g,h,kf)}
 \lmk  u(g,h)\rmk 
 \lcm \gamma_{gh}\lmk u(k,f)\rmk
 \lmk
 u(gh,kf)
\rmk\rcm\notag\\
&=c(g,h,k)
 c(g,hk,f) c(h,k,f)\overline{c(g,h,kf)}\overline{c(gh,k,f)}
  \lmk u(g,h)u(gh, k)u(ghk,f)\rmk.\notag
  \end{align}
Hence, we obtain
\begin{align}
c(g,h,k)
 c(g,hk,f) c(h,k,f)\overline{c(g,h,kf)}\overline{c(gh,k,f)}=1,\quad\text{for all}\quad
 g,h,k,f\in G.\notag
\end{align}
This means $c\in Z^{3}(G,\bbT)$.
Hence we obtain $[c]_{H^3(G,\Uo)}\in H^3(G,\Uo)$.
Note that resulting $[c]_{H^3(G,\Uo)}$ does not 
depend on the choice of $u$.
Indeed, if $\gamma_g\gamma_h=
\Ad\lmk \tilde u(g,h)\rmk\gamma_{gh}$
holds for another $\tilde u(g,h)\in \caM$,
the difference 
$u(g,h)^*\tilde u(g,h)$ belongs to the center
$\caM\cap\caM'=\bbC\unit$.
Hence there is some $a(g,h)\in \Uo$
such that 
$u(g,h)^*\tilde u(g,h)=a(g,h)$.
Let $\tilde c(g,h,k)\in\Uo$ be
the phase satisfying $\tilde u(g,h) \tilde u(gh,k)
=\tilde c(g,h,k)\gamma_g\lmk \tilde u(h,k)\rmk \tilde u (g,hk)$.
Then we have
\begin{align}
&a(g,h)a(gh,k)\lmk u(g,h) u(gh,k)\rmk\\
&=
\tilde u(g,h) \tilde u(gh,k)\notag\\
&=\tilde c(g,h,k)
 \gamma_g\lmk \tilde u(h,k)\rmk 
\tilde u(g,hk)\notag\\
&=\tilde c(g,h,k)a(h,k)a(g,hk)
\gamma_g\lmk u(h,k)\rmk 
 u(g,hk)\notag\\
&=\tilde c(g,h,k)a(h,k)a(g,hk)\overline{c(g,h,k)}
u(g,h) u(gh,k).\notag
\end{align}
Hence we have $\tilde c(g,h,k)=c(g,h,k)\overline{a(h,k)a(g,hk)}a(g,h)a(gh,k)$.
This means the difference between $\tilde c$ and $c$ is just a coboundary, and
we get $[c]_{H^{3}(G,\bbT)}=[\tilde c]_{H^{3}(G,\bbT)}$.


The above argument tells us,
if we have automorphisms $\{\gamma_g\}_{g\in G}$ on $\caA$
such that $\gamma_g\gamma_h\gamma_{gh}^{-1}$ is inner,
then we get some $[c]\in H^3(G,\Uo)$.
This can be a hint to find a $H^3(G,\Uo)$-valued index.

Before closing the subsection, we would like to emphasize the importance of
considering infinite system, not finite quantum spin systems.
On a matrix algebra, any automorphism is inner.
Therefore, if we collect some $\{\gamma_g\}_{g\in G}$,
they are represented as 
$\gamma_g=\Ad(u_g)$, by some unitary matrix $u_g$.
Then we may take 
$u(g,h)$ in
\begin{align}\label{gkereq}
\gamma_g\gamma_h=
\Ad\lmk u(g,h)\rmk\gamma_{gh}.
\end{align}
as $u_gu_hu_{gh}^*$. 
Substituting this to
\begin{align*}
u(g,h) u(gh,k)
=c(g,h,k)\gamma_g\lmk u(h,k)\rmk u (g,hk),
\end{align*}
we obtain $c(g,h,k)=1$.
Hence we get only a trivial element of $H^3(G,\Uo)$.
This can be seen as a reason that we consider {\it infinite }
quantum spin systems.

\subsection{Action of $\beta_g^U$ on automorphisms}\label{goodautosec}
From what we saw from the subsection \ref{gker}, 
our task can be finding automorphisms $\{\gamma_g\}_{g\in G}$ on $\caA$
such that $\gamma_g\gamma_h\gamma_{gh}^{-1}$ is inner.
But how can we find such a thing?
We have a genuine group action $\beta_g$ but a priori, there is no
set of automorphisms which are ``nicely obstacled''
like the ones in cocycle action.

A hint comes from former works on gapped phases.
The wisdom we can see from there is
{\it if you cut the system, you may obtain some non-trivial index}.
In fact, that is the heart of the phenomena
called {\it bulk-edge correspondence}, in free Fermions.
This correspondence says that if you would like to know about the bulk(=entire system), cut the system into two
and look at the edge.
We did the same thing in section \ref{1dimsec} for quantum spin chains as well.
We restricted our attention to the right half of the chain.
Also in \cite{bdf}, an integer-valued index for systems with $U(1)$-symmetry
 was derived by considering $U(1)$-action on half of the system.
From these works, it looks to be a good idea to consider the action of
the $G$-symmetry on the half-plane.
Recall that from the automorphic equivalence,
the ground state $\omega_{\Phi}$ of $\Phi\in\caP_{U.G.\beta}^0$
is given from a product state $\omega_{\Phi_0}$
via some automorphism $\alpha$, as $\omega_{\Phi}=\omega_{\Phi_0}\circ\alpha$.
This motivates us to think of an action (on the half infinite plane) of $G$ on automorphisms.
More precisely we consider the following action.

As we stated above, we consider an action of
the $G$-symmetry corresponding to the upper half-plane.
\begin{align*}
\beta_g^U:=\id_{\caA_{H_D}}
\otimes \bigotimes_{x\in H_U} 
\Ad\lmk U(g)\rmk,\quad g\in G.
\end{align*}
on the automorphism group $\Aut(\caA_{\bbZ^2})$.
The action is
given by
\begin{align*}
G\times \Aut(\caA_{\bbZ^2})\ni (g,\alpha) \to \beta_g^U\alpha \lmk \beta_g^U\rmk^{-1}
\in \Aut(\caA_{\bbZ^2}).
\end{align*}
An idea to allow us to derive $H^3(G,\Uo)$ is 
{\it cut further}.
Namely, we consider the following situation.
\begin{lem}\label{cco}
Suppose that
\begin{align}\label{autofac}
\beta_g^U\alpha\lmk \beta_g^U\rmk^{-1}\alpha^{-1}
=\inn\lmk \eta_{g,L}\otimes\eta_{g,R}\rmk.
\end{align}
Set $\gamma_{g}:=\eta_{g,R}\beta_g^{U,R}$. 
Then we have $\gamma_g\gamma_h\gamma_{gh}^{-1}=\Ad\lmk u(g,h)\rmk$,
with some unitary $u(g,h)\in \caA_{\bbZ^2}$.
Hence we obtain some
$[c]\in H^3(G,\Uo)$.
\end{lem}
\begin{proof}
Set $\gamma_{g, R}:=\eta_{g,R}\beta_g^{U,R}$
and $\gamma_{g, L}:=\eta_{g,L}\beta_g^{U,L}$.
Taking inverse of (\ref{autofac}), we get
\[
\alpha \beta_g^U\alpha^{-1}
=\inn\lmk \eta_{g,L}^{-1}\beta_g^{LU}\otimes\eta_{g,R}^{-1}\beta_g^{RU}\rmk
=\inn \lmk \gamma_{g,L}\otimes\gamma_{g,R}\rmk.
\]
Because $\beta_g^U$ is a genuine action of $G$, we get
\begin{align*}
&\id=
\alpha\beta_g^U\alpha^{-1}\circ
\alpha\beta_h^U\alpha^{-1}\circ \lmk \alpha\beta_{gh}^U\alpha^{-1}\rmk^{-1}\\
&=\inn \circ \lmk \gamma_{g,L}\gamma_{h,L} \gamma_{gh,L}^{-1}\rmk
\otimes\lmk \gamma_{g,R}\gamma_{h,R}\gamma_{gh,R}^{-1}\rmk.
\end{align*}
This means $\lmk \gamma_{g,L}\gamma_{h,L} \gamma_{gh,L}^{-1}\rmk
\otimes\lmk \gamma_{g,R}\gamma_{h,R}\gamma_{gh,R}^{-1}\rmk$
is inner. But as it is of the tensor product form
it means the automorphism $\gamma_{g,R}\gamma_{h,R}\gamma_{gh,R}^{-1}$
on $\caA_{H_R}$ is inner.
\end{proof}
In fact, in a model derived out of Dijkgraaf-Witten model in TQFT,
the ground state $\omega_{\Phi}$ is given by
some automorphism $\alpha$ and a product state $\omega_0$ 
as $\omega_\Phi=\omega_0\circ\alpha$.
This $\alpha$ satisfies the property (\ref{autofac}),
and above procedure allows us to recover the ''built in'' third cohomology class.
(See subsection \ref{dwsec}.)

We note from the proof that what is happening is heuristically the following:
we are considering a genuine action $\alpha\beta_g^U\alpha^{-1}$
on the whole space $\caA_{\bbZ^2}$.
This automorphism $\alpha\beta_g^U\alpha^{-1}$ can be 
split into left and right as in  (\ref{autofac}), modulo inner automorphisms.
The fact that the original $\alpha\beta_g^U\alpha^{-1}$
is a genuine action result in the 
fact that the right half fraction $\gamma_g$ of $\alpha\beta_g^U\alpha^{-1}$
is a cocycle action, i.e.,
$\gamma_g$ is a genuine action {\it modulo inner automorphisms}.

If the automorphism $\alpha$ in the automorphic equivalence satisfies
such a property, this would have allowed us to define a
$H^3(G,\Uo)$-valued index.
Unfortunately, it seems a it too much to ask.

\subsection{Action of $\beta^{U}$ on $\omega_{\Phi}$}
In the previous subsection, we saw that if our automorphism
$\alpha$ satisfy some nice factorization property (\ref{autofac}) with respect to left-right cut,
then we can obtain some $H^3(G,\Uo)$-valued index.
It does not look plausible (or at least we could not prove) that every $\alpha$ 
in automorphic equivalence
satisfies such a nice property.

Fortunately, a weak version of (\ref{autofac}) has turned out to be true.
This is the following Proposition. Recall the definition of $C_\theta$ (\ref{ctdef}).
\begin{prop}\label{betauprop}{\cite{2dSPT}}
Let $\Phi\in \caP_{U.G.\beta}^{0}$.
For any $0<\theta<\frac\pi 2$,
there are automorphisms 
$\eta_{g,\sigma}$, $\sigma=L,R$ localized in
${C_{\theta},\sigma}$, satisfying
\begin{align}\label{statefac}
\omega_{\Phi}=\omega_{\Phi}\circ
\lmk \eta_{g,L}\beta_{g}^{LU}\otimes\eta_{g,R}\beta_{g}^{RU}\rmk\circ\inn.
\end{align}
\end{prop}
If $\alpha$ in
 the automorphic equivalence
$\omega_\Phi=\omega_{\Phi_0}\circ\alpha$
satisfies the property (\ref{autofac})
and if $\omega_{\Phi_0}$ is invariant under
$\beta_g^U$, then $\omega_{\Phi}$
satisfies 
the property (\ref{statefac}) 
for some $\tilde \eta_{g,\sigma}$, $\sigma=L,R$ on $\Aut\lmk\caA_{H_\sigma}\rmk$.
Indeed, in this case, 
\begin{align*}
\omega_{\Phi}
&=\omega_{\Phi_0}\circ\alpha=\omega_{\Phi_0}\circ\lmk \beta_g^U\rmk^{-1}
\alpha\beta_g^U 
\alpha^{-1}\circ\alpha\lmk \beta_g^U\rmk^{-1}\\
&=\omega_{\Phi_0}\circ \inn\lmk \eta_{g^{-1},L}\otimes\eta_{g^{-1},R}\rmk
\circ\alpha\lmk \beta_g^U\rmk^{-1}\\
&=\omega_{\Phi}\circ \alpha^{-1} \lmk
\eta_{g^{-1},L}\otimes\eta_{g^{-1},R}\rmk
\circ\alpha\lmk \beta_g^U\rmk^{-1}\circ\inn.
\end{align*}
Now, by the quasi-locality of $\alpha$, there are automorphisms
$\tilde \eta_{g,\sigma}$, $\sigma=L,R$ on $\Aut\lmk\caA_{H_\sigma}\rmk$
such that 
\begin{align*}
\alpha^{-1} \lmk
\eta_{g^{-1},L}\otimes\eta_{g^{-1},R}\rmk
\circ\alpha
=\lmk \tilde \eta_{g^{-1},L}\otimes\tilde\eta_{g^{-1},R}\rmk\circ\inn.
\end{align*} 
(Lemma \ref{alex}.)
Substituting this to above, we obtain (\ref{statefac}) 
with $\eta_{g,\sigma}$ replaced by $\tilde \eta_{g,\sigma}$.

Intuitively, Proposition \ref{betauprop} can be understood as follows.
Let us rewrite (\ref{statefac}) as
\begin{align*}
\omega_{\Phi}\circ\beta_g^U=\omega_{\Phi}\circ
\lmk \eta_{g^{-1},L}\otimes\eta_{g^{-1},R}\rmk\circ\inn.
\end{align*}
On the left-hand side, we are considering the action of $\beta_g^U$
on $\omega_{\Phi}$.
Recall that $\omega_{\Phi}$ is $\beta_g$-invariant, 
\begin{align*}
\omega_{\Phi}\circ\beta_g=\omega_{\Phi}, 
\end{align*}
because 
it is a unique ground state of a $\beta$-invariant interaction $\Phi\in\caP_{U.G.\beta}^0$.
Clearly, on the upper half-plane, the action of $\beta_g^U$
and $\beta_g$ are the same, while
on the lower half-plane, the action of $\beta_g^U$
and $\id$ are the same, i.e.,
\begin{align*}
\left. \beta_g^U\right\vert_{\caA_{H_U}}
=\left. \beta_g\right\vert_{\caA_{H_U}},\quad
\left. \beta_g^U\right\vert_{\caA_{H_D}}
=\id_{\caA_{H_D}}.
\end{align*}
Now recall that as a 
unique gapped ground state, $\omega_{\Phi}$
 shows exponential decay of correlation functions,
 namely, it is ``almost'' a product state. 
 (Theorem \ref{expdec}.)
 From this ``low-correlation" and the fact that $\omega_\Phi$ is $\beta_g$-invariant,
 we intuitively expect
 that  the effect of $\beta_g^U$ 
 on $\omega_{\Phi}$ should be localized around the 
 $x$-axis.
 The Proposition \ref{betauprop}
 guarantees that {\it it is the case}.

Let us now prove Proposition \ref{betauprop}.
\begin{proofof}[Proposition \ref{betauprop}]
Let $\Phi\in \caP_{U.G.\beta}^{0}$.
Then by the automorphic equivalence Theorem \ref{autoinf},
there is an automorphism $\alpha$ such that $\omega_{\Phi}=\omega_{\Phi_{0}}\circ\alpha$.
Using the factorization property of $\alpha$, for any $0<\theta'<\frac \pi 2$, we obtain the following
decompositions
\begin{align}\label{aahok}
\begin{split}
&\alpha\circ
\beta_g^U\circ \alpha^{-1}
=\inn\circ \lmk \id_{\caA_{H_D}}\otimes 
\xi_U
\rmk\circ \lmk \Xi_{H_U,L}\otimes \Xi_{H_U,R}\rmk,\\
&
\alpha\circ
\beta_g\circ \alpha^{-1}
=\inn\circ \lmk  \xi_D\otimes \xi_U\rmk
\circ  \lmk \Xi_{\bbZ^2, L}\otimes  
\Xi_{\bbZ^2,R}\rmk.
\end{split}
\end{align}
Here, $\xi_{U}$ (resp. $\xi_{D}$)
 is an automorphism localized on the upper (resp. lower) half-plane $H_{U}$
 (resp. $H_{D}$),
while   $\Xi_{H_U,L}$ and $\Xi_{\bbZ^2,L}$
(resp. $\Xi_{H_U,R}$ and $\Xi_{\bbZ^2,R}$)
 are automorphisms localized on $C_{\theta', L}$
 (resp. $C_{\theta', R}$).
 
 Because $\omega_{\Phi}$ is a unique ground state of a $\beta$-invariant interaction,
 it is $\beta$-invariant.
 Therefore, we have
 \begin{align}\label{invt}
\omega_{\Phi_0}\circ \alpha\circ
\beta_g\circ \alpha^{-1}
=\omega_{\Phi_0}.
\end{align}
From the decomposition (\ref{aahok}),
we then get
\begin{align}\label{dulr}
\omega_{D}\circ\xi_{D}\otimes \omega_{U}\circ\xi_{U}
=
\omega_{\Phi_0}\circ\lmk  \xi_D\otimes \xi_U\rmk
=\omega_{\Phi_0}  \circ\lmk \Xi_{\bbZ^2, L}^{-1}\otimes  
\Xi_{\bbZ^2,R}^{-1}\rmk\circ\inn
\end{align}
 Note that the left-hand side is of the tensor product for cut
 $H_{D}-H_{U}$, while the right-hand side is
 almost tensor product form for cut $\lmk H_{D}\cup C_{\theta'}\rmk-\lmk H_{D}\cup C_{\theta'}\rmk^{c}$.
From this, the pure state $ \omega_{U}\circ\xi_{U}$
on $\caA_{H_{U}}$
satisfies the split property with respect to
 $\caA_{\lmk C_{\theta',U}\rmk}$ and $\caA_{\lmk C_{\theta'}\rmk^{c}\cap H_{U}}$.
 Furthermore, the restriction of $\omega_{U}\circ\xi_{U}$ onto $\caA_{\lmk  C_{\theta'}\rmk^{c}\cap H_{U}}$ is quasi-equivalent to the restriction of the pure product state $\omega_{\Phi_0}$
 onto the same algebra.
 As a result (Lemma \ref{splitlem5}), there is an automorphism $S$ on $\caA_{\lmk C_{\theta,U}\rmk}$
 such that 
 \begin{align}\label{ouxu}
 \begin{split}
  &\omega_{U}\circ\xi_{U}=\omega_{\Phi_0}\circ \lmk S\otimes \id_{\caA_{\lmk \lmk C_{\theta'}\rmk^{c}\cap H_{U}\rmk}}\rmk\circ\inn\\
& =
\lmk
 \lmk \omega_{\Phi_0}\vert_{\caA_{\lmk C_{\theta',U}\rmk}}\circ  S\rmk\otimes 
 \lmk \omega_{\Phi_0}\vert_{\caA_{\lmk \lmk C_{\theta'}\rmk^{c}\cap H_{U}\rmk}}\rmk\rmk
\circ\inn.
 \end{split}
 \end{align}

  On the other hand, from (\ref{dulr}), $\omega_{U}\circ\xi_{U}$ 
 also satisfies the split property with respect to
 ${H_{U}\cap H_{L}}$-${H_{U}\cap H_{R}}$ cut.
 As a result (Lemma \ref{splitlem5}), there are automorphisms  $\tilde \eta_{\sigma,g}\in\Aut\lmk
\caA_{\lmk C_{\theta'}\rmk_\sigma}\rmk$, $g\in G$, $\sigma=L,R$
such that 
\begin{align}
\omega_{\Phi_0}\vert_{\caA_{\lmk C_{\theta',U}\rmk}}\circ  S
 =\lmk
 \lmk \omega_{\Phi_0}\vert_{\caA_{\lmk C_{\theta',U}\cap H_L\rmk}}\circ\tilde\eta_{L,g}\rmk
 \otimes 
 \lmk \omega_{\Phi_0}\vert_{\caA_{\lmk C_{\theta',U}\cap H_R\rmk}}\circ\tilde\eta_{R,g}\rmk
 \rmk\circ\inn.
\end{align} 
Substituting this to (\ref{ouxu}),
we obtain
\begin{align*}
 \omega_{U}\circ\xi_{U}=
 \lmk \lmk
 \omega\vert_{\caA_{H_{U}\cap H_{L}}}\circ \tilde\eta_{L,g}\rmk
 \otimes \lmk \omega\vert_{\caA_{H_{U}\cap H_{R}}}\circ\tilde\eta_{R,g}\rmk\rmk
 \circ \inn
\end{align*}
Now, from this and (\ref{aahok}), we obtain a decomposition of
$\omega_{\Phi}\circ \beta_g^U$:
\begin{align}
\begin{split}
&\omega_{\Phi}\circ \beta_g^U=
\omega_{\Phi_0}\circ\alpha\circ
\beta_g^U\circ \alpha^{-1}\circ\alpha
=
\lmk \omega_{D}\otimes 
\omega_{U}\xi_U
\rmk\circ \lmk \Xi_{H_U,L}\otimes \Xi_{H_U,R}\rmk\circ\inn\circ\alpha\\
&=\omega_{\Phi_0}\circ
\lmk \hat \eta_{g,L}\otimes\hat\eta_{g,R}\rmk\circ\alpha\circ\inn,
\end{split}
\end{align}
where $\hat \eta_{g,\sigma}$ is an automorphism localized in
${C_{\theta'},\sigma}$.
Because of the factorization property of $\alpha$ for any
$\theta$ with $\theta'<\theta<\frac\pi 2$,
we may find (Lemma \ref{alex}) automorphisms $\eta_{g^{-1},\sigma}$, $\sigma=L,R$ localized in
${C_{\theta},\sigma}$, satisfying
\begin{align}
\begin{split}
&\omega_{\Phi}\circ \beta_g^U=
\omega_{\Phi_0}\circ
\lmk \hat \eta_{g,L}\otimes\hat\eta_{g,R}\rmk\circ\alpha\circ\inn
=\omega_{\Phi}\circ\alpha^{-1}\circ \lmk \hat \eta_{g,L}\otimes\hat\eta_{g,R}\rmk\circ\alpha\circ\inn\\
&=\omega_{\Phi}\circ
\lmk \eta_{g^{-1},L}\otimes\eta_{g^{-1},R}\rmk\circ\inn.
\end{split}
\end{align}
Hence we have obtained the  Proposition \ref{betauprop}.
\end{proofof}
\subsection{Definition of $H^{3}(G,\Uo)$-valued index}\label{defh3}
Recall when an automorphism satisfies a nice condition
(\ref{autofac}), we could obtain 
a $H^{3}(G,\Uo)$-valued index
by considering $u(g,h)\in\caA$ 
implementing  $\gamma_g\gamma_h\gamma_{gh}^{-1}$
, for $\gamma_{g}:=\eta_{g,R}\beta_g^{U,R}$.
From our weaker version Proposition \ref{betauprop}, we cannot gurantee
the existence of such $u(g,h)$
inside of $\caA$.
However, we can get some unitary $u(g,h)$
on the GNS-Hilbert space $\caH_R$ of $\omega_R$,
implementing 
\begin{align*}
\gamma_g\gamma_h\gamma_{gh}^{-1}=\eta_{g,R}\beta_g^{RU}
 \eta_{h,R} \lmk \beta_g^{RU}\rmk^{-1}
 \lmk \eta_{gh,R}\rmk^{-1}
\end{align*}
 with respect to
 a representation $\pi_R\alpha_R$ of $\caA_R$.
This is the following proposition.
\begin{prop}\label{defuprop}{\cite{2dSPT}}
For $\Phi\in \caP_{U.G.\beta}^{0}$, and  $0<\theta<\frac\pi 2$,
let $\eta_{g,\sigma}$, $\sigma=L,R$ be
automorphisms given by Proposition \ref{betauprop}.
Let $\alpha$ be an automorphism on $\caA_{\bbZ^2}$
satisfying the factorization property 
$\alpha=\lmk\alpha_L\otimes\alpha_R\rmk\circ\Theta\circ\inn$
for $\alpha_\sigma\in\Aut(\caA_{H_\sigma})$, $\sigma=L,R$
and $\Theta\in \Aut\lmk \caA_{C_\theta^c}\rmk$.
Suppose that $\omega_{\Phi}=\omega_{\Phi_0}\circ \alpha$.

Then there is a unitary $u(g,h)$ on $\caH_R$ such that
 \begin{align}\label{urur}
 \Ad\lmk u(g,h)\rmk\pi_R\circ\alpha_R(A)=
\pi_R\circ \alpha_R\eta_{g,R}\beta_g^{RU}
 \eta_{h,R} \lmk \beta_g^{RU}\rmk^{-1}
 \lmk \eta_{gh,R}\rmk^{-1}(A),\quad A\in \caA_R.
\end{align}
\end{prop}
\begin{rem}
Setting $\gamma_{g}:=\eta_{g,R}\beta_g^{U,R}$,
we can regard this as a weak version of Lemma \ref{cco}.
\end{rem}
\begin{proof}
With $\tilde \beta_g:=\alpha \lmk \eta_{g,L}\otimes\eta_{g,R}\rmk \beta_g^U\alpha^{-1}$,
we get $\omega_{\Phi_0}\circ \tilde\beta_g:=\omega_{\Phi_0}\circ \inn$.
Substituting the definition of $\tilde\beta_g$ and
the factorization property $\alpha=\lmk\alpha_L\otimes\alpha_R\rmk\circ\Theta\circ\inn$, we get
\begin{align*}
&\omega_{\Phi_0}\circ \inn=
\omega_{\Phi_0}\circ\tilde\beta_g\tilde\beta_h\tilde\beta_{gh}^{-1}\\
&=\omega_{\Phi_0}\circ
\alpha \lmk
\eta_{g,L}\beta_g^{LU}
 \eta_{h,L} \lmk \beta_g^{LU}\rmk^{-1}
 \lmk \eta_{gh,L}\rmk^{-1}
 \otimes
 \eta_{g,R}\beta_g^{RU}
 \eta_{h,R} \lmk \beta_g^{RU}\rmk^{-1}
 \lmk \eta_{gh,R}\rmk^{-1}
\rmk
\alpha^{-1}\\
&=
\omega_{\Phi_0}\circ
\lmk\alpha_L\otimes\alpha_R\rmk \Theta\notag\\
&\quad\lmk
\eta_{g,L}\beta_g^{LU}
 \eta_{h,L} \lmk \beta_g^{LU}\rmk^{-1}
 \lmk \eta_{gh,L}\rmk^{-1}
 \otimes
 \eta_{g,R}\beta_g^{RU}
 \eta_{h,R} \lmk \beta_g^{RU}\rmk^{-1}
 \lmk \eta_{gh,R}\rmk^{-1}
\rmk\notag\\
&\quad\Theta^{-1}
\lmk\alpha_L\otimes\alpha_R\rmk ^{-1}\circ\inn.
\end{align*}
However, $\Theta$ is localized at $\lmk C_\theta\rmk^c$,
while $\eta_{g,R}\beta_g^{RU}
 \eta_{h,R} \lmk \beta_g^{RU}\rmk^{-1}
 \lmk \eta_{gh,R}\rmk^{-1}$ is localized at $C_\theta$.
 Therefore, they commute, and we obtain
 \begin{align*}
&\lmk  \omega_L\otimes\omega_R\rmk\circ \inn
 =\omega_{\Phi_0}\circ \inn\\
& =\omega_{\Phi_0}\circ 
\lmk
\alpha_L\eta_{g,L}\beta_g^{LU}
 \eta_{h,L} \lmk \beta_g^{LU}\rmk^{-1}
 \lmk \eta_{gh,L}\rmk^{-1}\alpha_L^{-1}
 \otimes
 \alpha_R\eta_{g,R}\beta_g^{RU}
 \eta_{h,R} \lmk \beta_g^{RU}\rmk^{-1}
 \lmk \eta_{gh,R}\rmk^{-1}\alpha_R ^{-1}
\rmk\\
&=\omega_L\circ 
\lmk
\alpha_L\eta_{g,L}\beta_g^{LU}
 \eta_{h,L} \lmk \beta_g^{LU}\rmk^{-1}
 \lmk \eta_{gh,L}\rmk^{-1}\alpha_L^{-1}\rmk\\
&\quad \otimes
 \omega_R\circ 
\lmk
\alpha_R\eta_{g,R}\beta_g^{RU}
 \eta_{h,R} \lmk \beta_g^{RU}\rmk^{-1}
 \lmk \eta_{gh,R}\rmk^{-1}\alpha_R^{-1}\rmk.
 \end{align*}
Because $\omega_L$ and $\omega_R$ are pure, we may apply
Lemma \ref{elementlem}.
As a result, there is a unitary $u(g,h)$ on $\caH_R$ such that
\begin{align*}
\begin{split}
&\Ad\lmk u(g,h)\rmk\pi_R(A)
=\pi_R\circ \alpha_R\eta_{g,R}\beta_g^{RU}
 \eta_{h,R} \lmk \beta_g^{RU}\rmk^{-1}
 \lmk \eta_{gh,R}\rmk^{-1}\alpha_R^{-1}(A)
 \end{split}
\end{align*}
This proves the Proposition.
\end{proof}
Note that in general,
$\gamma_g=\eta_{g,R}\beta_g^{RU}$ 
does not need to be implemented 
by a unitary for the representation$(\caH_R,\pi_R\circ\alpha_R)$.
In general, representations 
$\pi_R\circ\alpha_R$ and $\pi_R\circ\alpha_R\circ\gamma_g$ are 
mutually singular. In fact, if 
they are quasi-equivalent 
the third cohomology class obtained below will be $0$.
Intuitively, $\pi_R\circ\alpha_R$ and $\pi_R\circ\alpha_R\circ\gamma_g$
are {\it macroscopically different}.
However, from the above Proposition \ref{defuprop},
if we consider the
combination
$\gamma_g\gamma_h\gamma_{gh}^{-1}$,
it is implementable in $(\caH_R,\pi_R\circ\alpha_R)$,
namely,
$\pi_R\circ\alpha_R$ and $\pi_R\circ\alpha_R\circ\gamma_g\gamma_h\gamma_{gh}^{-1}$
are quasi-equivalent.
Physically, they are macroscopically the same.

This situation is parallel to the situation in subsection \ref{goodautosec}.
Recall in subsection \ref{goodautosec},
the fact that the original $\alpha\beta_g^U\alpha^{-1}$
is a genuine action result in the 
fact that the right half fraction 
$\gamma_g$ is a genuine action  modulo inner automorphisms of $\caA$.
Here, in Proposition \ref{defuprop},
the fact that the original $\alpha\beta_g^U\alpha^{-1}$
is a genuine action result in
the fact that $\gamma_g$ is a genuine action {\it modulo
unitaries on $\caH_R$}.
\\

Having this $u(g,h)$,
we expect something like (\ref{3coc})
\begin{align*}
u(g,h) u(gh,k)
=c(g,h,k)\gamma_g\lmk u(h,k)\rmk u (g,hk),\quad
g,h,k\in G.
\end{align*}
should holds.
One question occurs : our  $\gamma_{g}:=\eta_{g,R}\beta_g^{U,R}$
is an automorphism on $\caA_R$, and it
does not act on a bounded opearators
on $\caH_R$, in general.
In particular, we do not know an action of 
our  $\gamma_{g}:=\eta_{g,R}\beta_g^{U,R}$
on $u(g,h)$.
Fortunately, there is an alternative.
\begin{prop}\label{wprop}{\cite{2dSPT}}
Let $\Phi\in \caP_{U.G.\beta}^{0}$, and  $0<\theta<\frac\pi 2$,
and $\eta_{g,\sigma}$, $\sigma=L,R$ be
automorphisms given by Proposition \ref{betauprop}.
Let $\alpha$ be an automorphism on $\caA_{\bbZ^2}$
satisfying the factorization property 
$\alpha=\lmk\alpha_L\otimes\alpha_R\rmk\circ\Theta\circ\inn$
with $\alpha_\sigma\in\Aut(\caA_{H_\sigma})$, $\sigma=L,R$
and $\Theta\in \Aut\lmk \caA_{C_\theta^c}\rmk$.
Suppose that $\omega_{\Phi}=\omega_{\Phi_0}\circ \alpha$.
Then there are unitaries $W_g$, $g\in G$ on $\caH_0$
such that
\begin{align}\label{wimp}
\Ad\lmk W_g\rmk\circ\pi_0
=\pi_0\circ\lmk \alpha_0\rmk\circ\Theta\circ\eta_g\beta_g^U\circ\Theta^{-1}\circ\alpha_0^{-1},\quad
g\in G,
\end{align}
where $\alpha_0:=\alpha_L\otimes\alpha_R$
and $\eta_g:=\eta_{g,L}\otimes\eta_{g, R}$.
Furthermore, there is a $c\in Z^3(G,\bbT)$
such that
\begin{align}\label{uwc}
\unit_{\caH_L}\otimes u(g,h) u(gh,k)
=c(g,h,k)
\lmk W_g\lmk \unit_{\caH_L}\otimes u(h,k)\rmk W_g^*\rmk
\lmk \unit_{\caH_L}\otimes u(g,hk)\rmk,
\end{align}
for all $g,h,k\in G$.
\end{prop}
\begin{proof}
As in Proposition \ref{defuprop},
an automorphism 
\[
\tilde \beta_g:=\alpha \lmk \eta_{g,L}\otimes\eta_{g,R}\rmk \beta_g^U\alpha^{-1}
=\alpha_0 \Theta\circ \lmk \eta_{g,L}\otimes\eta_{g,R}\rmk \beta_g^U
\Theta^{-1}\alpha_0^{-1}\circ\inn
\]
satisfies
$\omega_{\Phi_0}\circ \tilde\beta_g:=\omega_{\Phi_0}\circ \inn$.
It means $\omega_{\Phi_0}$ is invariant under
an automorphism 
\[
\alpha_0 \Theta\circ \lmk \eta_{g,L}\otimes\eta_{g,R}\rmk \beta_g^U
\Theta^{-1}\alpha_0^{-1}\circ\Ad(V_g)
\]
with some unitary $V_g\in\caA$.
Applying Theorem \ref{gnsthm},
we obtain $\tilde W_g$ satisfying
\begin{align*}
\Ad(\tilde W_g)\circ \pi_0
=\pi_0\circ
\alpha_0 \Theta\circ \lmk \eta_{g,L}\otimes\eta_{g,R}\rmk \beta_g^U
\Theta^{-1}\alpha_0^{-1}\circ\Ad(V_g).
\end{align*}
Absorbing $V_g$ part, we obtain $W_g$ satisfying (\ref{wimp}).

Next we prove that there is a number $c(g,h,k)\in\bbT$
satisfying (\ref{uwc}).
From Proposition \ref{defuprop}, we have
\begin{align}\label{stars}
\begin{split}
&\Ad\lmk \unit_{\caH_L}\otimes u(g,h) u(gh,k)
\rmk\pi_0\\
&=\pi_L\otimes \pi_R \circ \alpha_R\circ \lmk \eta_{g,R}\beta_g^{RU}\rmk
\lmk \eta_{h,R}\beta_h^{RU}\rmk
\lmk \eta_k^R\beta_k^{RU}\rmk
\lmk \eta_{ghk}^R\beta_{ghk}^{RU}\rmk^{-1}
\alpha_R^{-1}.
\end{split}
\end{align}
On the other hand,
from (\ref{wimp}) and (\ref{urur}), using the fact that the support of
$\Theta$ and $\eta$s are disjoint (as in the proof of Proposition \ref{betauprop}), we obtain
\begin{align*}
&\Ad\lmk W_g\lmk \unit_{\caH_L}\otimes u(h,k)\rmk W_g^*\rmk\circ\pi_0\\
&=\pi_L\otimes 
\pi_R\lmk
\alpha_R\eta_{g,R}\beta_g^{RU}
\lmk
\eta_{h,R}\beta_h^{R U}
\eta_k^R\lmk\beta_h^{R U}\rmk^{-1}\lmk \eta_{hk}^R\rmk^{-1}
\rmk
\lmk \eta_{g,R}\beta_g^{RU}\rmk^{-1}
\alpha_R^{-1}
\rmk
\end{align*}
 for all $g,h,k\in G$.
 Using this and
 Proposition  \ref{betauprop},
 we see that
\begin{align}
\Ad\lmk \lmk W_g\lmk \unit_{\caH_L}\otimes u(h,k)\rmk W_g^*\rmk
\lmk \unit_{\caH_L}\otimes u(g,hk)\rmk\rmk
\pi_0
\end{align}
is also equal to the right-hand side of (\ref{stars}).
Because
 $\pi_0$ is irreducible, this means that there is a number $c(g,h,k)\in\bbT$
satisfying (\ref{uwc}).

The proof that $c$ is a $3$-cocycle
is analogous to that in subsectoin \ref{gker}.
In order to carry the argument there out,
we just need the following:
\begin{align}\label{yon}
\Ad\lmk W_g W_h\rmk\lmk  \unit_{\caH_L}\otimes u(k,f)\rmk
=\lmk
\Ad\lmk  \lmk \unit_{\caH_L}\otimes u(g,h)\rmk W_{gh}\rmk
\rmk
\lmk \unit_{\caH_L}\otimes u(k,f)\rmk.
\end{align}
To see this, first we note from
(\ref{uwc}) that 
\begin{align}\label{hana}
\Ad\lmk   W_{gh}\rmk
\lmk \unit_{\caH_L}\otimes u(k,f)\rmk
\in \unit_{\caH_L}\otimes \caB(\caH_R),
\end{align}
for any $g,h,k,f\in G$.
We get
\begin{align*}
\begin{split}
&\Ad\lmk W_gW_hW_{gh}^*
\rmk\lmk
\unit_{\caH_L}\otimes \pi_R(A)
\rmk\\
&=\unit_{\caH_L}\otimes
\pi_R\circ\alpha_R
\circ\eta_{g,R}\beta_g^{RU}\eta_{h,R}\lmk \beta_g^{RU}\rmk^{-1}\eta_{gh,R}^{-1}\circ
\alpha_R^{-1}(A)\\
&=
\unit_{\caH_L}\otimes
\Ad u(g,h)\circ \pi_R(A),\end{split}
\end{align*}
for all $A\in\caA_R$,
again, using the commutativity of $\Theta$ and $\eta$s.
Because $\pi_R$ is irreducible,
this means
\begin{align*}
W_gW_hW_{gh}^*\lmk \unit_{\caH_L}\otimes u(g,h)\rmk^*
\in \caB(\caH_L)\otimes \unit_{\caH_R}. 
\end{align*}
Hence using (\ref{hana}),  we obtain
\begin{align*}
\begin{split}
&\Ad\lmk
W_gW_hW_{gh}^*\lmk \unit_{\caH_L}\otimes u(g,h)\rmk^*
\rmk
\lmk
\Ad\lmk  \lmk \unit_{\caH_L}\otimes u(g,h)\rmk W_{gh}\rmk
\rmk
\lmk \unit_{\caH_L}\otimes u(k,f)\rmk\\
&=\lmk
\Ad\lmk  \lmk \unit_{\caH_L}\otimes u(g,h)\rmk W_{gh}\rmk
\rmk
\lmk \unit_{\caH_L}\otimes u(k,f)\rmk.
\end{split}
\end{align*}
This implies the required property (\ref{yon}).
\end{proof}
Hence we have obtained 
some third group cohomology.
In order to define it, we made a lot of choices.
First of all, there can be many path to connect $\Phi$
to $\Phi_0$. According to that, 
there can be many choices for $\alpha$.
The factorization property was key to define the index,
but there is a freedom of the direction $\theta$
to cut the plane, as well as the choice of
$\alpha_R$, $\alpha_L$, $\Theta$.
Furthermore, the choice of $\eta_{g,L},\eta_{g,R}$
of Proposition \ref{betauprop} is neither unique.
Does resulting third cohomology class obtained in Proposition \ref{wprop}
depend on such choices?
Fortunately, the answer turns out to be no.
\begin{thm}\label{cochdef}{\cite{2dSPT}}
Let $\Phi\in \caP_{U.G.\beta}^{0}$.
The third cohomology class $[c]_{H^3(G,\Uo)}$ obtained from Proposition \ref{wprop} 
does not depend on the choice of
$\alpha$, 
$0<\theta<\frac\pi 2$,
 $\alpha_\sigma\in\Aut(\caA_{H_\sigma})$, $\sigma=L,R$,
$\Theta\in \Aut\lmk \caA_{C_\theta^c}\rmk$, 
$W_g$, $u(g,h)$, 
and $\eta_{g,\sigma}$, $\sigma=L,R$.
\end{thm}
\begin{defn}{\cite{2dSPT}}
We denote by $h(\Phi)$, the third cohomology class $[c]_{H^3(G,\Uo)}$ obtained
in Theorem \ref{cochdef}.
\end{defn}

\subsection{What $h(\Phi)$ represents}
Now let us investigate what $h(\Phi)$ represents.
In $1$-dimension, the $H^2(G,\Uo)$ represents how the left-half
and right-half of the chains are correlated, in terms of 
group action.
We can have an analogous observation in $2$-dimensional systems.
Let us denote $u(g,h)$ $c(g,h,k)$ given in the previous subsection by $u_R(g,h)$ and $c_R(g,h,k)$.
The same argument gives us $u_L(g,h)$ and $c_L(g,h,k)$, for left half-plane.
Hence we get
 \begin{align}\label{udefs}
 \Ad\lmk u_\sigma(g,h)\rmk\pi_\sigma \circ\alpha_\sigma (A)=
\pi_\sigma \circ \alpha_\sigma \eta_{g,\sigma }\beta_g^{\sigma U}
 \eta_{h,\sigma } \lmk \beta_g^{\sigma U}\rmk^{-1}
 \lmk \eta_{gh,\sigma }\rmk^{-1}(A),\quad A\in \caA_\sigma .
\end{align}
for $\sigma=L,R$.
By the definition of $W_g$, and the commutativity of $\Theta$ and $\eta$s, we have
\begin{align*}
&\Ad\lmk
{W_gW_hW_{gh}^*}
\rmk\circ \bigotimes_{\sigma=L,R}\pi_\sigma=
\bigotimes_{\sigma=L,R}
\pi_\sigma \circ \alpha_\sigma \eta_{g,\sigma }\beta_g^{\sigma U}
 \eta_{h,\sigma } \lmk \beta_g^{\sigma U}\rmk^{-1}
 \lmk \eta_{gh,\sigma }\rmk^{-1}\alpha_\sigma^{-1}.
\end{align*}
Comparing these equations, from the irreducibility of $\pi_0$, we have
\begin{align*}
\Ad\lmk
{W_gW_hW_{gh}^*}
\rmk(x)
=\Ad\lmk {u_L(g,h)}\otimes { u_R(g,h)}\rmk(x),\quad x\in \caB(\caH_0).
\end{align*}
This means there is some
$\sigma(g,h)\in\Uo$ such that
\begin{align*}
W_gW_hW_{gh}^*=\sigma(g,h)u_L(g,h)\otimes u_R(g,h).
\end{align*}
Absorbing the $\Uo$-phase to $u_L(g,h)$,
we may set $\sigma(g,h)=1$.
Hence we get
\begin{align*}
W_gW_hW_{gh}^*=u_L(g,h)\otimes u_R(g,h).
\end{align*}
By the same argument in the previous section for the left half-plane,
there are
$c_R, c_L\in Z^3(G,\Uo)$
such that
\begin{align}\label{ulur}
\begin{split}
&\unit_{\caH_L}\otimes u_R(g,h) u_R(gh,k)
={c_R(g,h,k)}
\lmk W_g\lmk \unit_{\caH_L}\otimes u_R(h,k)\rmk W_g^*\rmk
\lmk \unit_{\caH_L}\otimes u_R(g,hk)\rmk\\
& u_L(g,h) u_L(gh,k)\otimes\unit_{\caH_R}
={c_L(g,h,k)}
\lmk W_g\lmk u_L(h,k)\otimes \unit_{\caH_R} \rmk W_g^*\rmk
\lmk  u_L(g,hk)\otimes\unit_{\caH_L}\rmk.
\end{split}
\end{align}
Now we set
\[
V(g,h):=W_gW_hW_{gh}^*.
\]
Clearly, it satisfies the {$2$-cocycle relation}
\[
{V(g,h)}{V(gh,k)}=\Ad(W_g)\lmk {V(h,k)}\rmk\cdot 
{V(g,hk)},
\]
because
\[
\lmk{ W_gW_hW_{gh}^*}\rmk 
\lmk{ W_{gh}W_kW_{ghk}^*}\rmk 
=
W_g \lmk{ W_hW_kW_{hk}^*}\rmk W_g^*
\lmk{ W_gW_{hk}W_{ghk}^*}\rmk .
\]
On the other hand, $u_L, u_R$ satisfies the {obstructed versions},
(\ref{ulur}).
 Now, substituting $V(g,h)=W_gW_hW_{gh}^*=u_L(g,h)\otimes u_R(g,h)$,
we get 
\begin{align*}
[c_L]_{H^3(G,\Uo)}\cdot[c_R]_{H^3(G,\Uo)}=1.
\end{align*}
From this, we can interpret that 
what we see via the third cohomology valued index $h_\Phi$ is
how $V(g, h)$
satisfying a genuine $2$-cocycle relation
split into left-right $u^L$, $u^R$, satisfying ``obstacled  $2$-cocycle relations".

Suppose that $\omega_{\Phi}$ is of product form with respect to $H_D$-$H_U$ cut
\[
\omega_{\Phi}=\omega_{\Phi}\vert_{\caA_{H_D}}\otimes\omega_{\Phi}\vert_{\caA_{H_U}}.
\]
Because $\omega_{\Phi}$ is $\beta$-invariant,
we have
\[
\omega_{\Phi}\vert_{\caA_{H_D}}\circ\beta_g^{D}\otimes\omega_{\Phi}\vert_{\caA_{H_U}}\circ
\beta_g^{U}
=\omega_{\Phi}\circ\beta_g
=\omega_{\Phi}=\omega_{\Phi}\vert_{\caA_{H_D}}\otimes\omega_{\Phi}\vert_{\caA_{H_U}}.
\]
Hence we have
\begin{align*}
\omega_{\Phi}\vert_{\caA_{H_U}}\circ
\beta_g^{U}
=\omega_{\Phi}\vert_{\caA_{H_U}},
\end{align*}
and 
\begin{align*}
\omega_{\Phi}\circ\beta_g^U=\omega_{\Phi}.
\end{align*}
It means that in Proposition \ref{betauprop},
we may take $\eta_{g,R}$, $\eta_{g,L}$ as identities.
In this case, we may take $u_R$, $u_L$ as identities.
As a result, we get 
\begin{align*}
[c_L]_{H^3(G,\Uo)}=[c_R]_{H^3(G,\Uo)}=1.
\end{align*}
Hence $1$ is split into left and right in a trivial manner.

Suppose that $\omega_{\Phi}$ is of product form with respect to $H_L$-$H_R$ cut
\[
\omega_{\Phi}=\omega_{\Phi}\vert_{\caA_{H_L}}\otimes\omega_{\Phi}\vert_{\caA_{H_R}}.
\]
In this case, in order to define the $H^3(G,\Uo)$-valued index,
 we may take $\alpha$ of the form $\alpha=\alpha_L\otimes\alpha_R$,
 with $\alpha_L\in \Aut\lmk\caA_{H_L}\rmk$, and $\alpha_R\in \Aut\lmk\caA_{H_R}\rmk$.
(Actually, this is allowed because our $H^3(G,\Uo)$-valued
index can actually be defined on the wider class of systems. See the original paper \cite{2dSPT}.)
As a result, we may set 
\begin{align*}
&\omega_{\Phi}\vert_{\caA_{H_L}}=\omega_L\circ\alpha_L,\quad
\omega_{\Phi}\vert_{\caA_{H_R}}=\omega_R\circ\alpha_R,\\
&\Theta=\id.
\end{align*}

Now, from Proposition \ref{betauprop}, 
there are automorphisms 
$\eta_{g,\sigma}$, $\sigma=L,R$ localized in
$C_{\theta,\sigma}$, satisfying
\begin{align*}
\omega_{\Phi}=\omega_{\Phi}\circ
\lmk \eta_{g,L}\beta_{g}^{LU}\otimes\eta_{g,R}\beta_{g}^{RU}\rmk\circ\inn.
\end{align*}
From this, we see that
\begin{align*}
\begin{split}
&\omega_L\circ\alpha_L\circ\lmk \beta_g^{LU}\rmk^{-1}
\otimes \omega_R\circ\alpha_R\circ\lmk \beta_g^{RU}\rmk^{-1}\\
&=\omega_{\Phi}\circ \lmk \beta_g^U\rmk^{-1}\sim_{q.e.}\omega_{\Phi}\circ
\lmk \eta_{g,L}\otimes\eta_{g,R}\rmk\\
&=\omega_L\circ \alpha_L\circ \eta_{g,L}\otimes \omega_R\circ\alpha_R\circ\eta_{g,R}.
\end{split}
\end{align*}
This implies that
\begin{align*}
\omega_\sigma\circ\alpha_\sigma\circ\lmk \beta_g^{\sigma U}\rmk^{-1}\sim_{q.e.}
\omega_\sigma\circ \alpha_\sigma\circ \eta_{g,\sigma},
\end{align*}
for $\sigma=L,R$.
Because $\omega_\sigma$ is pure,
we see that there is a unitary $v_{g,\sigma}$ on $\caH_\sigma$
such that 
\begin{align*}
\Ad\lmk v_{g,\sigma}\rmk \circ\pi_\sigma (A)
=\pi_\sigma\circ\alpha_\sigma\circ   \eta_{g,\sigma}\beta_g^{\sigma U}\circ
\alpha_\sigma^{-1}(A),\quad
A\in\caA_{H_\sigma},
\end{align*}
for each $\sigma=L,R$.
Because $\Theta=\id$ now, we have
\begin{align*}
\begin{split}
\Ad\lmk W_g\rmk\circ\pi_0
&=\pi_0\circ\lmk \alpha_0\rmk\circ\Theta\circ\eta_g\beta_g^U\circ\Theta^{-1}\circ\alpha_0^{-1}\\
&=\pi_0\circ\lmk \alpha_0\rmk\circ\eta_g\beta_g^U\circ\alpha_0^{-1}\\
&=\Ad\lmk v_{g,L}\otimes v_{g,R}\rmk\circ\pi_0.
\end{split}
\end{align*}
Because 
$\pi_0$ is irreducible, (absorbing possible $\Uo$-phase in
$v_{g, L}$) we get
\[
W_g= v_{g,L}\otimes v_{g,R}.
\]
Now, we may choose $u_\sigma(g,h)$ in (\ref{udefs}) 
as
\[
u_\sigma(g,h)
=v_{g,\sigma}v_{h,\sigma} v_{gh,\sigma}^*.
\]
With this choice, we immediately see
\begin{align*}
[c_L]_{H^3(G,\Uo)}=[c_R]_{H^3(G,\Uo)}=1.
\end{align*}
Hence $1$ is split into left and right in a trivial manner.

As we have seen, if either $H_D$-$H_U$ cut or
$H_L$-$H_R$ cut is trivial, namely, product,
then our $H^3(G,\Uo)$-valued index become
trivial.
In this sense, 
$[{c_R}]_{H^3(G,\Uo)}$ indicates
how left- right/ up-down planes are correlated, 
in terms of the group action.

\subsection{Two-dimensional Dijkgraaf-Witten model}\label{dwsec}
When the unique ground state is given by an automorphism $\alpha$
satisfying the factorization property (\ref{autofac}), the $H^3(G, \Uo)$-valued index
can be calculated without going through the GNS representation.
In this section, we give such an example of $2$-dimensional SPT.
\begin{defn}
For $\sigma^{(0)}\in \Aut\lmk\caA_{\bbZ^2}\rmk$ and $\Gamma\subset\bbZ^2$,  we set
\begin{align}\label{gdeff}
\lmk d^{0}_{\Gamma}\sigma^{(0)}\rmk (g):=\lmk \sigma^{(0)}\rmk^{-1}
\beta_{g}^{\Gamma}\circ \sigma^{(0)}\circ \lmk \beta_{g}^{\Gamma}\rmk^{-1},\quad g\in G.
\end{align}
For $\sigma^{(1)}: G\to  \Aut\lmk\caA_{\bbZ^2}\rmk$ and $\Gamma\subset\bbZ^2$,  we set
\begin{align}
\lmk d^{1}_{\Gamma}\sigma^{(1)}\rmk (g,h):=
 \sigma^{(1)}(g)\beta_g^\Gamma \sigma^{(1)}(h)\lmk \beta_g^{\Gamma}\rmk^{-1}
 \lmk\sigma^{(1)}(gh)\rmk,
\end{align}
for $g,h\in G$.
\end{defn}
\begin{defn}\label{gfac}
For $\alpha\in \Aut\lmk\caA_{\bbZ^2}\rmk$, we say that $d^{0}_{H_{U}}\alpha$ is factorized into left and right if there are automorphisms
$\gamma_{g,\sigma}\in\Aut\lmk\caA_{H_{\sigma}}\rmk$, $g\in G$, $\sigma=L,R$
such that 
\begin{align}\label{dzd}
\lmk d^{0}_{H_{U}}\alpha\rmk (g)=\inn\circ \lmk \gamma_{g,L}\otimes\gamma_{g,R}\rmk,\quad
g\in G.
\end{align}
\end{defn}
Note that this situation is the same as that of Lemma \ref{cco}.
Hence
 there are unitaries $v_{\sigma}(g,h)\in \caU\lmk \caA_{H_{\sigma}}\rmk$,
$g,h\in G$, $\sigma=L,R$ such that
\begin{align}\label{usec}
\gamma_{g,{\sigma}}\beta_{g}^{{\sigma}U}\gamma_{h,{\sigma}}\beta_{h}^{{\sigma}U}
 \lmk \gamma_{gh,{\sigma}}\beta_{gh}^{{\sigma}U}\rmk^{-1}
 =\Ad\lmk v_{\sigma}(g,h)\rmk.
\end{align}
From Lemma \ref{cco} and the basic fact given in section \ref{cco} on cocycle actions, we have
\begin{align}\label{vcl}
v_R(g,h) v_R(gh,k)
=c_R(g,h,k)\gamma_{g,{R}}\beta_{g}^{{R}U}
\lmk v_R(h,k)\rmk v_R(g,hk),\quad
g,h,k\in G,
\end{align}
for some $3$-cocycle $c_R$.
The
index $[c_R]_{H^3(G,\Uo)}$ coincides with our index obtained in subsection \ref{defh3}, under the following conditions.
\begin{thm}\label{autothm}{\cite{2dSPT}}
Suppose that $\omega_{\Phi}$ has a form
 $\omega_{\Phi}=\omega_0\alpha$, where 
 $\omega_0$ is a $\beta$-invariant infinite tensor product state
 and  $\alpha$ is an automorphism.
Suppose that $\alpha$ allows a decomposition like (\ref{tfac})
for an arbitrary direction $0<\theta<\frac\pi 2$ and 
 $d^{0}_{H_{U}}\alpha$ is factorized into left and right as in (\ref{dzd})
with some $\gamma_{g,\sigma}\in \Aut\lmk\caA_{C_{\theta_{0}},\sigma}\rmk$
and $0<\theta_{0}<\frac\pi 2$, for $\sigma=L,R$.
Then we have
$h(\Phi)=[c_{R}]_{H^{3}(G,\bbT)}$.
\end{thm}
(See Theorem 6.4 of \cite{2dSPT} for a more general setting.)
\begin{proof}
We factorize $\alpha$ with respect to the angle $\theta_0$ as in  (\ref{tfac}).
Let $v_R(g,h)$ $c_R(g,h,k)$ be the unitary and $3$-cocycle given in (\ref{usec}) and (\ref{vcl}).

Because $\omega_0$ is infinite tensor product state and $\omega_0\beta_g=\omega_0$,
we have $\omega_0\beta_g^U=\omega_0$.
From this and (\ref{tfac}), (\ref{dzd}), we get
\begin{align*}
\begin{split}
\omega_\Phi\beta_g^U=\omega_0\alpha\beta_g^U
=\omega_0\lmk \beta_g^U\rmk^{-1}\alpha\beta_g^U
=\omega_0\alpha\lmk \gamma_{g^{-1},L}\otimes\gamma_{g^{-1},R}\rmk\circ\inn\\
=\omega_\Phi\lmk \gamma_{g^{-1},L}\otimes\gamma_{g^{-1},R}\rmk\circ\inn.
\end{split}
\end{align*}
This means we may take $\eta_{g,\sigma}:=\gamma_{g,\sigma}$
in Proposition \ref{betauprop}.
With this choice, because of (\ref{usec}), we may take $u(g,h)$ in Proposition \ref{defuprop} as
$u(g,h):=\pi_R\alpha_R\lmk v_R(g,h)\rmk$.
From (\ref{usec}) and the fact that our $\gamma_{g,R}$s are localized in
$C_{\theta_{0}, R}$, we see that
$v_R(g,h)$s are elements in $\caA_{C_{\theta_{0}, R}}$. 
For $W_g$ given in Proposition \ref{wprop}, we have
\begin{align*}
\begin{split}
&\Ad(W_g)\lmk\unit_{\caH_L}\otimes u(h,k)\rmk
=\Ad(W_g)\lmk\unit_{\caH_L}\otimes \pi_R\alpha_R\lmk v_R(h,k)\rmk\rmk\\
&=\pi_0
\lmk \alpha_0\rmk\circ\Theta\circ\eta_g\beta_g^U\circ\Theta^{-1}\circ\alpha_0^{-1}
\lmk \unit_{\caA_{H_L}}\otimes\alpha_R\lmk v_R(h,k)\rmk\rmk\\
&=\pi_0
\lmk \unit_{\caA_{H_L}}\otimes\lmk\alpha_R\gamma_{g,R}\beta_g^{RU} v_R(h,k)\rmk\rmk,
\end{split}
\end{align*}
because $\Theta$ is localized in $C_{\theta_0}^c$.
(We substituted $\eta_{g,\sigma}:=\gamma_{g,\sigma}$ at the last line.)
Substituting this, we get
\begin{align*}
\begin{split}
&\lmk W_g\lmk \unit_{\caH_L}\otimes u(h,k)\rmk W_g^*\rmk
\lmk \unit_{\caH_L}\otimes u(g,hk)\rmk\\
&=
 \unit_{\caA_{H_L}}\otimes \pi_R\alpha_R\lmk \gamma_{g,R}\beta_g^{RU} \lmk v_R(h,k)\rmk
  v_R(g,hk)\rmk\\
&=\overline{c_R(g,h,k)}\cdot
\unit_{\caA_{H_L}}\otimes \pi_R\alpha_R\lmk
v_R(g,h) v_R(gh,k)
\rmk\\
&=\overline{c_R(g,h,k)}\cdot
\unit_{\caA_{H_L}}\otimes 
u(g,h) u(gh,k).
\end{split}
\end{align*}
This proves the theorem.
\end{proof}

Now we apply this to the analysis of an example of two-dimensional model related to
Dijkgraaf-Witten model \cite{cglw} \cite{Beni2016} \cite{MillerMIyake2016}.
We learned about this model from Hal Tasaki.
The calculation below is carried out by him.
Analogous analysis can be carried out 
for general $d$-dimensional Dijkgraaf-Witten model (or more generally,
for models with some abelian structures)
to derive a $H^{d+1}(G,\Uo)$-valued index.

The standard basis of $\bbR^2$ is denoted by
$\{\bm e_{j}\}_{j=1,2}$.
Let us recall the canonical triangulation of the unit hypercube $[0,1]^2$.
We denote by $\sym(2)$ the symmetric group of order $2$.
For $\pi\in \sym(2)$, $\sgn(\pi)$ denotes the signature of the permutation.
For each $\pi\in \sym(2)$, we set
\begin{align}
T_{\pi}:=\left\{
\bm{x}=(x_i)_{i=1}^2\in [0,1]^2\mid
x_{\pi(2)}\le x_{\pi(1)}
\right\}.
\end{align}
The unique integral corners of $T_\pi$ are the following $3$ points:
\begin{align}\label{vtachi}
\begin{array}{cccc}
\;&\; & x_{\pi(2)}&x_{\pi(1)}\\
{\bm{v}}_0^{(\pi)} &:=&0&0\\
{\bm{v}}_1^{(\pi)} &:=&0&1\\
{\bm{v}}_{2}^{(\pi)} &:=&1&1\\
\end{array}
\end{align}
The simplex $T_\pi$ is equal to the convex combination $\co\{{\bm{v}}_{k}^{(\pi)}\}_{k=1}^2$
of vectors $\{{\bm{v}}_{k}^{(\pi)}\}_{k=1}^2$.
The simplices give a simplicization $[0,1]^2:=\cup_{\pi\in\sym(2)} T_\pi$.

For each $\pi\in\sym(2)$ and $j=0,1,2$, we denote by
$S_j^{(\pi)}$ the convex combination $\co\{{\bm{v}}_{k}^{(\pi)}\mid k\neq j\}$
of vectors $\{{\bm{v}}_{k}^{(\pi)}\}_{k\neq j}$.
By (\ref{vtachi}), we have
\begin{align}
&S_0^{(\pi)}:=\left\{
\bm{x}=(x_i)_{i=1}^2\in [0,1]^2\mid
x_{\pi(1)}=1
\right\},\label{szp}\\
&S_1^{(\pi)}
=\left\{
\bm{x}=(x_i)_{i=1}^2\in [0,1]^2\mid
 x_{1}=x_{2}
\right\},\label{sjp}\\
&S_2^{(\pi)}:=\left\{
\bm{x}=(x_i)_{i=1}^2\in [0,1]^2\mid
x_{\pi(2)}=0
\right\}.\label{smp}
\end{align}
For $\pi,\pi'\in\sym(2)$ with $\pi\neq\pi'$,
we have $T_\pi\cap T_{\pi'}=S_1^{(\pi)}=S_1^{(\pi')}$.
For $\pi,\pi'\in\sym(2)$ and $j=1,2$, 
$T_\pi\cap\lmk T_{\pi'}+\bm{e}_j\rmk$ is $1$-dimensional
if and only if $\pi'\neq\pi$ and $\pi(1)=\pi'(2)=j$.
In this case, we have
$T_\pi\cap\lmk T_{\pi'}+\bm{e}_j\rmk
=S_0^{(\pi)}
=S_2^{(\pi')}+\bm{e}_j
$.
For $\bm x\in \bbZ^{2}$,  $\pi\in\sym(2)$, and $k=0,1,2$, we set
$\bm{v}_{k,\bm x}^{(\pi)}:=\bm{v}_k^{(\pi)}+\bm x\in\bbZ^2$, and
$
S_{k,\bm x}^{(\pi)}:=S_k^{(\pi)}+\bm x=\co\{ \bm{v}_{i,\bm x}^{(\pi)}\mid i\neq k\}
$.

For any $(\bm x,\pi,k)\in \bbZ^2\times \sym(2)\times\{0,1,2\}$,
there exists a unique 
$(\bm x',\pi',k')\in \bbZ^2\times \sym(2)\times\{0,1,2\}$
with
$(\bm x,\pi,k)\neq (\bm x',\pi',k')$
such that
$S_{k,\bm x}^{(\pi)}=S_{k',\bm x'}^{(\pi')}$.
Furthermore, for this  $(\bm x',\pi',k')$, one of the following occurs.
\begin{description}
\item[(i)]
We have $\bm x=\bm x'$, $k=k'=1$,
$\pi'=\pi\circ(1,2)$, and $\bm{v}_{j,\bm x}^{(\pi)}=\bm{v}_{j,\bm x'}^{(\pi')}$,
for any $j=0,2$.
\item[(ii)]
We have $\bm x'=\bm x+\bm e_{\pi(1)}$, $k=0$, 
$k'=2$,
$\pi'=\pi\circ(1,2)$
and $\bm{v}_{j,\bm x}^{(\pi)}=\bm{v}_{j-1,\bm x'}^{(\pi')}$,
for $j=1,2$.
\item[(iii)]
We have $\bm x=\bm x'+\bm e_{\pi'(1)}$, $k=2$, 
$k'=0$,
$\pi=\pi'\circ(1,2)$
and $\bm{v}_{j-1,\bm x}^{(\pi)}=\bm{v}_{j,\bm x'}^{(\pi')}$,
for $j=1,2$.
\end{description}
For such
$(\bm x,\pi,k)\in \bbZ^2\times \sym(2)\times\{0,1,2\}$,
and $(\bm x',\pi',k')\in \bbZ^2\times \sym(2)\times\{0,1,2\}$ 
we write $(\bm x,\pi,k)\sim (\bm x',\pi',k')$.
Each $(\bm x,\pi,k)\in \bbZ^2\times \sym(2)\times\{0,1,2\}$
has exactly one $(\bm x',\pi',k')\in \bbZ^2\times \sym(2)\times\{0,1,2\}$
satisfying $(\bm x,\pi,k)\sim (\bm x',\pi',k')$.

Now let us introduce the model related to the Dijkgraaf-Witten model.
Let $G$ be a fixed finite group.
Let $\delta_g$, $g\in G$ be the orthogonal basis of $l^2(G)$,
and $\{e_{g,h}\in B(l^2(G))\mid g,h\in G\}$ the system of matrix units associated to the basis $\{\delta_g\mid g\in G\}$.
We denote by $U$ the left action of $G$ on $l^2(G)$, i.e.,
$U_g \delta_h:=\delta_{gh}$, for any $g,h\in G$.
Note that the algebra of diagonal matrices $\Dia_d$ with respect to the basis $\{\delta_g\mid g\in G\}$
is invariant under $\Ad(U_g)$, $g\in G$.
For any finite $\Lambda\subset \bbZ^2$ and $\bm s=(s\lmk \bm x\rmk)_{\bm x\in\Lambda},
\bm s'=(s'\lmk {\bm x}\rmk)_{\bm x\in\Lambda}\in G^{\times \Lambda}$,
we set $e^{(\Lambda)}_{\bm s,\bm s'}
:=\bigotimes_{\bm x\in\Lambda} e_{s\lmk {\bm x}\rmk,{\bm s'\lmk \bm x\rmk}}$.
For each $\Gamma\subset\bbZ^2$, we denote by  $\caD_\Gamma$   the abelian subalgebra of $\caA_\Gamma$ generated by the tensor product of 
the diagonal matrices.
For each $L\in\bbN$,
set
\begin{align*}
\begin{split}
&\Lambda_L^{(1)}:=[-L,L]\quad
\partial\Lambda_L^{(1)}:=\{-L,L\},\\
&\Lambda_L^{(2)}:=[-L,L]^2\quad
\partial\Lambda_L^{(2)}:=
 \lmk \{-L,L\}\times [-L,L]\rmk\cup \lmk[-L,L]\times  \{-L,L\}\rmk.
\end{split}
\end{align*}
Here, with a bit abuse of notation, we denote $[-L,L]^k\cap\bbZ^k$ 
by $[-L,L]^k$. Analogous notation will be used for $[0,L]$, $[-L,-1]$
etc below.

Let us recall a basic fact about group cohomology.
The group ring $\bbZ[G^{\times n+1}]$ is a $G$-module with a $G$-action
\[
g\cdot (g_0,\cdots, g_{n})=(gg_0,\cdots,gg_n ).
\]
We associate a trivial action of $G$ on $\Uo$.
For a homomorphism $\varphi\in \Hom_{\bbZ[G]}\lmk \bbZ[G^{\times i+1}],\Uo\rmk$
from a $G$-module $ \bbZ[G^{\times i+1}]$ and to a $G$-module $\Uo$
we have
\begin{align}\label{cohomog}
\varphi(g g_0, gg_1,\ldots, gg_i)=
\varphi( g_0, g_1,\ldots, g_i).
\end{align}
Define $D^i: \Hom_{\bbZ[G]}\lmk \bbZ[G^{\times( i+1)}],\Uo\rmk\to
\Hom_{\bbZ[G]}\lmk \bbZ[G^{\times i+2}],\Uo\rmk$
 by
\begin{align*}\label{dfd}
\begin{split}
&\lmk D^i\varphi\rmk\lmk g_0, g_1,\cdots, g_{i}, g_{i+1}\rmk
:=
\prod_{j=0}^{i+1}\varphi\lmk g_0, \cdots,g_{j-1}, \hat g_j,g_{j+1},\cdots g_{i+1}\rmk^{(-1)^j},
\end{split}
\end{align*}
where $\hat g_j$ means that $g_j$ is excluded.
The group of $i$-cochains of $G$ with coefficients in $\Uo$ is the set of functions from
$G^{\times i}$ to $\Uo$.
The $i$-th differential 
$d^i : C^i\lmk G,\Uo\rmk 
\to C^{i+1}\lmk G, \Uo\rmk 
$
is given by
\begin{align*}
\begin{split}
&\lmk d^i_{\Gamma_0}\sigma\rmk\lmk g_0,\ldots, g_i\rmk\\
&=g_0\cdot\sigma\lmk g_1,\ldots,g_i\rmk\lmk
\prod_{k=1}^i \sigma(g_0,\ldots,g_{k-1}g_k,\ldots,g_i)^{(-1)^k}\rmk
\sigma\lmk g_0,\ldots,g_{i-1}\rmk^{(-1)^{i+1} }.
\end{split}
\end{align*}
There are isomorphisms 
$\Psi^i:\Hom_{\bbZ[G]}\lmk \bbZ[G^{\times i+1}],\Uo\rmk\to C^{i}\lmk G, \Uo\rmk$, $i=0,1,2,\ldots$
such that 
\begin{align*}
\Psi^i(\varphi)(g_0,\ldots, g_{i-1})
=\varphi\lmk
e, g_0, g_0g_1,g_0g_1g_2,\cdots, g_0g_1g_2\cdots g_{i-1}
\rmk.
\end{align*}
With this isomorphism, we have
\[
\Psi^{i+1}\circ D^i=d^i\circ \Psi^i
\]
 for all $i=0,1,2,\ldots$.

The two-dimensional Dijkgraaf-Witten model
is defined by fixing an arbitrary $\nu\in \ker D^{3}=\lmk \Psi^{3}\rmk^{-1}\lmk Z^{3}\lmk G, \Uo\rmk\rmk$.
For the rest of this section, we fix such $\nu$.
By definition, we have
\begin{align}\label{dfd}
\frac{\nu(g_1,g_2, g_3,g_4)\nu(g_0,g_1,g_3,g_4)\nu(g_0,g_1,g_2,g_3)}{\nu(g_0,g_2,g_3,g_4)\nu(g_0,g_1,g_2,g_4)}=1,
\end{align}
for any $g_0,g_1,g_2,g_3,g_4,g_5\in G$.

For $\bm s : \bbZ^2 \to G$, $\bm s=(s(\bm v))_{\bm v\in\bbZ^2}$, set
\begin{align*}
\begin{split}
&{\mathfrak q}\lmk\bm s, \bm{x} \rmk
:=\prod_{\pi\in \sym(2)}\nu\lmk e,
s\lmk {{\bm{v}}_{0,\bm{x}}^{\lmk \pi\rmk} }\rmk,
s\lmk {{\bm{v}}_{1,\bm{x}}^{\lmk \pi\rmk} }\rmk,
s\lmk {{\bm{v}}_{2,\bm{x}}^{\lmk \pi\rmk} }\rmk
\rmk^{\sgn\pi},\; \bm{x}\in \bbZ^2\quad\text{and}\\
&{\mathfrak p}\lmk g, \bm s, {y}\rmk
:=\nu\lmk
e, g,
s\lmk ( {y},0)\rmk,
s\lmk ( {y}+1,0)\rmk
\rmk,\quad g\in G,\quad y\in\bbZ.
\end{split}
\end{align*}
For each $L\in\bbN$, set
\begin{align*}
\begin{split}
&V_L^{(0)}:=
\sum_{\bm s\in G^{\times \lmk \Lambda_{L+1}^{(2)}\rmk}}
\;
\lmk \prod_{\bm x \in \Lambda_L^{(2)}}\;
\lmk
{\mathfrak q}\lmk\bm s, \bm{x} \rmk
\rmk\rmk\cdot  e_{\bm s, \bm s}^{\lmk \Lambda_{L+1}^{(2)}\rmk},\\
&V_L^{(1)}\lmk g\rmk=
\sum_{\bm s\in G^{\times \lmk \Lambda_{L+1}^{(1)}\times\{\bm 0\}\rmk}}
\;
\lmk \prod_{ y\in \Lambda_L^{(1)}}
\lmk
{\mathfrak p}\lmk g, \bm s, {y}\rmk
\rmk^{-1} \rmk\cdot e_{\bm s, \bm s}^{\lmk \Lambda_{L+1}^{(1)}\times\{\bm 0\}\rmk}.
\end{split}
\end{align*}
We also introduce
\begin{align*}
\begin{split}
&V_{+,L}^{(0)}:=
\sum_{\bm s\in G^{\times \lmk \Lambda_{L+1}^{(1)}\times [0,L+1]\rmk}}\;\;
\lmk
\prod_{\bm x \in \Lambda_L^{(1)}\times [0,L]}\;\;
\lmk
{\mathfrak q}\lmk\bm s, \bm{x} \rmk
\rmk\rmk\cdot e_{\bm s, \bm s}^{\lmk \Lambda_{L+1}^{(1)}\times [0,L+1]\rmk},\\
&V_{+,L}^{(1)}\lmk g\rmk
:=
\sum_{\bm s\in G^{\times \lmk  [0,L+1]\times\{\bm 0\}\rmk}}\;\;
\lmk \prod_{ y\in [0,L]}\;\;
\lmk
{\mathfrak p}\lmk g, \bm s, {y}\rmk
\rmk^{-1}\rmk\cdot e_{\bm s, \bm s}^{\lmk  [0,L+1]\times\{\bm 0\}\rmk},\\
\end{split}
\end{align*}
and
\begin{align*}
\begin{split}
&V_{-,L}^{(0)}:=\sum_{\bm s\in G^{\times \lmk \Lambda_{L+1}^{(1)}\times[-L,-1]\rmk}}\;\;
\lmk \prod_{\bm x \in \Lambda_L^{(1)}\times[-L,-2]}\;\;
\lmk
{\mathfrak q}\lmk\bm s, \bm{x} \rmk
\rmk \rmk\cdot e_{\bm s, \bm s}^{ \lmk \Lambda_{L+1}^{(1)}\times[-L,-1]\rmk},\\
&V_{-,L}^{(1)}\lmk g\rmk:=
\sum_{\bm s\in G^{\times \lmk [-L,-1]\times\{\bm 0\}\rmk}}\;\;
\lmk \prod_{ y\in [-L,-2]}
\quad\lmk
{\mathfrak p}\lmk g, \bm s, {y}\rmk
\rmk^{-1}\rmk\cdot e_{\bm s, \bm s}^{\lmk [-L,-1]\times\{\bm 0\}\rmk}.
\end{split}
\end{align*}
Furthermore, set
\begin{align*}
\begin{split}
&V_{\partial, L}^{(0)}
:=
\sum_{\bm s\in G^{\times \lmk \Lambda_{L+1}^{(1)}\times \{-1,0\}\rmk}}
\lmk
\prod_{\bm x \in \Lambda_L^{(1)}\times \{-1\}}\;
{\mathfrak q}\lmk\bm s, \bm{x} \rmk\rmk
e_{\bm s, \bm s}^{ \lmk \Lambda_{L+1}^{(1)}\times \{-1,0\}\rmk},\\
&V_{\partial, L,+}^{(0)}
:=
\sum_{\bm s\in G^{\times \lmk [0,L+1]\times \{-1,0\}\rmk}}
\lmk \prod_{\bm x \in [0,L]\times \{-1\}}\;
\lmk
{\mathfrak q}\lmk\bm s, \bm{x} \rmk
\rmk\rmk
e_{\bm s, \bm s}^{ \lmk [0,L+1]\times \{-1,0\}\rmk},\\
&V_{\partial, L,-}^{(0)}
:=
\sum_{\bm s\in G^{\times \lmk [-L,-1]]\times \{-1,0\}\rmk}}
\lmk \prod_{\bm x \in [-L,-2]
\times \{-1\}}\;
\lmk
{\mathfrak q}\lmk\bm s, \bm{x} \rmk
\rmk\rmk
e_{\bm s, \bm s}^{ \lmk [-L,-2]\times \{-1,0\}\rmk},\\
\\
&V_{\partial, L,0}^{(0)}
:=
\sum_{\bm s\in G^{\times \lmk \{-1,0\}\times \{-1,0\}\rmk}}
\lmk\prod_{\bm x =(-1,-1)}
\lmk
{\mathfrak q}\lmk\bm s, \bm{x} \rmk
\rmk\rmk
e_{\bm s, \bm s}^{ \lmk \{-1,0\}\times \{-1,0\}\rmk},\\
&V_{\partial, L}^{(1)}\lmk g\rmk
:=
\sum_{\bm s\in G^{\times \lmk \{-1,0\}\times\{\bm 0\}\rmk}}
\lmk{\mathfrak p}\lmk g, \bm s, {-1}\rmk
\rmk^{-1}e_{\bm s, \bm s}^{\lmk \{-1,0\}\times\{\bm 0\}\rmk}.
\end{split}
\end{align*}
Finally, for a unitary $u\in\caU(\caA_{\bbZ^2})$, we set
\begin{align*}
\lmk \widetilde{d^0}{u}\rmk\lmk g \rmk
={u}^{-1 }\beta_g^U\lmk {u}\rmk
,\quad g\in G.
\end{align*}
For these objects, the following Lemma holds.
\begin{lem}\label{vml}
In the setting above,
\begin{description}
\item[(i)] For any  $5\le L\in\nan$ and  $g\in G$,
\begin{align}
\begin{split}
&V_L^{(0)}
=V_{\partial, L}^{(0)}\lmk V_{-,L}^{(0)}\otimes V_{+,L}^{(0)}\rmk,\\
&V_L^{(1)}(g)
=V_{\partial, L}^{(1)}(g)
\lmk V_{-,L}^{(1)}(g)
\otimes V_{+,L}^{(1)}(g)\rmk,\\
&V_{\partial, L}^{(0)}= V_{\partial, L,0}^{(0)}\lmk V_{\partial, L,-}^{(0)}\otimes V_{\partial, L,+}^{(0)}\rmk
\end{split}
\end{align}
\item[(ii)]
For any  $5\le L\in\nan$ and  $g\in G$,
\begin{align}
&\lmk
\tilde d^0 V_{+,L}^{(0)}\rmk (g)
=\lmk \text{unitary in }\caD_{\partial\Lambda_L} \rmk\cdot
V_{L}^{(1)} (g)\\
&\lmk
\tilde d^0 V_{L}^{(0)}\rmk (g)
=\lmk \text{unitary in }\caD_{\partial\Lambda_L} \rmk.\label{bulkv}
\end{align}
\item[(iii)]
For any $5\le L\in\nan$ and  $(g, h)\in G^{\times 2}$,
\begin{align}
V_{+,L}^{(1)}(g)\cdot\lmk \beta_g^U\lmk V_{+,L}^{(1)}(h)\rmk\rmk\cdot
\lmk V_{+,L}^{(1)}(gh)\rmk^{-1}
=\lmk \text{unitary in }\caD_{\partial\Lambda_L^{(1)}} \rmk\cdot
V\lmk
g, h
\rmk,
\end{align}
where
\begin{align}\label{ud}
V\lmk
g,h
\rmk
:=\sum_{s\in G} 
\nu\lmk
e, g, 
gh,s
\rmk \cdot e_{s,s}^{\{\bm 0\}}.
\end{align}
\item[(iv)]
For the unitary $V$ in (\ref{ud}), we have
\begin{align}
\overline{\nu(e,g,gh,ghk)} V(g,h)
 V(gh,k)
 =
 \beta_{g}^{{\sigma}U}\lmk V(h,k)\rmk
 V\lmk g, hk\rmk,\quad g,h,k\in G.
 \end{align}
\end{description}
\end{lem}
\begin{proof}
(i) is trivial from the definitions.\\
To see (ii), by definition, we first have
\begin{align}
\begin{split}
\tilde d^0V_{+,L}^{(0)}(g)=&
{V_{+,L}^{(0)}}^{-1 }\beta_g^U\lmk {V_{+,L}^{(0)}}\rmk\\
=&
\sum_{\bm s\in G^{\times \lmk \Lambda_{L+1}^{(1)}\times [0,L+1]\rmk}}\;\;
\prod_{\bm x \in \Lambda_L^{(1)}\times [0,L]}\;\;
\prod_{\pi\in \sym(2)}\\
&\nu\lmk
g,
s\lmk {{\bm{v}}_{0,\bm{x}}^{\lmk \pi\rmk} }\rmk,
s\lmk {{\bm{v}}_{1,\bm{x}}^{\lmk \pi\rmk} }\rmk,
s\lmk {{\bm{v}}_{2,\bm{x}}^{\lmk \pi\rmk} }\rmk
\rmk^{\sgn{\pi}}\\
&\nu\lmk
e, 
s\lmk {{\bm{v}}_{0,\bm{x}}^{\lmk \pi\rmk} }\rmk,
s\lmk {{\bm{v}}_{1,\bm{x}}^{\lmk \pi\rmk} }\rmk,
s\lmk {{\bm{v}}_{2,\bm{x}}^{\lmk \pi\rmk} }\rmk
\rmk^{-{\sgn{\pi}}}
 e_{\bm s, \bm s}^{\lmk \Lambda_{L+1}^{(1)}\times [0,L+1]\rmk}.
\end{split}
\end{align}
The first $\nu$ comes
  from $\beta_g^{U}\lmk {V_{+,L}^{(0)}}\rmk$
 and 
we used change of varianles $g^{-1}s\to s$ and (\ref{cohomog}).
Because $\nu\in \ker D^{3}$, using (\ref{dfd}), we obtain
\begin{align}\label{nust}
\begin{split}
&\tilde d^0V_{+,L}^{(0)}\lmk g\rmk\\&=
\sum_{\bm s\in G^{\times \lmk \Lambda_{L+1}^{(1)}\times [0,L+1]\rmk}}\;\;
\prod_{\bm x \in \Lambda_L^{(1)}\times [0,L]}\;\;
\prod_{\pi\in \sym(2)}\\
&
\quad\lmk
\begin{gathered}
\nu\lmk
e, g,
s\lmk {\bm{v}}_{1,\bm{x}}^{\lmk \pi\rmk} \rmk,
s\lmk {\bm{v}}_{2,\bm{x}}^{\lmk \pi\rmk} \rmk
\rmk^{-1}
\nu\lmk
e, g,
s\lmk {\bm{v}}_{0,\bm{x}}^{\lmk \pi\rmk} \rmk,
s\lmk {\bm{v}}_{2,\bm{x}}^{\lmk \pi\rmk} \rmk
\rmk\\
\nu\lmk
e, g,
s\lmk {\bm{v}}_{0,\bm{x}}^{\lmk \pi\rmk} \rmk,
s\lmk {\bm{v}}_{1,\bm{x}}^{\lmk \pi\rmk} \rmk\rmk^{-1}
\end{gathered}
\rmk^{\sgn{\pi}}
\quad\quad e_{\bm s, \bm s}^{\lmk \Lambda_{L+1}^{(1)}\times [0,L+1]\rmk}.
\end{split}
\end{align}
Now for each $\bm s : \bbZ^2\to G$, $g\in G$ and 
$(\bm x,\pi,j)\in \lmk \Lambda_L^{(1)}\times [0,L]\rmk\times \sym(2)
\times \{0,1,2\}$, we set
\begin{align*}
\psi_{\bm s, g}\lmk \bm x,\pi,j\rmk:=\left\{
\begin{gathered}
\nu\lmk
e, g,
s\lmk {\bm{v}}_{1,\bm{x}}^{\lmk \pi\rmk} \rmk,
s\lmk {\bm{v}}_{2,\bm{x}}^{\lmk \pi\rmk} \rmk
\rmk^{-\sgn{\pi}},\quad \text{if} \quad j=0\\
\nu\lmk
e, g,
s\lmk {\bm{v}}_{0,\bm{x}}^{\lmk \pi\rmk} \rmk,
s\lmk {\bm{v}}_{2,\bm{x}}^{\lmk \pi\rmk} \rmk
\rmk^{\sgn{\pi}}\quad \text{if} \quad j=1\\
\nu\lmk
e, g,
s\lmk {\bm{v}}_{0,\bm{x}}^{\lmk \pi\rmk} \rmk,
s\lmk {\bm{v}}_{1,\bm{x}}^{\lmk \pi\rmk} \rmk
\rmk^{-\sgn{\pi}}\quad \text{if} \quad j=2
\end{gathered}
\right..
\end{align*}
With this notation, we have
\begin{align}
\begin{split}
&\tilde d^0V_{+,L}^{(0)}\lmk g\rmk\\&=
\sum_{\bm s\in G^{\times \lmk \Lambda_{L+1}^{(1)}\times [0,L+1]\rmk}}\;\;
\lmk
\prod_{\bm x \in \Lambda_L^{(1)}\times [0,L]}\;\;
\prod_{\pi\in \sym(2)}\prod_{j=0}^2
\psi_{\bm s, g}\lmk \bm x,\pi,j\rmk\rmk \cdot 
 e_{\bm s, \bm s}^{\lmk \Lambda_{L+1}^{(1)}\times [0,L+1]\rmk}
\end{split}
\label{igr}
\end{align}
By the observation about
 $(\bm x,\pi,j)\sim (\bm {x'},{\pi'},j')$,
we see that if $(\bm x,\pi,j)\sim (\bm {x'},{\pi'},j')$,
then $\psi_{\bm s, g}\lmk \bm x,\pi,j\rmk
\psi_{\bm s, g}\lmk \bm {x'},{\pi'},j'\rmk=1$.
Therefore, a term $\psi_{\bm s, g}(\bm x,\pi,j)$ corresponding to $(\bm x,\pi,j)$ cancels out with that of $ (\bm {x'},{\pi'},j')$
, if the term $\psi_{\bm s, g}(\bm{ x'},{\pi'},j')$  also appears in (\ref{igr}).
Recall that there is exactly one  $ (\bm {x'},{\pi'},j')$ such that $(\bm x,\pi,j)\sim (\bm {x'},{\pi'},j')$.
The terms
which may not disappear is that of either $\bm x\in \partial\Lambda_L^{(2)}$
or 
\begin{align}
\bm x\in \Lambda_L^{(1)}\times \{0\},\quad\text{and}\quad
j=2,\quad\text{and}\quad \pi=\id.
\end{align}
The former terms form some unitary in $\caD_{\partial\Lambda_L^{(2)}}$.
For the latter case,
setting $ y\in \Lambda_L^{(1)}$ so that $( y,0)=\bm x$,
we obtain
\begin{align}
{\bm{v}}_{0,\bm{x}}^{\lmk \pi\rmk}
=\lmk
y,0
\rmk,\quad 
{\bm{v}}_{1,\bm{x}}^{\lmk \pi\rmk}
=\lmk
y+1,0
\rmk.
\end{align}

Hence we have
\begin{align*}
\begin{split}
&\tilde d^0V_{+,L}^{(0)}\lmk g
\rmk\\
&=\lmk \text{unitary in }\caD_{\partial\Lambda_L^{(2)}} \rmk\cdot
\sum_{\bm s\in G^{\times \lmk \Lambda_{L+1}^{(1)}\times\{\bm 0\}\rmk}}\\
&\quad\prod_{y \in \Lambda_L^{(1)}}
\nu\lmk
\begin{gathered}
e, g,
s\lmk(y,0)
\rmk,
s\lmk
(y+1,0)
\rmk
\end{gathered}
\rmk^{-1}\cdot  e_{\bm s, \bm s}^{\lmk \Lambda_{L+1}^{(1)}\times\{\bm 0\}\rmk}
\\
&
=\lmk \text{unitary in }\caD_{\partial\Lambda_L^{(2)}} \rmk\cdot
V_{L}^{(1)} (g).
\end{split}
\end{align*}
This completes the proof of the first part of (ii).
The proof for the second part of (ii) is the same,
except for this time, $[0,L]$ is replaced by $[-L,L]$.
Therefore, there is no contribution from out side of $\partial\Lambda_L^{(2)}$.

The proof of (iii) is basically the same.
We have
\begin{align*}
\begin{split}
&V_{+,L}^{(1)}(g)\cdot\lmk \beta_g^U\lmk V_{+,L}^{(1)}(h)\rmk\rmk\cdot
\lmk V_{+,L}^{(1)}(gh)\rmk^{-1}\\
&=
\sum_{\bm s\in G^{\times \lmk  [0,L+1]\times\{0\}\rmk}}\;\;
\prod_{y\in  [0,L]}\\
&
\quad\lmk
\begin{gathered}
\nu\lmk
g, gh,
s(y,0),
s( y+1,0)
\rmk^{-1}\\
\nu\lmk
e, gh,
s( y,0),
s( y+1,0)
\rmk
\nu\lmk
e, g,
s( y,0),
s(y+1,0)
\rmk^{-1}
\end{gathered}
\rmk
e_{\bm s, \bm s}^{   [0,L+1]\times\{0\}}\\
&=\sum_{\bm s\in G^{\times \lmk [0,L+1]\times\{\bm 0\}\rmk}}\;\;
\prod_{y\in  [0,L]}\lmk
\nu\lmk
e, g, gh,
s(y+1,0)
\rmk^{-1}
\nu\lmk
e, g, gh,
s( y,0)
\rmk
\rmk e_{\bm s, \bm s}^{   [0,L+1]\times\{0\}}\\
&=\lmk \text{unitary in }\caD_{\partial\Lambda_L^{(1)}} \rmk
\sum_{s\in G} 
\nu\lmk
e, g, 
gh,s
\rmk \cdot e_{s,s}^{\{\bm 0\}}.
\end{split}
\end{align*}
For the second equality, we used (\ref{dfd}).
Third equality is just a simple cancellation.

(iv) can be checked directly, using (\ref{dfd}).
\end{proof}
We can easily take the thermodynamic limit.
\begin{lem}\label{vvl}
There are automorphisms
\begin{align*}
\begin{split}
&\sigma^{(0)}\in \Aut\lmk {\caA_{\bbZ^{2}}}\rmk,\quad
\sigma_+^{(0)}
\in \Aut\lmk \caA_{ \bbZ\times[0,\infty)} \rmk,\quad
\sigma_-^{(0)}
\in \Aut\lmk \caA_{ \bbZ\times(-\infty,-1]} \rmk
\\
&\Xi^{(0)}\in
\Aut\lmk \caA_{ \bbZ^{(1)}\times\{-1,0\}}\rmk,\quad
\Xi_+^{(0)}\in
\Aut\lmk \caA_{ {[0,\infty)}\times\{-1,0\}}\rmk,
\Xi_-^{(0)}\in
\Aut\lmk \caA_{ (-\infty,-1]\times\{-1,0\}}\rmk,
\quad\\
&\sigma_+^{(1)}(g)
\in \Aut\lmk \caA_{ [0,\infty)\times\{0\}} \rmk,\quad
\sigma_-^{(1)}(g)
\in \Aut\lmk \caA_{(-\infty,-1]\times\{0\}} \rmk
\end{split}
\end{align*}
such that
\begin{align}\label{tdl}
\begin{split}
&\sigma^{(0)}(A)=\lim_{L\to\infty}\Ad\lmk V_L^{(0)}\rmk(A),\quad
\sigma_+^{(0)}(A)=\lim_{L\to\infty}\Ad\lmk V_{+,L}^{(0)}\rmk(A),\\
&\sigma_-^{(0)}(A)=\lim_{L\to\infty}\Ad\lmk V_{-,L}^{(0)}\rmk(A),\quad
\Xi^{(0)}(A)
=\lim_{L\to\infty}\Ad\lmk V_{\partial L}^{(0)}\rmk(A)
,\\
&\Xi_+^{(0)}(A)
=\lim_{L\to\infty}\Ad\lmk V_{\partial, L,+}^{(0)}\rmk(A),\quad
\Xi_-^{(0)}(A)
=\lim_{L\to\infty}\Ad\lmk V_{\partial, L,-}^{(0)}\rmk(A),\\
&\sigma^{(1)}(g)(A)
=\lim_{L\to\infty}\Ad\lmk V_L^{(1)}\lmk g\rmk\rmk(A),\quad 
\sigma_+^{(1)}\lmk g\rmk(A)
=\lim_{L\to\infty}\Ad\lmk V_{+,L}^{(1)}\lmk g \rmk\rmk(A),\\
&\sigma_-^{(1)}\lmk g\rmk(A)
=\lim_{L\to\infty}\Ad\lmk V_{-,L}^{(1)}\lmk g \rmk\rmk(A),\quad
\Xi^{(1)}\lmk g\rmk(A)
=\lim_{L\to\infty}\Ad\lmk V_{\partial, L}^{(1)}\lmk g\rmk\rmk(A),
\end{split}
\end{align}
for any $g\in G$, $L\in\bbN$, and $A\in\caA_{\Lambda_{L-1}^{(2)}}$.\\
Furthermore, the following holds:
\begin{description}
\item[(i)]
For each 
\begin{align}
\begin{split}
&\sigma^{(0)}=
\Xi^{(0)}
\lmk
\sigma_-^{(0)}
\otimes
\sigma_+^{(0)}\rmk,\quad
\Xi^{(0)}=\inn\lmk \Xi_-^{(0)}\otimes\Xi_+^{(0)}\rmk,\\
&\sigma^{(1)}\lmk g\rmk=
\Xi^{(1)}\lmk g\rmk\circ
\lmk
\sigma_-^{(1)}\lmk g \rmk
\otimes
\sigma_+^{(1)}\lmk g\rmk\rmk.
\end{split}
\end{align}
\item[(ii)]
For each $g\in G$,
\begin{align}
d^1_{H_U} \sigma_+^{(0)}\lmk g\rmk
=\sigma^{(1)}\lmk g\rmk.
\end{align}
\item[(iii)]
\begin{align}
\sigma_+^{(1)}(g)\beta_g^U\sigma_+^{(1)}\lmk h\rmk\lmk\beta_g^{U}\rmk^{-1}\lmk \sigma_+^{(1)}(gh)\rmk^{-1}
=\Ad\lmk V\lmk g,h\rmk\rmk
\end{align}
\end{description}
\end{lem}
\begin{proof}
For any $M, L\in\bbN$ with $M+1\ge L$, $ V_L^{(m)}$ is of the form
\begin{align}
 V_L^{(m)}= V_{M+1}^{(m)}\lmk \text{unitary in }\quad \caD_{\Lambda_L^{(d)}\setminus \Lambda_M^{(d)}}\rmk.
\end{align}
From this and the fact that
$V_L^{(m)}, V_{M+1}^{(m)}$ belong to a common abelian subalgebra, we have
\begin{align}
\Ad\lmk {V_{L}^{(m)}}\rmk (A)
=\Ad\lmk {V_{M+1}^{(m)}}\rmk (A),\quad
\Ad\lmk {\lmk V_{L}^{(m)}\rmk^* }\rmk (A)
=\Ad\lmk {\lmk V_{M+1}^{(m)}\rmk^*}\rmk (A),
\end{align}
for all $A\in\caA_{\Lambda_M^{(2)}}$.
Therefore, taking $L\to\infty$ limit, we obtain an automorphism
satisfying (\ref{tdl}).
The existence of other automorphisms are the same.
The properties (i), (ii), (iii) corresponds to the (i), (ii), (iii) of Lemma \ref{vml}.
\end{proof}
Let 
\begin{align}
\xi:=\frac{1}{\sqrt{|G|}}\sum_{g\in G} \delta_g,
\end{align}
and set a state $\rho$ on $\caB(l^2(G))$ such that
\begin{align}
\rho(A)
=\braket{\xi}{A\xi},\quad A\in \caB(l^2(G)).
\end{align}
Combining Lemma \ref{vml} and Lemma \ref{vvl}, we obtain the following:
\begin{thm}
Let $\sigma^{(0)}$ be an automorphism given by Lemma \ref{vvl}.
Let $\omega_0$ be an infinite tensor product state on 
$\caA_{\bbZ^2}:=\bigotimes_{\bm x\in \bbZ^2}\caB(l^2(G))$
defined by
\begin{align}
\omega_0:=\bigotimes_{\bm x\in\bbZ^2} \rho
\end{align}
Then the $H^3(G,\Uo)$-valued index of the model is
$
\left[ 
  \Psi^{3}\lmk \nu
\rmk\right]_{H^{3}(G,\Uo)}
$.
\end{thm}
\begin{proof}
Note that all the automorphisms in Lemma \ref{vvl} and their conjugation with respect to
$\beta_g^U$ mutually commute.
Using  Lemma \ref{vvl}, we obtain
\begin{align}
\begin{split}
&d^0_{H_U}\lmk \sigma^{(0)}\rmk(g)
=
d^0_{H_U}\lmk 
\Xi^{(0)}
\lmk
\sigma_-^{(0)}
\otimes
\sigma_+^{(0)}
\rmk
\rmk(g)\notag\\
&=
\inn\circ
\lmk
d^0_{H_{U}\cap H_L}\Xi^{(0)}_-(g)
\otimes
d^0_{H_{U}\cap H_R} \Xi^{(0)}_+(g)
\rmk
\circ
 d^0_{H_U\cap H_R} \sigma_+^{(0)}\lmk g\rmk
\\
&=
\inn\circ
\lmk
d^0_{H_{U}\cap H_L}\Xi^{(0)}_-(g)\circ\sigma_-^{(1)}(g)
\otimes
d^0_{H_{U}\cap H_R} \Xi^{(0)}_+(g)\circ\sigma_+^{(1)}(g)
\rmk
\end{split}
\end{align}
Hence
$d^{0}_{H_{U}}\sigma^{(0)}$ is factorized into left and right with respect to automorphisms
\begin{align}
\begin{split}
\gamma_{g,L}
=d^0_{H_{U}\cap H_L} \Xi^{(0)}_-(g)\circ\sigma_-^{(1)}(g)\\
\gamma_{g,R}
=d^0_{H_{U}\cap H_R} \Xi^{(0)}_+(g)\circ\sigma_+^{(1)}(g).
\end{split}
\end{align}
For this $\gamma_{g,R}$, by a straight forward calculation, we get
\begin{align}
\begin{split}
&d^1_{H_U\cap H_R}\gamma_{R}(g,h)=
d^1_{H_U\cap H_R}d^0_{H_{U}\cap H_R} \Xi^{(0)}_+(g,h)\circ
d^1_{H_U}\sigma_+^{(1)}(g,h)
=\Ad\lmk V\lmk g,h\rmk\rmk.
\end{split}
\end{align}
Hence 
we may set $v(g,h)=V(g,h)$ in Theorem \ref{autothm}.
Then the corresponding $3$-cocycle 
$c_R$ of Theorem \ref{autothm} 
is $c(g,h,k)=\Psi_3(\nu)(g,h,k)$
from (iv) of Lemma \ref{vml}.
This proves the theorem.
\end{proof}

\section*{Acknowledgment.}
The author is grateful to all the 2020 Current
Development of Mathematics Conference organizers of, for giving her the opportunity to give a talk there
and write this lecture note.
The author is thankful to Hal Tasaki for his help on this lecture note, particularly
for authorizing her to write down his calculation on
Dijkgraaf-Witten model in this lecture note.
This work is supported by JSPS KAKENHI Grant Number 16K05171 and 19K03534
and JST CREST Grant Number JPMJCR19T2.

\end{document}